\documentclass{LMCS}

\usepackage{enumerate}
\usepackage{hyperref}
\usepackage{color}
\usepackage{xspace}
\usepackage{amssymb}
\usepackage{amsfonts}
\usepackage{amsmath}
\usepackage{txfonts} 
\usepackage{url}
\usepackage[cmtip]{xypic}
\usepackage{graphicx}
\usepackage{semantic}

\def\enbox#1{\enskip\mbox{#1}\enskip}

\usepackage{paralist} 
\usepackage{proof} 

\usepackage{pic}
\usepackage{types}
\usepackage{logic}
\usepackage{shorthand}
\usepackage{symbols}
\usepackage[small]{display-figure}

\renewcommand{\evaluate}{\mathclose{\Downarrow}}
\renewcommand{\paralS}{\mathbin{\|}}
\renewcommand{\permin}[1]{\ensuremath{\mathopen{\modin}#1}\xspace}
\renewcommand{\permout}[1]{\ensuremath{\mathopen{\modout}#1}\xspace}

\renewcommand{\esep}{\mathbin{\bot}}

\renewcommand{\envmap}[2]{\ensuremath{#1\colon#2}\xspace}

\renewcommand{\tproc}[3]{\ensuremath{#1\vdash#2\colon#3}\xspace}

\renewcommand{\sepproc}[3]{{\ensuremath{ #2 \esep_{#1} #3}}}

\renewcommand{\figref}[1]{Figure~\ref{#1}}
\renewcommand{\secref}[1]{Section~\ref{#1}}
\renewcommand{\defref}[1]{Definition~\ref{#1}}
\renewcommand{\exref}[1]{Example~\ref{#1}}
\renewcommand{\thmref}[1]{Theorem~\ref{#1}}

\newenvironment{theorembis}[2]{\begin{trivlist}
\item[\hskip \labelsep {\bfseries Theorem #1}] {(#2).}
\it}{\end{trivlist}}

\newenvironment{lemmabis}[2]{\begin{trivlist}
\item[\hskip \labelsep {\bfseries Lemma #1}] {(#2).}
\it}{\end{trivlist}}

\newenvironment{propbis}[2]{\begin{trivlist}
\item[\hskip \labelsep {\bfseries Proposition #1}] {(#2).}
\it}{\end{trivlist}}

\def\doi{7 (3:07) 2011}
\lmcsheading%
{\doi}
{1--47}
{}
{}
{Sep.~10, 2010}
{Sep.~\phantom01, 2011}
{}   

\begin{document}

\title[Permission-based Separation Logic for Message-Passing Concurrency]{Permission-based Separation Logic for Message-Passing Concurrency}

\author[A.~Francalanza]{Adrian Francalanza\rsuper a}	
\address{{\lsuper a}ICT, University of Malta}	
\email{adrian.francalanza@um.edu.mt}  

\author[J.~Rathke]{Julian Rathke\rsuper b}	
\address{{\lsuper{b,c}}ECS, University of Southampton, UK}	
\email{jr2@ecs.soton.ac.uk, vs@ecs.soton.ac.uk}  

\author[V.~Sassone]{Vladimiro Sassone\rsuper c}	
\address{\vskip -6 pt}	





\begin{abstract}
  \noindent  We develop local reasoning techniques for message passing concurrent programs based on ideas from separation logics and resource usage analysis.
We extend processes with permission-resources and define a reduction semantics 
for this extended language. This provides a foundation for interpreting separation formulas for message-passing concurrency. We also define a sound proof system 
permitting us to infer satisfaction compositionally using local, separation-based reasoning.
\end{abstract}

\maketitle


\section{Introduction}
\label{sec:introduction}

Reasoning about concurrent programs is widely acknowledged to be a
difficult business due to the intricate interferences between 
threads scheduled non-deterministically and to the intrinsic difficulty 
of scaling reasoning techniques to account for these.  
The use of \emph{local} reasoning techniques in the guise of separation logic 
\cite{Reynolds02,Ohearn04} represents a promising advance for this area. 
Here, the state of resources acted upon by threads are reasoned about
independently, where possible.  This approach has spawned numerous  papers \cite{Brookes07,
  permission,Vafeiadis07amarriage,Calcagno07modularsafety,Feng, ParkinsonBornat,
 Calcagno07localaction} targetting the 
shared-variable concurrency model. 

An alternative, albeit slightly higher-level, model of concurrency is
that of message-passing whereby the only shared resources
allowed are the message-passing channels themselves.  Access to these
shared resources is  controlled by the message-passing
programming interface and so interfering behaviour is more explicit 
and therefore can be tracked more readily. This paradigm has been extensively studied using
process calculi \cite{Hoare78, Milner89,
  Milner99,SangiorgiWalker01} but
has also been efficiently implemented  and deployed in more programming
 oriented settings \cite{Galletly97,Reppy99,Armstrong07}.

\smallskip

In this paper we develop a local reasoning proof system for
message-passing concurrent programs, based on ideas from
both concurrent separation logics~\cite{Ohearn04} and permission-based resource 
analyses~\cite{Boyland03,permission}.  Our initial step towards the broader and ambitious goal of local reasoning 
for message-passing systems focusses on the study of confluent value-passing programs, a 
class large enough to present a significant 
theoretical challenge while still being of considerable practical interest.

Our approach to using processes as a model for separation-based Hoare-style reasoning centers around the conceptual partitioning of message-passing programs  
into `program state',  \ie  the values emitted on asynchronous outputs,  and `program code', \ie  the remaining parallel processes acting on this state.    
For instance, one way  to view the program
\begin{align}\label{eq:intro1}
  & \pioutA{c_1}{4} \quad\paralS\quad \pioutA{c_2}{2}\quad \paralS\quad \piin{c_1}{x}{\piin{c_2}{y}{\,\cmpCB{x = y}{\pioutA{c_1}{(x, y, x\!+\!x)}\paral \pioutA{d}{\,}}{\pioutA{c_2}{(x, y, x\!+\!y)}\paral \pioutA{d}{\,} }}}
\end{align}
would be to consider the  asynchronous
outputs
\begin{align} \label{eq:intro-state-proc}
&\pioutA{c_1}{4} \paralS \pioutA{c_2}{2}  
\end{align}
 as the `state', holding values $4$ and $2$  at `addresses' $c_1$ and $c_2$ and the process
 \begin{align}\label{eq:intro-st-tran-proc}
 &   \piin{c_1}{x}{\piin{c_2}{y}{\,\cmpCB{x =
         y}{\pioutA{c_1}{(x,x\!+\!x)}\paral \pioutA{d}{\,}}{\pioutA{c_2}{(x,y,x\!+\!y)} \paral \pioutA{d}{\,}}}}
 \end{align}
as the `code', or 
state transformer, consuming the values on channels $c_1$ and $c_2$ and producing a new state holding the previous values consumed from  $c_1$ and $c_2$ together with their summation  on either of the \emph{previously used} channels $c_1$ and $c_2$, depending on whether these values were equal or not, and signals on channel $d$.
Using such an analogy, we can decompose our analysis  and reason about sub-programs independently. We can
 interpret assertions over  processes such as 
\begin{align}
  &\label{eq:intro-state-formula}
  \fcons{\fstate{c_1}{4}}{\fstate{c_2}{2}} 
\end{align}
 This  assertion, 
  a conjunction, describes a process reducing to a `soup' of two asynchronous outputs on channels $c_1$ and $c_2$, holding values $4$ and $2$, respectively; the process in \eqref{eq:intro-state-proc} would satisfy this assertion.      This state-based process view also permits an intuitive formulation of Hoare-style  sequents of the form 
\begin{align}
  & \label{eq:intro-conj-formula}
  \seqNEB{ \fcons{\fstate{c_1}{4}}{\fstate{c_2}{2}}}{\_\_\_\_}{\,\fcons{\fstate{c_2}{4, 2, 6}}{\fstate{d}{}}} 
\end{align}
Such a sequent   describes a process that, once composed with the state described by the precondition
 $ \fcons{\fstate{c_1}{4}}{\fstate{c_2}{2}}$,   reduces to some other stable state  described by the postcondition \fcons{\fstate{c_2}{4,2,6}}{\fstate{d}{}}, with  values $4,2,6$  on channel $c_2$ and an empty tuple on channel $d$ acting as a signal, indicating that the data on channel $c_2$ can now be accessed; the process in \eqref{eq:intro-st-tran-proc} would satisfy this sequent. 
  In compositional fashion,  we can then determine that the entire program of (\ref{eq:intro1}) reduces to a stable state satisfying \fcons{\fstate{c_2}{4,2,6}}{\fstate{d}{}}  from  separate analyses relating to the two sub-programs.

This state-based logical view of processes  lends itself well to the specification of deterministic computation whose operation can be decomposed into asynchronous parallel subcomponents.  Application examples range from   parallel processing of data, \cite{Knuth98}, to distributed agreement problems, \cite{Lynch96}.   State-based specifications would allow a more natural expression of the expected behaviour of these algorithms because they are agnostic \wrt the specific temporal order of the generation and consumption of this state.  For instance, as opposed to temporal logics such as \cite{Stirling01},   the formula \eqref{eq:intro-state-formula} does not specify whether the sub-state \fstate{c_1}{4} is to be produced before \fstate{c_2}{2} or vice-versa. Dually, sequents such as \eqref{eq:intro-conj-formula} do not necessarily specify if and how this data on channels is to be consumed.   The temporal agnosticism in `spatial' specifications is also more amenable to intuitive decompositions and composition of analysis; we can verify that a process $P$ satisfies the formula \eqref{eq:intro-state-formula} from sub-processes making up $P$ that satisfy \fstate{c_1}{4} and \fstate{c_2}{2}.   

The state-based logical process view is also appealing because the specifications of the algorithms we are considering are also, in some sense, more data-centric rather than control-centric and focus more on the relationships between data at the beginning and the end of computation. One can in fact view the sequent in \eqref{eq:intro-conj-formula} as a description on how the data on channels $c_1$ and $c_2$ in the precondition changes to the data on $c_2$ in the postcondition; the dependencies between such data will be made more explicit later on once we introduce value variables.  Finally, data-centric applications such as in-place sorting also tend to reuse data-placeholders during computation, possibly at different types and formats \eg the code in \eqref{eq:intro-st-tran-proc}, in order to minimise resource usage.  Correctness specifications such as the sequent in \eqref{eq:intro-conj-formula} handle this aspect rather naturally as opposed to traditional correctness analysis for message passing programs, such as  type systems in \cite{BergerHondaYoshida08,TerauchiAiken08}, which often limit channel usage to one form of data.

\medskip

A central assumption underlying our process interpretations is the absence of program interference and the deterministic reduction of processes.  In a message-passing paradigm, program interference is caused by \emph{races} for  values, through  multiple outputs or inputs competing for  \emph{shared} channels.   In cases such as \eqref{eq:intro1} above, where channels are reused,  rudimentary analysis based on the free names of processes \eg \cite{AmDam96} are too coarse for adequate race detection.  Moreover, these type based safety analyses \eg \cite{BergerHondaYoshida08,TerauchiAiken08} tend to avoid reasoning about data, approximating control over branching as a result.

To reason about such interferences in the presence of channel reuse, we define a resource-semantics for processes, based on \emph{linear} input and output permissions.   Every process is embellished with a set of permissions,  \conCP, denoting that  process $P$ `owns' the permissions in set \eV (\cf\emph{ownership hypothesis},\cite{Ohearn04}).  The resource-semantics limits communication to the permissions owned by a process.  Thus, for example, for the following reduction  to occur
\begin{align}\label{eq:perm-trans}
  & \conC{\,\pioutA{c_1}{4}\,} \paralS \con{\,\eVV}{\piin{c_1}{x}{P}\,} \quad\reduc\quad \con{\eV\cup\eVV}{\,P\subC{4}{x}\,}
\end{align}
the output process, \pioutA{c_1}{4}, (\resp the input process, \piin{c_1}{x}{P}), 
must have the permission to output  (\resp to input) on channel $c_1$ in its permission-set $\eV$ (\resp \eVV). 
  Since permissions are  not part of the original process semantics (they are only added in the resource-semantics to aid reasoning) the above enriched reduction also describes the \emph{implicit} transfer of permissions \eV from the output process, \pioutA{c_1}{4}, to the input process, \piin{c_1}{x}{P}, \ie adding \eV to the already owned permissions \eVV, as a result of their synchronisation (\cf \emph{ownership transfer} \cite{Ohearn04}).   

  \begin{align}
    \label{eq:61}
    & \seqNEBS{ \fcons{\fstate{c_1}{4}}{\fstate{c_2}{2}}}{\con{\sset{\permin{c_1},\permin{c_2},\permout{d}}}{
        \begin{array}{l}
        \piin{c_1}{x}{\piin{c_2}{y}{\,\cmpCA{x =y}{\pioutA{c_1}{(x,x\!+\!x)}\paral \pioutA{d}{\,}}\\
         \qquad\qquad\qquad\qquad\elseA{\pioutA{c_2}{(x,y,x\!+\!y)} \paral \pioutA{d}{\,}}}}
   \end{array} }}{\,\fcons{\fstate{c_2}{4, 2, 6}}{\fstate{d}{}}}
  \end{align}

The earlier sequent \eqref{eq:intro-conj-formula} can now be stated in terms of the process of \eqref{eq:intro-st-tran-proc} confined by the permissions   $\permin{c_1},\permin{c_2}$ and $\permout{d}$, as shown in \eqref{eq:61}.  Note how channel reuse    
manifests itself through the fact that our permission-confined 
process in \eqref{eq:61} does not own the output permissions $\permout{c_1}$ and $\permout{c_2}$, even though they are clearly used in this code.  These however will be obtained from the precondition; from a permission perspective, the inputs on channels $c_1$ and $c_2$ act as guards, masking the use of the permissions   $\permout{c_1}$ and $\permout{c_2}$.

Making ownership explicit also simplifies the detection of races in the model and provides an immediate notion of process separation in terms of owned permissions.  For instance, in  \eqref{eq:perm-trans} we determine that there are no races across the two processes  \conC{\,\pioutA{c_1}{4}\,} and \con{\,\eVV}{\piin{c_1}{x}{P}\,} without having to analyse the actual structure of the respective confined processes \pioutA{c_1}{4} and \piin{c_1}{x}{P}; instead we  simply check that their permission sets are disjoint \ie $\eV \cap \eVV = \emptyset$    (\cf \emph{separation property} \cite{Ohearn04}).   This assumed disjunction of permissions will also play a major role in the semantic definition of \eqref{eq:61}, as it allows us to give a separation interpretation to our sequents.

\begin{align}
  \label{eq:62}
    & \seqNEBS{ \fcons{\fstate{c_1}{4}}{\fstate{c_2}{2}}}{
        \piin{c_1}{x}{\piinBB{c_2}{y}{\begin{array}{l}
        \,\cmpCA{x =y}{\pioutA{c_1}{(x,x\!+\!x)}\paral \pioutA{d}{\,}} \\
          \qquad\qquad
         \elseA{\pioutA{c_2}{(x,y,x\!+\!y)} \paral \pioutA{d}{\,}}
    \end{array}}}
 }{\,\fcons{\fstate{c_2}{4, 2, 6}}{\fstate{d}{}}}
  \end{align}

Another pleasing property of this embellishment is that, in the absence of races,  
this resource-semantics corresponds to the standard (permission-less) reduction semantics. Thus the permission semantics can  be used as a \emph{narrative} to support reasoning about confluent reductions of processes. This therefore means that we can abstract over the existence of such a narrative 
in our sequents and express  \eqref{eq:61} simply as the permission-less sequent in \eqref{eq:62}, thereby returning to our original aim and obtaining Hoare-triple specifications in terms of processes. 

\medskip

We  define a sound proof system for the aforementioned logic and resource-confined processes with judgements of the form:
\begin{align*}
  \seqEB{\fV}{S}{\fVV}
\end{align*}
The environment, \env, associates channels with ownership transfer invariants  of permissions, 
and  $S$ denotes a system of  processes confined by permissions. These sequents depart slightly from previous work on concurrent separation logic \cite{Ohearn04}, as value-domain assertions - assertions interpreted exclusively in terms of the domain of values communicated  and thus independent of the process  structure, $S$ -  are extracted from the pre and post-conditions, \fV
  and \fVV,    and consolidated  as a boolean expression, $\bV$.    Correctness proofs in this proof system weave together two inter-dependent mechanisms. On the one hand, they verify, \emph{in sequential fashion}, the satisfaction of the post-condition \fVV for system $S$, assuming the precondition \fV; the soundness of this sequential analysis 
  stems from the non-interference properties guaranteed by the resource semantics of $S$.  On the other hand, sequents construct  race-free systems $S$, using assumptions from the environment, \env, and the pre-condition, \fV.

We have already argued for the naturality of our specifications \wrt deterministic message-passing programs and how our analysis can handle more refined branching control analysis, even when this is \emph{data dependent} as in (\ref{eq:intro1}). Another, perhaps even more crucial  advantage of our approach over existing  analysis techniques for message-passing concurrency (\eg  \cite{Hliu93,AnStWi94,Dam02}) is \emph{locality of reasoning}.  By concentrating on deterministic code, our reasoning need not take into account the different interleaving of concurrent code executing in  context; this facilitates substantially proof \emph{compositionality} and induces a \emph{lightweight} sequential form of  analysis.  Explicit permission ownership simplifies 
interference  delineation, even in the presence of channel reuse; 
such delineation 
is a major obstacle when defining manageable compositional proof rules (\eg \cite{Dam02}). 

The paper is structured as follows. We introduce  our  language in \secref{sec:language}. In \secref{sec:conf-value-pass} we define a resource-semantics for permission-confined processes and state its  key properties.  We define our assertion logic and interpret it using a separation model over confined processes in \secref{sec:logic}.  In \secref{sec:proof-system} we present our proof system  and declare its soundness whereas in \secref{sec:application} we apply this system to prove properties about message-passing programs.  Finally, in \secref{sec:conclusion} we make  concluding remarks regarding related and future work.

\section{Language}
\label{sec:language}

Our language, an asynchronous value-passing CCS, is described 
 in \figref{fig:structural-reduc} and consists of three syntactic categories. 
 Values, $v,\, u\in\Values$, are numerals denoting integers. 
 \emph{Side-effect free} expressions, $e$, denote integer operations 
 that may contain variables $x,\,y\in\Vars$.  
 We assume an evaluation function from closed expressions to values, 
 $e\evaluate v$.   We also assume a 
 denumerable set of channel names $c,\,d\in \Names$ and process names $\pV\in\PNames$ and denote lists of values, variables and channels as $\lst{v}, \lst{x}$ and $\lst{c}$ respectively.

\begin{small_display}{Processes, Structural Equivalence and Reduction} {fig:structural-reduc}
  \begin{align*}
\\
 &\lefteqn{\textbf{Values, Expressions, Boolean Expressions and Processes}}\\
  &\begin{array}{r@{\hspace{1ex}}c@{\hspace{1ex}}l r@{\hspace{1ex}}c@{\hspace{1ex}}l r@{\hspace{1ex}}c@{\hspace{1ex}}l}
\; v,u     &  \bnfdef  &  \num{0}  \bnfsepp  \num{1}  
  {}\mid{}  \ldots  \quad\qquad 
 & e     &  \bnfdef  &  v   \bnfsepp   x   \bnfsepp   e+e   \bnfsepp   e-e  \quad\qquad
&\bV &   \bnfdef  &  e\leq e   \bnfsepp   \bnot{\bV}   \bnfsepp   \band{\bV}{\bV}\\[4pt]
\; P,Q  &  \bnfdef 	&  \multicolumn{7}{l}{  \pioutA{c}{\lst{e}} \; \bnfsepp \; \piin{c}{\lst{x}}{P}\;   \bnfsepp \; \cmpC{\bV}{P}{Q}\;  \bnfsepp\; \pCall{\pV}{\lst{e}}\subS{\lst{c}}{\lst{d}}
   \; \bnfsepp \; \inert \; \bnfsepp \;  P\paralS Q   \; \bnfsepp \;    \rest{c}{P} 
 }
\end{array}\\[10pt]
  &\lefteqn{\textbf{Structural Equivalence Rules}}\\
  &  \begin{array}{rlrl}
      \rtit{sCom} \quad&  P\paralS Q  \steqS  Q\paralS P \quad  &
  \rtit{sAss}  \quad&  P\paralS(Q\paralS R)  \steqS  (P\paralS Q)\paralS R \\[4pt]
\rtit{sNew} \quad & \rest{c}{\inert} \steqS \inert  \qquad\qquad&
 \rtit{sSwp} \quad & \rest{c}{\rest{d}{P}} \steqS \rest{d}{\rest{c}{P}} \\[4pt]
  \rtit{sNil} \quad&  P\paralS \inert  \steqS  P   &
\rtit{sExt} \quad &  P\paralS\rest{c}{Q}  \steqS  \restB{c}{P\paralS Q} \qquad  \text{if }c\not\in\fn{P}
    \end{array}\\[10pt]
  &\textbf{Reduction Rules}\\
    &\begin{array}{c}
  \inference[\rtit{rThn}]{\bV\evaluate \boolT }{\cmpC{\bV}{P}{Q}\reduc P}\qquad \qquad
    \inference[\rtit{rEls}]{\bV\evaluate \boolF }{\cmpC{\bV}{P}{Q}\reduc Q} \\[15pt]
 \inference[\rtit{rCom}]{\lst{e}\evaluate \lst{v}}{\pioutA{c}{\lst{e}} \paralS \piin{c}{\lst{x}}{P}\reduc P\subC{\lst{v}}{\lst{x}}}\qquad\qquad
    \inference[\rtit{rPrc}]{\pCall{\pV}{\lst{x}}\deftri P \qquad \lst{e}\evaluate\lst{v}}{\pCall{\pV}{\lst{e}}\subS{\lst{c}}{\lst{d}}\reduc P\subC{\lst{v}}{\lst{x}}\subC{\lst{c}}{\lst{d}}}\\[15pt]
\inference[\rtit{rRes}]{P\reduc P'}{\rest{c}{P}\reduc\rest{c}{P'}}\quad
 \inference[\rtit{rPar}]{P\reduc P'}{P\paralS Q\reduc P'\paralS Q} \quad 
\inference[\rtit{rStr}]{P\steq P'\reduc Q'\steq Q}{P\reduc Q} \\
    \end{array}\\
  \end{align*}
\end{small_display}

\subsection{Syntax}
\label{sec:syntax}
The main syntactic category is that of processes which can asynchronously send the evaluation of expressions on 
 a channel\footnote{Our language does not allow channel names to be communicated, as in the \pic \cite{Milner99,SangiorgiWalker01}.}, \pioutA{c}{\lst{e}}, receive values on a channel,  \piin{c}{\lst{x}}{P}, and branch on the evaluation of boolean expressions,  \cmpC{\bV}{P}{Q}. Processes may assume a number of parameterised (possibly recursive) process definitions,  $\pCall{\pV}{\lst{x}}\deftri P$; these  can be invoked by the call $\pCall{\pV}{\lst{e}}\subS{\lst{c}}{\lst{d}}$, instantiating the process variables $\lst{x}$ with the evaluation of $\lst{e}$ and renaming the names $\lst{d}$ to $\lst{c}$.  Finally, processes may also be inactive, \inert, execute in parallel, $P \paral Q$, and can restrict the scope of  channels to subsets of processes, \rest{c}{P}.


 \subsection{Reduction Semantics}
 \label{sec:reduction-semantics}

 The rules for the judgement $P \reduc Q$ in \figref{fig:structural-reduc} 
   describe the dynamics of \emph{closed} processes \ie processes whose
  message variables~$\lst{x}$  are all bound by 
 input constructs \piin{c}{\lst{x}}{\_}, and process names are all defined.   Closed boolean expressions, \ie boolean formulas without free variables,  have a \emph{classical} interpretation over the boolean domain $\sset{\boolT,\,\boolF}$,  characterised by the two judgements $\bV\evaluate\boolT$ and $\bV\evaluate\boolF$.   Although this is entirely standard, we explicitly stated here in \defref{fig:bool-interp} due to its central role in subsequent development (\cf \secref{sec:proof-system}).

\begin{defi}[Boolean Condition Interpretation]\label{fig:bool-interp}
  \begin{align*}
      e_1\leq e_2   &  \evaluate
  \begin{cases}
    \boolT & \text{if } e_1\evaluate v_1,\;e_2\evaluate v_2 \text{ and } v_1 \leq v_2\\
    \boolF & \text{if } e_1\evaluate v_1,\;e_2\evaluate v_2 \text{ and } v_2 < v_1
  \end{cases} \qquad\qquad\quad
 \bnot{\bV}    \evaluate
 \begin{cases}
   \boolT & \text{if } \bV\evaluate\boolF\\
  \boolF & \text{if } \bV\evaluate\boolT
 \end{cases}
   \\ 
\band{\bV_1}{\bV_2}  &  \evaluate
\begin{cases}
  \boolT  & \text{if } \bV_1\evaluate \boolT \text{ and }  \bV_2\evaluate \boolT \\
  \boolF& \text{if } (\bV_1\evaluate \boolF \text{ and }  \bV_2\evaluate \boolT) \text{ or } (\bV_1\evaluate \boolT \text{ and }  \bV_2\evaluate \boolF) 
\text{ or } (\bV_1\evaluate \boolF \text{ and }  \bV_2\evaluate \boolF)
\end{cases} 
  \end{align*}
\end{defi}

\noindent A number of shorthand conventions are used. We write \pioutA{c}{} for \pioutA{c}{\lst{e}} and \piin{c}{}{P}   for  \piin{c}{\lst{x}}{P} when 
 $|\lst{e}| =0$ (\resp $|\lst{x}| =0$). We elide arguments and renaming from process calls, \resp \pCallNA{\pV}\subS{\lst{c}}{\lst{d}} and \pCall{\pV}{\lst{e}},  whenever these are empty lists.
 We also write  $e_1=e_2$ for $\band{(e_1\leq e_2)}{(e_2\leq e_1)}$, $e_1<e_2$ for \bnot{(e_2\leq e_1)}, \btrue for $0\leq 1$, \bfalse for  $1\leq 0$, \bor{\bV_1}{\bV_2} for \bnot{(\band{\bnot{\bV_1}}{\bnot{\bV_2}})} and \bimpl{\bV_1}{\bV_2} for \bor{\bnot{\bV_1}}{\bV_2}.  Finally, we use the shorthand $\lst{e}\evaluate \lst{v}$ for the evaluation of lists of expressions $e_1\evaluate v_1 \ldots e_n\evaluate v_n$ whenever $\lst{e}=e_1\ldots e_n $ and $\lst{v}=v_1\ldots v_n $.   

Substitutions, $\sV\in\Subs$,  are total maps from variables to values,  $\Vars \map \Values$,  and are used to define the semantics of rules \rtit{rCom} and \rtit{rPrc}.  They are finitely denoted as $\subC{\lst{v}}{\lst{x}}$, meaning that every $x_i\in\lst{x}$ is mapped to its respective $v_i \in \lst{v}$, while abstracting over all the other variable mappings in the substitution.   In the case of \rtit{rPrc} only, we abuse this notation to express the \emph{renaming}  of \lst{d} to \lst{c}.  In \secref{sec:proof-system} we abuse again this notation to describe substitutions from variables to expressions, $\subC{\lst{e}}{\lst{x}}$.  Our semantics assumes the following property of expression evaluations, which will be useful later in \secref{sec:proof-system}.

\begin{asm}\label{asm:expression-eval}
  \begin{math}
    e_1\subC{\lst{v}}{\lst{x}} \evaluate v_1 \text{ and } \lst{e}\evaluate\lst{v} \text{ implies } e_1\subC{\lst{e}}{\lst{x}} \evaluate v_1
  \end{math}
\end{asm}

A brief note on some conventions used. To improve readability we have attempted to minimise the use of universal and existential quantifiers in our statements.  Thus, unless explicitly stated, free variables introduced to the left of an implication are to be understood as universally quantified, whereas free variables introduced to the right  of an implications are  understood as
existentially quantified.

As is standard in process calculi presentations \cite{Milner99,SangiorgiWalker01}, the definition of the reduction semantics is kept compact through the rule \rtit{rStr} and  the use of process structural equivalence rules, $P\steq Q$, defined also in \figref{fig:structural-reduc}. Later on, this structural equivalence will play a role in abstracting away from the precise structure of processes when describing the satisfaction of our logic (\cf \secref{sec:logic}).

 \subsection{Process Determinism}
 \label{sec:process-determinism}

 The reduction semantics of \figref{fig:structural-reduc}  induces the following definitions relating to  stability,
 evaluation and determinism, where $\reduc^\ast$ denotes the reflexive transitive closure of $\reduc$.

\begin{defi}[Stability]\label{def:proc-stability} \rm 
  \quad\begin{math}
    P \reducNot \;\deftxt\;\; \not\exists Q.\  P \reduc Q
  \end{math}
\end{defi}

\begin{defi}[Evaluation]\label{def:proc-evaluation} \rm 
  \quad
   \begin{math}
    P \evaluate Q \;\deftxt\;\;  \exists Q'.  \;  P \reduc^\ast Q' \text{ and } Q' \reducNot \text{ and } Q' \steq Q
  \end{math}
\end{defi}

Our definition of process determinism, \defref{def:proc-determinism},  differs from that given in \cite{Milner89} in that it requires convergence, $P\converge$ \cf \defref{def:convergece}.   
We also define divergence, as the dual of convergence in standard fashion, in order to describe the existence of an infinite reduction path.  Defining determinism in terms of convergence  carries other advantages apart from the obvious relevance of termination in resource-aware settings of computation; it arguably allows for a more intuitive definition of determinism in terms of the comparison of the stable processes evaluated to (the second clause in \defref{def:proc-determinism}).  Moreover, it  fits well with our running theme of a state-based view of processes.  

\begin{defi}[Convergence and Divergence]\label{def:convergece} $\converge$ is the  least predicate over processes satisfying the conditions: 
  \begin{align*}
    P \converge & =  P \reducNot \text{ or } (\forall Q.\ P \reduc Q \text{ implies } Q\converge )
  \end{align*}
Divergence, $P \diverge$, denotes the inverse, $P \not\converge$.
\end{defi}

\newpage

\begin{defi}[Determinism]\label{def:proc-determinism} \rm P is deterministic iff:
  \begin{enumerate}[(1)]
  \item
    \begin{math}
      P \converge
    \end{math}

  \item   \begin{math}
      P \evaluate {Q_1} \text{ and } P \evaluate Q_2 \text{ implies } {Q_1\steq Q_2}
  \end{math}
  \end{enumerate}
\end{defi}





Concurrent code is notoriously hard to analyse.   One major source of complication is the potential non-deterministic behaviour of this code,  which impacts  the ability to tractably define manageable compositional proof rules (\eg \cite{Dam02}) .   More precisely,  generic non-deterministic code requires one to take into account the various interleaving of concurrent code executing in its context potentially affecting its execution.   

Although message passing concurrency minimises this interference to well defined interfaces, problems persist due to races on shared channels.  Channel reuse together with the lack of an explicit account of resource usage  makes interference hard to delineate.

 \begin{exa}\label{ex:running}  The (composite) process $\ptit{Prg}$ takes two inputs $x_1,\,x_2$ on channels $c_1,\,c_2$ respectively. It discards $x_2$ and, if  $x_1$ is less than $10$,  outputs the value $x_1$ itself together with its double on $c_1$ 
   while using $c_4$ as a signal. 
   Otherwise, it uses  $c_4$ to output $x_1$ by itself .  
   \begin{align*}
      \ptit{Prg} &\deftri \restB{c_3}{ \ptit{Fltr} \paral \ptit{Dbl}}  \\  
      \ptit{Dbl}  & \deftri \piin{c_2}{x_2}{\piin{c_3}{x_4}{\pioutA{c_1}{(x_4\!+\!x_4)}}}\\
     \ptit{Fltr} & \deftri \piin{c_1}{x_1}{\cmp{x_1}{9}{\,\pioutA{c_3}{x_1}\paralS\piinBB{c_1}{x_3}{\pioutA{c_1}{(x_1,x_3)}\paralS \pioutA{c_4}{}} \,}{\,\pioutA{c_4}{x_1}}}
   \end{align*}
Internally,  \ptit{Prg} is composed of two sub-processes, \ptit{Fltr} and \ptit{Dbl}, sharing a scoped channel, $c_3$. Process \ptit{Fltr}\ filters whether $x_1$ is less than $10$ and forwards the value to process \ptit{Dbl}\ on channel $c_3$ which, in turn, \emph{reuses} channel $c_1$ to return the doubled value.  

The process $\ptit{Prg}$ trivially converges as it is stable. When placed in the context of race-free outputs such as $\pioutA{c_1}{v_1}\paral\pioutA{c_2}{v_2}$,   $\ptit{Prg}$ still converges and
 evaluates \emph{deterministically} to
\begin{align*}
\ptit{Prg}\paral\pioutA{c_1}{v_1}\paral\pioutA{c_2}{v_2}  &\quad\evaluate\quad  \pioutA{c_4}{} \paralS \pioutA{c_1}{(v_1, 2 \times v_1)}  &&\quad \text{when } v_1 \leq 9 \;\text{and;} \\
\ptit{Prg}\paral\pioutA{c_1}{v_1}\paral\pioutA{c_2}{v_2}  &\quad\evaluate\quad \pioutA{c_4}{v_1} \paralS \restB{c_3}{\piin{c_3}{x_4}{\pioutA{c_1}{(x_4\!+\!x_4)}}}  &&\quad \text{when } v_1 > 9 
\end{align*}
On the other hand, races on,  for example,
channel $c_1$ make   $\ptit{Prg}$ behave \emph{non-deterministically}.   For instance, when placed in the context of two outputs on $c_1$, such as $\pioutA{c_1}{1} \paral \pioutA{c_2}{v_2}\paral\pioutA{c_1}{3}$,   we have a race for  the processing of $\ptit{Prg}$ yeilding two possible outcomes;
\begin{align*}
  \ptit{Prg}\paralS \pioutA{c_1}{1}\paral \pioutA{c_2}{v_2}\paral\pioutA{c_1}{3} & \quad \evaluate \quad \pioutA{c_4}{} \paralS \pioutA{c_1}{(1,2)} \paralS \pioutA{c_1}{3} && \text{or};\\
  \ptit{Prg}\paralS \pioutA{c_1}{1} \paral \pioutA{c_2}{v_2} \paral\pioutA{c_1}{3} & \quad \evaluate\quad \pioutA{c_4}{} \paralS \pioutA{c_1}{(3,6)} \paralS \pioutA{c_1}{1}
\end{align*}
More subtly, $\ptit{Prg}\paralS \pioutA{c_1}{1}\paral \pioutA{c_2}{v_2}\paral\pioutA{c_1}{3}$ may also behave in unexpected ways, since we have a second race condition when channel $c_1$ is reused internally in  $\ptit{Prg}$, \ie when \ptit{Dbl} sends back its answer to \ptit{Fltr} on $c_1$,  thereby obtaining
\begin{align*}
  \ptit{Prg}\paralS \pioutA{c_1}{1}\paral \pioutA{c_2}{v_2}\paral\pioutA{c_1}{3} & \quad \evaluate \quad \pioutA{c_4}{} \paralS \pioutA{c_1}{(1,3)} \paralS \pioutA{c_1}{2} && \text{or};\\
  \ptit{Prg}\paralS \pioutA{c_1}{1} \paral \pioutA{c_2}{v_2} \paral\pioutA{c_1}{3} & \quad \evaluate\quad \pioutA{c_4}{} \paralS \pioutA{c_1}{(3,2)} \paralS \pioutA{c_1}{6}
\end{align*}
When placed in the context of two outputs on $c_1$ with values that are less than $10$ and also values that are bigger or equal to $10$, such as $\pioutA{c_1}{1} \paral \pioutA{c_2}{v_2}\paral\pioutA{c_1}{10}$, non-deterministic behaviour varies even more widely in structure.  In fact we can have:
\begin{align*}
  \ptit{Prg}\paralS \pioutA{c_1}{1} \paral \pioutA{c_2}{v_2} \paral\pioutA{c_1}{10}& \quad \evaluate \quad \pioutA{c_4}{} \paralS \pioutA{c_1}{(1,2)} \paralS \pioutA{c_1}{10} && \text{or};\\
  \ptit{Prg}\paralS \pioutA{c_1}{1} \paral \pioutA{c_2}{v_2} \paral\pioutA{c_1}{10}& \quad \evaluate\quad \pioutA{c_4}{10}  \paralS \restB{c_3}{\piin{c_3}{x_4}{\pioutA{c_1}{(x_4\!+\!x_4)}}} \paralS \pioutA{c_1}{1}
\end{align*}
Dually, when  $\ptit{Prg}$ is placed in the context of $\pioutA{c_1}{1}\paral \pioutA{c_2}{v_2} \paral\piin{c_1}{x}{\inert}$, which introduces another input competing for the output on $c_1$, we have even more non-deterministic behaviour.  We can have: 
\begin{align*}
  \ptit{Prg}\paralS \pioutA{c_1}{1} \paral \pioutA{c_2}{v_2} \paral\piin{c_1}{x}{\inert}& \quad \evaluate \quad \pioutA{c_4}{} \paralS \pioutA{c_1}{(1,2)} \paralS \piin{c_1}{x}{\inert} && \text{or};\\
  \ptit{Prg}\paralS \pioutA{c_1}{1} \paral \pioutA{c_2}{v_2} \paral\piin{c_1}{x}{\inert}& \quad \evaluate\quad \restB{c_3}{\ptit{Fltr} \paral  \piin{c_3}{x_4}{\pioutA{c_1}{(x_4\!+\!x_4)}}} && \text{or even};\\
  \ptit{Prg}\paralS \pioutA{c_1}{1}  \paral \pioutA{c_2}{v_2}\paral\piin{c_1}{x}{\inert} & \quad \evaluate\quad  \piinBB{c_1}{x_3}{\pioutA{c_1}{(1,x_3)}\paral \pioutA{c_4}{}}
\end{align*}
\end{exa}\vspace{3 pt}

\noindent In practice, a substantial body of concurrent code  is expected to behave deterministically under some form of non-interference assumptions. One example is the in-place quicksort algorithm, which can be encoded in our language as shown in \exref{ex:quicksort}.  In this example, determinism is even harder to ascertain because, apart from channel reuse, the code is also recursively defined. This gives us scope for developing refined analysis techniques  for deterministic code
which lend themselves better to compositionality.  

\newcommand{\quick}[1]{\ensuremath{\ptit{Qck}(#1)\xspace}}
\newcommand{\prtn}[1]{\ensuremath{\ptit{Prt}(#1)\xspace}}
\newcommand{\prtnlp}[1]{\ensuremath{\ptit{Trv}(#1)\xspace}}

\begin{exa}[In-Place Quicksort] \label{ex:quicksort} The process definition \quick{i,j} defines a quicksort algorithm, sorting arrays of values \emph{in-place} and signalling on channel $r$ once sorting completes. Arrays of integers $a=[v_1,\ldots,v_n]$  are represented as a set of  messages   \( \pioutA{a_1}{v_1}\paral\ldots\paral\pioutA{a_n}{v_n}  \)  on an indexed set of channels $a_1\ldots a_n$.\footnote{Since our language can branch on integer values, channel indexing can be encoded.}  When arrays are of length $1$, \quick{i,i} signals immediately on channel $r$.  Otherwise, it chooses the value at the lowest index, $\pioutA{a_i}{v_i}$, as the pivot, partitions the array, and then calls quicksort recursively on the two partitions, renaming the returning signal to a fresh channel name in each case. Once the two sub-sortings signal back, the process can signal back on $r$.  
  \begin{align*}
    \quick{i,j} & \deftri 
    \begin{cases}
      \cmpCA{i= j}{\pioutA{r}{}}\\[-15pt]
      \qquad\quad \elseA{\restB{r_3}{
          \begin{array}{l}
              \prtn{i,j}\subS{r_3}{r}\\[-15pt]     
              \paralS 
             \;\piin{r_3}{x}{\restB{r_1,r_2}{
                 \begin{array}{l}
                    \quick{i,x-1}\subS{r_1}{r} \\
                    \paralS\; \quick{x\!+\!1,j}\subS{r_2}{r}  \\
                    \paralS\; \piin{r_1}{}{\piin{r_2}{}{\pioutA{r}{}}}
                 \end{array}}}
           \end{array}
          }}
    \end{cases}
\end{align*}
At the heart of quicksort is \prtn{i,j}, which partitions an array into two sub-arrays separated by a pivot cell, \pioutA{a_p}{v_p}, and signals completion by outputting the partition index as a value, $\pioutA{r}{p}$.  After partitioning completes, the values in the first sub-array (\ie indexes less than $p$) are less than $v_p$ and the values of the second sub-array (\ie indexes greater than $p$) are bigger or equal to $v_p$. Partitioning  calls the array traversal process \prtnlp{l,h,x,p,c}, initialising the pivot value $x$ to $v_i$, the pivot index $p$ to $i$, the counter index $c$ to $i\!+\!1$ and  low and high array boundaries $l,h$ to $i$ and $j$ respectively. 
\begin{align*}
    \prtn{i,j} &\deftri \piin{a_i}{x}{\,\prtnlp{i,j,x,i,i\!+\!1}}
\end{align*}
Traversal loops through the indexes $i+1$ up to $h$, $(6)$ then $(1)$, comparing their values with the pivot value, $(2)$.  If the current value is greater or equal to $x$, in-place partitioning restores the cell and increments the counter, $(3)$. Otherwise, it increments the pivot index to $p\!+\!1$, swaps the current value with the value at the new pivot index, and proceeds to the next index, $(4)$. Since reads are destructive in value passing concurrency, swapping occurs only when the two indexes are distinct, $(5)$. Once traversal exceeds the highest index of the array, $(6)$, the pivot value at the lowest index $l$ is swapped with the value at the current pivot index $p$ and the pivot index  is returned as the return value  $\pioutA{r}{p}$, $(7)$; again swapping is avoided if these two indexes are the same.  
\begin{align*}
    \prtnlp{l,h,x,p,c} & \deftri 
    \begin{cases} \small
      \cmpCA{\overbrace{c > h}^{(6)}}{\overbrace{\cmpC{l\!=\! p} {(\pioutA{a_l}{x}\paral \pioutA{r}{p})}  {\piinBB{a_p}{y}{\pioutA{a_l}{y}\paral \pioutA{a_p}{x} \paral \pioutA{r}{p}  }}}^{(7)}  }\\ 
      \small\hspace{10pt}\elseA{\overbrace{a_c?y.}^{(1)}  {
          \left(\begin{array}{l}
              \cmpCA{ \overbrace{x \leq y }^{(2)}}{\overbrace{\pioutA{a_c}{y} \paralS\prtnlp{l,h,x,p,c\!+\!1}}^{(3)} }\\
              \quad\elseAB{\overbrace{
                  \begin{array}{l}
                    \cmpCA{\overbrace{c\! =\! p\!+\!1}^{(5)}}{(\pioutA{a_c}{y} \paral\prtnlp{l,h,x,p\!+\!1,c\!+\!1})}\\
                    \quad\elseA{\piinBB{a_{p+1}}{z}{
                        \begin{array}{l}
                          \pioutA{a_c}{z}\paral\pioutA{a_{p+1}}{y}\paral\\
                          \prtnlp{l,h,x,p\!+\!1,c\!+\!1}
                        \end{array}
                      }}
                  \end{array}}^{(4)}
              }
            \end{array}\right)
          }}      
    \end{cases}
  \end{align*}
\end{exa}\vspace{3 pt}

\noindent We note that the splitting of the array during recursive calls in \quick{i,j} in \exref{ex:quicksort} is \emph{data dependent}, based on the pivot value returned after a call to  \prtn{i,j}.   This fact complicates confluence analysis through static techniques such as type systems for resource usage (\eg \cite{BergerHondaYoshida08,TerauchiAiken08}). To be able to deal with the refined analysis required for this example, we define a resource-semantics for our processes in \secref{sec:conf-value-pass}, which does not approximate over data dependent branching. This extended semantics then serves as a model for a resource-aware separation logic for processes, given in \secref{sec:logic}.  In \secref{sec:proof-system}  we then define a compositional proof system for verifying properties in this logic.

\section{Resourcing for Processes}
\label{sec:conf-value-pass}

We define a reduction semantics for our programs by \emph{confining} their behaviour through linear permissions for channel input and output.   This confined-process semantics  
helps us to reason about deterministic behaviour  of processes
and lays the foundation for the semantics of the logic to be presented in \secref{sec:logic}.  In particular, it 
\begin{inparaenum}[\itshape \upshape(1\upshape)]
\item gives us a basis for process separation, in terms of the permissions owned by processes,
\item assists race detection, and
\item acts as a narrative as to why a process is deterministic.
\end{inparaenum}

\begin{small_display}[h]{A Permission-Confined CCS} {fig:confined-lang}
  \begin{align*} \\
   \lefteqn{\textbf{Confined Processes (Systems)}} \\
  &\begin{array}{rl}
    \; S, T, R  &  \bnfdef    \conCP \; \bnfsepp \;   S\paralS T   \; \bnfsepp \;    \rest{c}{\,S} 
  \end{array}    \\[10pt]
   \lefteqn{\textbf{Permission Violation Detection Rules} }\\
  & \begin{array}{c}
    \hspace{1cm}\inference[\rtit{eOut}]{\permout{c}\not\in\eV}{\conC{\pioutA{c}{\lst{e}}} \err} \qquad \inference[\rtit{eIn}]{\permin{c}\not\in\eV}{\conC{\piin{c}{\lst{x}}{P}} \err}  \\[20pt]
    \hspace{2cm}\inference[\rtit{ePar}]{S\err}{S\paral T \err}\qquad \inference[\rtit{eRes}]{S\err}{\rest{c}{S} \err} \qquad \inference[\rtit{eStr}]{T\steq S\err}{T \err} 
  \end{array}\\[10pt]
  \lefteqn{\textbf{Structural Equivalence Rules} }\\
  &  \begin{array}{rlrlrl}
       \rtit{scCom} \quad&  S\paralS T  \steqS  T\paralS S \quad\;  &
   \rtit{scAss}  \quad&  S\paralS(T\paralS R)  \steqS  (S\paralS T)\paralS R \\[4pt]
\rtit{scNew} \quad & \rest{c}{\con{\emptyset}{\inert}} \steqS \con{\emptyset}{\inert}  \;& 
 \rtit{scSwp} \quad & \rest{c}{\rest{d}{S}} \steqS \rest{d}{\rest{c}{S}} \\[4pt]
  \rtit{scNil} \quad&  S\paralS \con{\emptyset}{\inert}  \steqS  S   &
\rtit{scExt} \quad &  S\paralS\rest{c}{T}  \steqS  \restB{c}{S\paralS T} \qquad  \text{if }c\not\in\fn{S}
    \end{array}\\[10pt]  
 \lefteqn{\textbf{Reduction Rules}}\\
    &\begin{array}{c}
\inference[\rtit{cThn}]{\bV\evaluate \boolT }{\conC{\cmpC{\bV}{P}{Q}}\reduc \conC{P}}\qquad 
     \inference[\rtit{cEls}]{\bV\evaluate \boolF }{\conC{\cmpC{\bV}{P}{Q}}\reduc \conC{Q}}  \\[20pt]
\inference[\rtit{cCom}]{\lst{e}\evaluate \lst{v}\qquad \permout{c}\in\eV\qquad\permin{c}\in\eVV}{\con{\eV}{\pioutA{c}{\lst{e}}} \paralS \con{\eVV}{\piin{c}{\lst{x}}{P}}\reduc \con{\eVV\cup\eV}{P\subC{\lst{v}}{\lst{x}}}}\qquad
    \inference[\rtit{cPrc}]{\pCall{\pV}{\lst{x}}\deftri P \qquad \lst{e}\evaluate\lst{v}}{\conC{\pCall{\pV}{\lst{e}}\subS{\lst{c}}{\lst{d}}}\reduc \conC{P\subC{\lst{v}}{\lst{x}}\subC{\lst{c}}{\lst{d}}}} 
\\[25pt] 
\inference[\rtit{cRes}]{S\reduc S'}{\rest{a}{S}\reduc\rest{a}{S'}}\qquad
 \inference[\rtit{cPar}]{S\reduc S'}{S\paralS T\reduc S'\paralS T} \qquad 
\inference[\rtit{cStr}]{S\steq S' \quad S'\reduc T' \quad T'\steq T}{S\reduc T}\\[20pt]
\inference[\rtit{cSpl}]{ }{\con{\eV\eadd\eVV}{P\paral Q} \reduc \con{\eV}{P} \paralS \con{\eVV}{Q}} \qquad
\inference[\rtit{cLcl}]{ }{\conC{\rest{c}{P}}\reduc \rest{c}{\,\con{\eV\uplus\eset{\permin{c},\permout{c}}}{P}}}\\[20pt]
 \inference[\rtit{cTgh}]{\sset{\permin{c},\permout{c}} \cap \eV \neq \emptyset \qquad c\not\in\fn{P}}{\restB{c}{\conCP \paralS S} \reduc \con{\eV\setminus\sset{\permin{c},\permout{c}}}{P} \paralS \rest{c}{S}} \qquad 
  \inference[\rtit{cDsc}]{\eV\neq \emptyset}{\conC{\inert} \reduc \con{\emptyset}{\inert}}
    \end{array}\\
  \end{align*}
\end{small_display}

\subsection{Systems}
\label{sec:systems}

We start by defining \emph{permission sets}.  These are used as logical embellishments to readily track channel usage and detect race conditions through conflicting permission usage.

\begin{defi}[\rm\textit{Permissions}]%
\label{def:permissions-effects}\rm
 The set of \emph{permissions} is
  $\Perm\deftxt
  \sset{\modin,\modout}\times\Names$, where $\permin{c}$ 
  (resp.\ $\permout{c}$) represents the permission to input 
  (resp.\ output) on channel $c$. A \textit{permission-set}, ranged over by the variables $\eV,\eVV$, is a subset of
  permissions, $\eV\subseteq\Perm$.
\end{defi}


Permissions are linear in the sense that there is at most one output permission and one input permission per channel.  This is not to be confused with linear (\resp affine) assumptions\cite{Girard:1989}  or types \cite{KobayashiPT:linearity}, which restrict channel usage to exactly (\resp at most) once.  In our case,  permissions are not consumed once used, but are instead transferred around and reused.  Thus, instead of restricting the number of uses of a particular channel,  they ensure that, at any stage during computation, there is at most one processes that can output (\resp input) on a particular channel.

\figref{fig:confined-lang} defines the syntax and semantics of \emph{systems} of confined processes,  $S,T,R \in \Sys$, whereby processes, $P$, are confined by permission sets, \eV, and denoted as  \conCP. Systems can also be composed in parallel and their channels can also be scoped.  

Confinement 
allows us to define \emph{separation} across systems, $S \perp T$ on the basis of the (visible) permissions owned by a system, \defref{def:owned-perm}.  In what follows, we assume systems of confined processes to always be \emph{well-resourced},  meaning that all confined parallel processes are \emph{separate}, \ie there is no overlap across owned permission sets, and that permissions are \emph{linear}.  System well-resourcing, denoted $\vdash S$,  is formalised in \defref{def:wf-systems}.  It can be easily checked statically by induction on the structure of systems.

\begin{defi}[\rm \textit{(Visibly) Owned Permissions}]\label{def:owned-perm}
   \begin{align*}
    \perm{S} & \deftxt
    \begin{cases}
      \eV & \text{if}\ S=\conCP\\
      \perm{T}\cup\perm{R} & \text{if}\ S= T\paral R\\
      \perm{T}\setminus\sset{\permin{c},\permout{c}} &  \text{if}\ S=\rest{c}{T}
    \end{cases} 
\end{align*}
\end{defi}

\begin{defi}[\rm\textit{Separation}]  \label{def:separation} \quad
 \begin{math}
    S \perp T \;\deftxt\; \perm{S} \cap \perm{T} = \emptyset 
  \end{math}
\end{defi}  

\begin{defi}[\rm \textit{Well-Resourced System}]\label{def:wf-systems}  A system $S$ is well-resourced, denote as $\vdash S$, if it is included in the least set defined by the following three rules.
  \begin{displaymath}
      \begin{array}{c}
      \inference[\rtit{wPrc}]{}{\vdash \conCP} \qquad \inference[\rtit{wPar}]{\vdash S \quad \vdash T \quad S \perp T}{\vdash S \paral T} \qquad \inference[\rtit{wRes}]{\vdash S}{\vdash \rest{c}{S}}
    \end{array}
  \end{displaymath}
\end{defi}\vspace{5 pt}

\noindent Process confinement also facilitates the \emph{detection of races}, which leads to non-deterministic behaviour in the 
process semantics  of \secref{sec:language}.   
The judgement $S \err$, defined by the rules in \figref{fig:confined-lang}, describes the \emph{detection} of 
permission violations.  As we shall see later on in \secref{sec:system-determinism} and \secref{sec:process-determinism-1}, the absence of permission violations  also
implies the absence of channel communication races.

The reduction rules in \figref{fig:confined-lang} enforce proper permission usage. Rule \rtit{cCom}  imposes additional restrictions  to \rtit{rCom} of \figref{fig:structural-reduc}: the output process (\resp the input process) is required to own the permission to output, \permout{c} (\resp input, \permin{c}) on channel $c$.   Confined processes cannot arbitrarily create 
permissions but need to \emph{transfer} them to other processes at specific interaction points (\ie communication through \rtit{cCom}).    The new rules \rtit{cSpl} and \rtit{cLcl} enforce this \emph{resourcing} of  permissions:  \rtit{cSpl} requires that newly spawned processes \emph{partition} the parent permissions amongst them whereas \rtit{cLcl} ensures that scoped names generate a single pair of input-output permissions for every channel. Note that \rtit{cSpl} is inherently non-deterministic as it does not specify how the permissions are partitioned amongst the parallel processes: \cf Section~\ref{sec:discussion} for a discussion of this.  Rules \rtit{cThn}, \rtit{cEls}, \rtit{cPrc}, \rtit{cRes}, \rtit{cPar} and \rtit{cStr} in  \figref{fig:confined-lang} are analogous to those in   \figref{fig:structural-reduc}. Structural equivalence extends to systems directly with $\con{\emptyset}{\inert}$ as the parallel composition identity.

The rule \rtit{cDsc}  allows  systems to discard permissions whenever it is clear that they will not be used anymore, whereas \rtit{cTgh} is a convenient rule that allows us to tighten name scoping irrespective of permissions; together with \rtit{cStr} and \rtit{scNil} and \rtit{scNew} it allows us to discard redundant scoping of channels as computation progresses (\cf \exref{ex:running-confined} for an example on how this rule is used.)  These last two rules are not essential for  determining whether a process is deterministic but help declutter extraneous permissions.  This enables us to express eventual stable systems more succinctly which, in turn, permits simpler definitions for assertion satisfaction later on in \secref{sec:logic}.     

\medskip

\subsection{Dynamic Properties of Systems}
\label{sec:dynam-prop-syst}

Reductions preserve locality.  This means that the permissions owned by a process provide a footprint for its reductions and that any process  it reduces to will be confined to these permissions.  This property is key for compositional reasoning when ensuring that global properties, such as that of being well-resourced, are preserved.  For instance, if the system $S \paral T$ is well-resourced, then by \defref{def:wf-systems} it must be the case that the two sub systems are separate \ie $S \perp T$.  If $S \reduc S'$, locality \ie $\perm{S'} \subseteq \perm{S}$   immediately implies that $S' \perp T$ and therefore, that the global system   $S' \paral T$  is still well-resourced.   Thus reduction also preserves well-resourcing.

\begin{lem}[Locality]\quad \label{lem:locality}
    \begin{math}
     S \reduc T \quad \text{implies} \quad  \perm{T} \subseteq \perm{S}  
  \end{math}
\end{lem}

\begin{lem}[Resourcing]\quad\label{lem:resourcing}
  \begin{math}
    \vdash S \text{ and } S \reduc T \quad \text{implies} \quad \vdash T
  \end{math}
\end{lem}

\proof[\textit{(Proof for Lemma~\ref{lem:resourcing} \& Lemma~\ref{lem:locality})}]  The proof is by rule induction on $S \reduc T$.  The main cases are:
  \begin{desCription}
  \item\noindent{\hskip-12 pt\bf \rtit{cCom}:}\ $S = \conC{\pioutA{c}{\lst{e}}} \paral \conCC{\piin{c}{\lst{x}}{P}}$ $T=\con{\eV\cup\eVV}{P\subC{\lst{v}}{\lst{x}}}$ where $\lst{e}=\lst{v}$.  It is immediate that $\perm{S} = \perm{T}$.  Moreover, $\vdash T$ by \rtit{wPrc}. 
  \item\noindent{\hskip-12 pt\bf \rtit{cPar}:}\ We have $S= R_1\paral R_2$ , $T=R'_1\paral R_2$ and $R_1\reduc R'_1$.  Moreover, $\vdash S$ implies $\perm{R_1}\cap\perm{R_2}=\emptyset$, $\vdash R_1$ and $\vdash R_2$.  Also recall that $\perm{S} = \perm{R_1}\cup\perm{R_2}$ and that $\perm{T} = \perm{R'_1}\cup\perm{R_2}$.

By  $\vdash R_1$,  $R_1\reduc R'_1$ and I.H. we obtain $\vdash R'_1$ and $\perm{R'_1}\subseteq \perm{R_1}$. By, $\perm{R'_1}\subseteq \perm{R_1}$ and  $\perm{R_1}\cap\perm{R_2}=\emptyset$ we deduce $\perm{R'_1}\cap\perm{R_2}=\emptyset$ and by $\vdash R'_1$ and $\vdash R_2$ we deduce $\vdash T$.  Moreover, by $\perm{R'_1}\subseteq \perm{R_1}$ we obtain $\perm{T}\subseteq \perm{S}$.
  \item\noindent{\hskip-12 pt\bf\rtit{cSpl}:}\ $S=\con{\eV\uplus\eVV}{P\paral Q}$ and $T=\conCP\paral\conCC{Q}$.  $\eV\uplus\eVV$ implies $\perm{\conCP}\cap\perm{\conCC{Q}}=\emptyset$ and since $\vdash \conCP$ and $\vdash \conCC{Q}$ (by \rtit{wPrc}), we get $\vdash T$.  Moreover $\perm{T}=\perm{S}$.
  \item\noindent{\hskip-12 pt\bf \rtit{cDsc}:}\ $S=\conC{\inert}$ and $T=\con{\emptyset}{\inert}$. Trivially, $\vdash T$ (by \rtit{wPrc}) and $\perm{T}=\emptyset\subseteq\perm{S}$. \qed
  \end{desCription}

\medskip

Another important property of our resource semantics is that reductions do not hide prior permission violations \ie permission violations are preserved by reductions.  This allows us to ignore intermediary steps during the evaluation of a confined process (\cf \defref{def:safe-stability-eval}) and simply inspect the resulting stable system to determine whether that evaluation 
resulted in any permission violations.  In what follows, we shall refer to evaluations without permission violations as \emph{safe}.

\begin{lem}[Violation Preservation]\quad\label{lem:violation-pres}
  \begin{math}
    S \reduc^\ast T\text{ and }\, S\err \quad\text{implies}\quad  T\err
  \end{math}
 \end{lem}

\proof
  First we show \begin{math}
    S \reduc T\text{ and }\, S\err \quad\text{implies}\quad  T\err
  \end{math} by rule induction on $S \reduc T$.  The main cases are:
  \begin{desCription}
  \item\noindent{\hskip-12 pt\bf \rtit{cCom}:}\ $S = \conC{\pioutA{c}{\lst{e}}} \paral \conCC{\piin{c}{\lst{x}}{P}}$ where $\permout{c}\in\eV$ and $\permin{c}\in\eVV$.  By case analysis, if $S\err$ then either $\conC{\pioutA{c}{\lst{e}}}\err$ because $\permout{c}\not\in\eV$ (by \rtit{eOut}) or $\conCC{\piin{c}{\lst{x}}{P}}$ because $\permin{c}\not\in\eVV$ (by \rtit{eIn}); both cases lead to a contradiction.
    \item\noindent{\hskip-12 pt\bf \rtit{cPar}:}\ $S= R_1\paral R_2$ , $T=R'_1\paral R_2$ and $R_1\reduc R'_1$.  By $S= R_1\paral R_2$, \rtit{ePar}, \rtit{eStr} and \rtit{scCom} we know $S\err$ because either:
      \begin{description}
      \item[$R_1\err$]  By $R_1\reduc R'_1$ and I.H. $R'_1\err$ and by $T=R'_1\paral R_2$ and \rtit{ePar} we get $T\err$.
      \item[$R_2\err$]  By $T=R'_1\paral R_2, \rtit{ePar}, \rtit{eStr} \text{ and } \rtit{scCom}$ we obtain  $T\err$.
      \end{description}
  \end{desCription}
  The second part of the proof is by induction on the number $n$ of reductions used \ie $S \reduc^n T$. \qed

\smallskip

\subsection{System Determinism}
\label{sec:system-determinism}

The first two main results of our resource semantics establish that system evaluation is deterministic up-to the terminal  permissions owned (\cf \thmref{thm:eval-eq-upto-permissions} and \thmref{thm:sys-eval-implies-sys-conv}).   

We first lay the ground for these results by giving the following definitions.    Systems evaluation in \defref{def:safe-stability-eval}, $S \evaluate T$,  is limited to safe-stability, $T \checkmark$,  and excludes reductions to racy systems.  The operation ${| - |}$ denotes a permission-erasure function whereby  $|S|$ returns the process in $S$ stripped of all its confining permissions; it allows us to express equivalence up-to owned permissions in \thmref{thm:eval-eq-upto-permissions}.  System Convergence, \defref{def:sys-converge}, is the least set of systems  that  converge to a stable state (but not necessarily a safe one) and is used for \thmref{thm:sys-eval-implies-sys-conv}.

\begin{defi}[Safe-Stability and Evaluation]\label{def:safe-stability-eval} \rm 
  \quad\begin{align*}
     S \checkmark &\;\deftxt\;\;  S\reducNot \text{ and } S \errNot 
   \\
    S \evaluate T &\;\deftxt\;\;  \exists T'.\   S\reduc^\ast T' \text{ and } T' \checkmark \text{ and } T \steq T' 
  \end{align*}
\end{defi}

\begin{defi}[Permission Confinement Erasure]\label{def:perm-erasure}
  \begin{align*}
    |S| & \;\deftxt\quad
    \begin{cases}
      P  &\;\text{if }  S = \conCP \\
      |T| \;\paralS \;|R| &\;\text{if } S = T \paral R \\ 
      \rest{c}{\;|T|}   &\;\text{if } S = \rest{c}{T}
    \end{cases}
  \end{align*}
\end{defi}

\begin{defi}[System Convergence]\label{def:sys-converge}  \converge\ is the least predicate over systems satisfying the equation\begin{displaymath}
    S \converge\; =\; S \reducNot \text{ or } (\forall T.\ S \reduc T \text{ implies } T \converge)
  \end{displaymath}
\end{defi}\medskip

\noindent In conformance with \defref{def:proc-determinism}, by system determinism we understand that 
\begin{inparaenum}[\itshape \upshape(1\upshape)]
\item no system can evaluate to two distinct safely-stable systems, up-to owned permissions  \ie \thmref{thm:eval-eq-upto-permissions} and that
\item no system can evaluate to a safely-stable system and, at the same time, diverge along a different execution path \ie\thmref{thm:sys-eval-implies-sys-conv}.
\end{inparaenum}

\begin{thm}[Evaluation Determinism]\label{thm:eval-eq-upto-permissions}\quad
  \begin{math}  
    S \evaluate T_1 \text{ and } S \evaluate T_2 \text{ implies } |T_1| \steq |T_2|
  \end{math}
\end{thm}

\begin{thm}[System Evaluation implies System Convergence] \label{thm:sys-eval-implies-sys-conv}
  \begin{math}
    S\evaluate \text{ implies } S\converge
  \end{math}
\end{thm}

 These properties follow, at an intuitive level,  from the partial-confluence property, as stated in Lemma~\ref{lem:confluence2}.  

\begin{lem}[Partial Confluence]\label{lem:confluence2}
    \begin{math}
    S \reduc T_1 \text{ and }  S \reduc T_2 \; \text{implies  either of the following:}
   \end{math}
    \begin{enumerate}[\em(1)]
    \item  $|T_1| \steq |T_2| \text{ or;}$
    \item $\exists T_3. \, T_1 \reduc T_3 \text{ and } T_2 \reduc T_3$
    \end{enumerate}
\end{lem}

However, the full technical details of the proofs for \thmref{thm:eval-eq-upto-permissions} and \thmref{thm:sys-eval-implies-sys-conv} are more delicate; on first reading, the reader may skip them and progress to \secref{sec:process-determinism-1}. Before though, we highlight Proposition~\ref{prop:safe-stab-vs-system-structure}, which establishes sufficient and necessary conditions on the structure of safely-stable systems; these conditions  will then act as a guiding principle when formulating our logic formulas. In essence, safely stable systems consist of mismatching asynchronous outputs and input-blocked processes composed in parallel, each owning the respective output and input permissions so as not to generate an error.

\begin{prop}[Safe-Stability and System Structure]\label{prop:safe-stab-vs-system-structure}
  \begin{align*}
    S\checkmark & \quad\text{iff}\quad  S \steq\restB{\lst{d}}{ \;\paral_{i=0}^n \con{\eV_i}{\pioutA{c_i}{\lst{e_i}}} \; \paral_{j=0}^m \con{\eVV_j}{\piin{c'_j}{\lst{x_j}}{P_j}}}
  \end{align*}
  where
  \begin{iteMize}{$\bullet$}
  \item $\sset{c_1,\ldots, c_n} \cap \sset{c'_1,\ldots,c'_m} = \emptyset$
  \item $\bigwedge_{i=0}^n \permout{c_i} \in \eV_i$
  \item $\bigwedge_{j=0}^m \permin{c_j} \in \eVV_j$
  \end{iteMize}
  and where $\;\paral_{i=0}^0 \con{\eV_i}{\pioutA{c_i}{\lst{e_i}}}$ and  $\;\paral_{j=0}^0 \con{\eVV_j}{\piin{c'_j}{\lst{x_j}}{P_j}}$ denote = $\con{\emptyset}{\inert}$.
\end{prop}

\bigskip

The proofs for \thmref{thm:eval-eq-upto-permissions} and \thmref{thm:sys-eval-implies-sys-conv} require us to work at a tighter relation than process structural equivalence for the intermediary steps of an evaluation,  namely $\quaseq$ defined in \defref{def:quaseq}, because process structural equivalence, $\steq$, loses information \wrt the currently owned permissions of a system.  The relation   $\quaseq$ lies between system structural equivalence and the respective process structural equivalence  after confinement erasure (\cf Proposition~\ref{lem:quas-impl-unconst-steq}.)

\begin{defi}[Equivalence up-to owned permissions]\label{def:quaseq} $S\quaseq T$ is defined as the least relation satisfying the following rules:
\begin{align*}
  & \inference{}{\conCP\quaseq\conCC{P}}\quad \inference{S_1 \quaseq S_2 \quad T_1 \quaseq T_2 }{S_1\paral T_1 \quaseq S_2\paral T_2}\quad \inference{S_1 \quaseq S_2 }{\rest{c}{S_1} \quaseq \rest{c}{S_2}}\quad \inference{S_1\steq S_2\quaseq T_2\steq T_1}{S_1\quaseq T_1}
\end{align*}
  \end{defi}
\smallskip

\begin{prop}
  \label{lem:quas-impl-unconst-steq}
  \begin{math}
    S\steq T \quad\text{implies}\quad S\quaseq T \quad\text{implies}\quad |S|\steq |T|
  \end{math}
\end{prop}

Note that $|S|\steq |T|$ does not imply $S\quaseq T$.  For instance, $| \con{\eV\uplus\eVV}{P \paral Q}| \steq |\conCP \paral \conCC{Q}|$ but $\con{\eV\uplus\eVV}{P \paral Q} \not\quaseq \conCP \paral \conCC{Q}$.

\begin{lem}[Properties of \quaseq\ with respect to reductions]\label{lem:quaseq-preserve-absence-err}\quad
  \begin{enumerate}[\em(1)]
  \item   
   \begin{math}
    S\quaseq T\text{ and } T\reduc T'\text{ and }\text{S\errNot}\quad\text{implies}\quad 
    S\reduc S'\text{ and } S'\quaseq T'
  \end{math}
 \item 
     \begin{math}
    S\quaseq T\text{ and } \text{S\checkmark} \quad\text{implies}\quad T\reducNot
  \end{math}
  \end{enumerate}
\end{lem}
\proof See Appendix~\ref{sec:conf-proc} \qed

The system relation $\quaseq$ allows us to specify a tighter relationship which characterises more precisely Partial Confluence, \ie Lemma~\ref{lem:confluence}.  This is then used to prove Lemma~\ref{lem:quaseq-impl-quaseq-eval}, upon which Theorem~\ref{thm:eval-eq-upto-permissions} rests.  We here relegate the proofs of  Lemmas used by   Lemma~\ref{lem:quaseq-impl-quaseq-eval} to Appendix~\ref{sec:conf-proc}.  Note also that Lemma~\ref{lem:confluence2}, stated earlier to give an intuition for how linear permissions ensure confluence, follows immediately from Lemma~\ref{lem:confluence} and Proposition~\ref{lem:quas-impl-unconst-steq}.   

\begin{lem}[Partial Confluence]\label{lem:confluence}
    \begin{math}
    S \reduc T_1 \text{ and }  S \reduc T_2 \; \text{implies  either of the following:}
   \end{math}
    \begin{enumerate}[\em(1)]
    \item  $T_1\quaseq T_2 \text{ or;}$
    \item $\exists T_3. \, T_1 \reduc T_3 \text{ and } T_2 \reduc T_3$
    \end{enumerate}
\end{lem}
\proof See Appendix~\ref{sec:conf-proc} \qed

\begin{defi}[System Evaluation Predicates]
  \begin{math}
    S \evaluate \quad \deftxt \quad \exists T. S\evaluate T 
  \end{math}
\end{defi}

\begin{lem}[Evaluation Preservation for \quaseq]\label{lem:eval-preserve-quaseq}\quad
  \begin{displaymath}
    S\quaseq T \text{ and } S\evaluate \text{ and } T\reduc T'\quad\text{implies}
    \quad
    S\reduc S' \text{ where } S'\quaseq T'\text{ and } S'\evaluate
  \end{displaymath}
\end{lem}
\proof See Appendix~\ref{sec:conf-proc} \qed


\begin{lem}[Evaluation and \quaseq]\label{lem:quaseq-impl-quaseq-eval} 
  \begin{displaymath}
    S\quaseq T \text{ and } S\reduc^{n}S'\checkmark \text{ and } T\reduc^{m}T'\checkmark \quad\text{implies}\quad S'\quaseq T'\text{ and }n=m
  \end{displaymath}
\end{lem}

\proof
By (strong) induction on the number of reductions leading to a safely-stable system from any system  $ S\reduc^{n}S'$.
\begin{desCription}
\item\noindent{\hskip-12 pt\bf $n=0$\,:}\ \ By $S\reducNot$ and Lemma~\ref{lem:quaseq-preserve-absence-err}(2) we know $T\reducNot$ which implies $m=0$ and $T'=T\quaseq S$.\smallskip
\item\noindent{\hskip-12 pt\bf $n=k+1$\,:}\ \ We have 
  \begin{align}
    \label{eq:23}
    \exists S'' \text{ such that }&S\reduc S'' \reduc^k S'
  \end{align}
  Lemma~\ref{lem:violation-pres} and $S'\checkmark$, $T'\checkmark$ implies 
  \begin{align*}
    &S\errNot \text{ and }T\errNot
  \end{align*}
and $S\reduc S''$ and Lemma~\ref{lem:quaseq-preserve-absence-err}(1) implies that $m > 0$ \ie
  \begin{align}
    \label{eq:28}
    \exists T'' \text{ such that }&T\reduc T''\reduc^{m-1}T'
  \end{align}
  Moreover, $S\reduc^{n}S'\checkmark$ and $T\reduc^{m}T'\checkmark$ imply $S\evaluate S',\; T\evaluate T'$ respectively, and by $S\quaseq T$, $S\reduc S''$  and Lemma~\ref{lem:eval-preserve-quaseq} 
  we obtain
  \begin{align}
    \label{eq:25}
    \exists T_1,\, T'_1,\, l\text{ such that }& T\reduc T_1 \\
    \label{eq:26}
    &T_1\quaseq S'' \\
    \label{eq:27}
    &T_1\reduc^lT'_1\checkmark
  \end{align}
   By $S'' \reduc^k S'$ from \eqref{eq:23},~\eqref{eq:26},~\eqref{eq:27} and I.H. we obtain
   \begin{align}
     \label{eq:24}
     S'\quaseq T'_1\text{ and } l=k 
   \end{align} 
    \ie  $T_1\reduc^kT'_1\checkmark$. Now by Lemma~\ref{lem:confluence}, \eqref{eq:25} and $T\reduc T''$ from \eqref{eq:28} we have two sub-cases:\medskip
   \begin{description}
   \item[$T_1\quaseq T''$]  By \eqref{eq:27} and \eqref{eq:24} we know $T_1\reduc^kT'_1\checkmark$ and by, $T''\reduc^{m-1}T'$ from \eqref{eq:28} I.H. we deduce
     \begin{align*}
         & T'\quaseq T'_1 \text{ and } (m-1) = k
     \end{align*}
     and by transitivity and \eqref{eq:24} we conclude $T'\quaseq S'$ and $m=(k+1)=n$ as required.\medskip
   \item[$\exists T_3.\;T_1\reduc T_3\text{ and }T''\reduc T_3 $] We
     here have two further sub-cases:\medskip 
     \begin{desCription}
     \item[$\exists T'_3,\, h\text{ such that }T_3\reduc^h T'_3\checkmark\,$:]  This implies $T_1\reduc^{h+1} T'_3\checkmark$ and by \eqref{eq:23},~\eqref{eq:26} and I.H. we obtain
       \begin{align}\label{eq:29}
         &T'_3\quaseq S' \text{ and }(h+1) = k
       \end{align}
       We also know that $T''\reduc^{h+1} T'_3\checkmark$ and by \eqref{eq:29} we obtain $T''\reduc^k T'_3\checkmark$ and,  since $T''\quaseq T''$ (reflexivity of \quaseq), using \eqref{eq:28} and I.H. we obtain
       \begin{align*}
         & T'_3\quaseq T'\text{ and } (m-1) = k
       \end{align*}
       which, first implies $m=(k+1)=n$ and then, by \eqref{eq:29}, implies $T'\quaseq S'$ as required.\medskip
     \item[$T_3\not\evaluate\,$:]   By $T_1\reduc T_3$, $T_1\quaseq T_1$ (reflexivity of \quaseq), \eqref{eq:27}  and Lemma~\ref{lem:eval-preserve-quaseq} we know
       \begin{align}
         \label{eq:30}
         \exists T_4,\, T'_4,\,i\text{ such that }&T_1\reduc T_4 \\
         \label{eq:31}
         &T_4\quaseq T_3 \\
         \label{eq:32}
         &T_4\reduc^i T'_4\checkmark
       \end{align}
       Similarly, by $T''\reduc T_3$, $T''\quaseq T''$, \eqref{eq:28}, $T'\checkmark$ and Lemma~\ref{lem:eval-preserve-quaseq}
       \begin{align}
         \label{eq:33}
         \exists T_5,\, T'_5,\,j\text{ such that }&T''\reduc T_5 \\
         \label{eq:34}
         &T_5\quaseq T_3 \\
         \label{eq:35}
         &T_5\reduc^j T'_5\checkmark
       \end{align}
        Now \eqref{eq:30} and \eqref{eq:32} imply $T_1\reduc^{i+1} T'_4\checkmark$ and by \eqref{eq:24}, $T_1\quaseq T_1$ and I.H. we obtain
        \begin{align}
          \label{eq:36}
          & T'_4 \quaseq T'_1 \quaseq S' \text{ and } (i+1)=k \text{ \ie } T_4\reduc^{k-1}T'_4\checkmark
        \end{align}
        Moreover, \eqref{eq:31},\eqref{eq:34} and transitivity imply $T_4\quaseq T_5$, and by \eqref{eq:36},~\eqref{eq:35} and I.H. we obtain
        \begin{align}
          \label{eq:37}
          &T'_5\quaseq T'_4 \quaseq S' \text{ and } j = (k-1)
        \end{align}
        By \eqref{eq:33} and  \eqref{eq:37} we obtain $T''\reduc^{k}T'_5$ and by $T''\quaseq T''$, \eqref{eq:28} and I.H. we obtain
        \begin{align*}
          &T'\quaseq T'_5 \quaseq S' \text{ and } (m-1) = k
        \end{align*}
        which also implies $m=(k+1)=n$ as required. \qed
     \end{desCription}
   \end{description}
\end{desCription}

\smallskip

\begin{theorembis}{\ref{thm:eval-eq-upto-permissions}}{Evaluation Determinism}\quad
  \begin{math}  
    S \evaluate T_1 \text{ and } S \evaluate T_2 \text{ implies } |T_1| \steq |T_2|
  \end{math}
\end{theorembis}

\begin{proof}
  By reflexivity we know $S\quaseq S$ and by Lemma~\ref{lem:quaseq-impl-quaseq-eval} we know $T_1\quaseq T_2$ which, by Proposition~\ref{lem:quas-impl-unconst-steq}, implies $ |T_1| \steq |T_2|$. 
\end{proof}

\medskip

Convergence for systems, \thmref{thm:sys-eval-implies-sys-conv},  largely follows from Lemma~\ref{lem:eval-preserve-quaseq} and Lemma~\ref{lem:quaseq-impl-quaseq-eval}.    We prove \thmref{thm:sys-eval-implies-sys-conv} by generalising the hypothesis to systems related by $\quaseq$ in Lemma~\ref{lem:sys-eval-implies-sys-conv}, so as to make the induction go through.

\begin{lem}\label{lem:sys-eval-implies-sys-conv}
  \begin{math}
    S \evaluate \text{ and } S\quaseq T \text{ implies }  T\converge.     
  \end{math}
\end{lem}

\proof By induction on $n$ where $S\reduc^n R \checkmark$ for some witness safely-stable $R$ justifying $S \evaluate$.
\begin{desCription}
\item\noindent{\hskip-12 pt\bf $n=0\,$:}\ \ This means that $S\checkmark$ and thus by Lemma~\ref{lem:quaseq-preserve-absence-err}(2) we have $T\reducNot$ which implies $T\converge$.
\item\noindent{\hskip-12 pt\bf $n=k+1\,$:}\ \ We have
  \begin{align}\label{eq:55}
   & S\reduc S'\reduc^k R\checkmark
  \end{align}
 We have two sub-cases.  If $T\reducNot$ then this trivially implies convergence.  Otherwise, if $T\reduc T'$, by Lemma~\ref{lem:eval-preserve-quaseq} we obtain
  \begin{align}\label{eq:56}
    &S \reduc S'' \text{ such that } S''\quaseq T'  \text{ and } S''\evaluate 
  \end{align}
  $S''\evaluate$ implies that for some $m$ and $R'$, $S\reduc^m R'\checkmark$, and since $S\quaseq S$, by \eqref{eq:55} and Lemma~\ref{lem:quaseq-impl-quaseq-eval} this implies that $m=k+1$ which means that $S''\reduc^k R'$.  Thus by $S''\quaseq T'$ from \eqref{eq:56} and I.H. we obtain $T'\converge$ which implies $T\converge$. \qed
\end{desCription}

\begin{theorembis}{\ref{thm:sys-eval-implies-sys-conv}}{System Evaluation implies System Convergence} 
  \begin{math}
    S\evaluate \text{ implies } S\converge
  \end{math}
\end{theorembis}

\proof Immediate by Lemma~\ref{lem:sys-eval-implies-sys-conv} and $S\quaseq S$. \qed

\medskip

\subsection{Process Determinism}
\label{sec:process-determinism-1}

The second main batch of results relate system evaluations in our resource semantics with process determinism in the unconstrained semantics of \secref{sec:language} (\cf Corollary~\ref{cor:sys-eval-implies-proc-deter}).   In particular, \thmref{thm:eval-implies-determinism} states that  any well-resourced permission allocation $S$ that allows a process $|S|$ to evaluate down to a safely-stable system, $T$, implies that any evaluation for  process $|S|$ - in the unconstrained semantics -  corresponds, up to structural equivalence, to this system $T$ stripped of its constraining permissions \ie $|T|\steq Q$ whenever $|S| \evaluate Q$.  On the other hand, \thmref{thm:process-convergence} states that if $S$ evaluates successfully to a safely-stable process, then the corresponding process $|S|$ must be convergent.     Together, these two theorems effectively  state that finding a single allocation (narrative) $S$ of linear permissions for a process $|S|$ that allows it to evaluate to some $T$ suffices to show that $|S|$ is deterministic in the unconstrained semantics (Corollary~\ref{cor:sys-eval-implies-proc-deter}.) 

\begin{thm}[Process Evaluation Determinism] \label{thm:eval-implies-determinism}\quad
  \begin{math}
    S \evaluate T,\ |S| \evaluate Q_1,\ |S|\evaluate Q_2 \;\text{implies}\; Q_1 \steq Q_2 \steq |T| 
  \end{math}
\end{thm}

\begin{thm}[Process Convergence]\label{thm:process-convergence}
  \begin{math}
    S\evaluate \text{ implies } |S|\converge
  \end{math}
\end{thm}

\begin{cor}\label{cor:sys-eval-implies-proc-deter}
  $S \evaluate$ implies $|S|$ is deterministic.
\end{cor}

\proof  By \defref{def:proc-determinism}, \thmref{thm:eval-implies-determinism} and \thmref{thm:process-convergence}. \qed

\medskip

We next discuss in detail the proofs for \thmref{thm:eval-implies-determinism} and \thmref{thm:process-convergence}; the reader may safely skip them on first reading and proceed to \secref{sec:conf-semant-appl}. 

\thmref{thm:eval-implies-determinism}  follows directly from Lemma~\ref{thm:eval-determinism}, which in turn relies heavily on Lemma~\ref{thm:reduc-determinism}.  In essence, this lemma states that a system that evaluates to a safely stable system can match any sequence of reductions (in the unconstrained semantics) of the system stripped of its constraining permission.  This lemma is based on Lemma~\ref{lem:eval-correspondence}, which proves the property for the case of a single unconstrained reduction, and  also depends on the the property of corrective reductions, Lemma~\ref{cor:corrective-reductions}.  This lemma states that any system that can evaluate safely, $S \evaluate$, is guaranteed to be able to  ``correct''  wrong partitioning of permissions (\cf \rtit{cSpl} in \figref{fig:confined-lang}) along a particular reduction path that result in systems that can not evaluate safely.  Stated otherwise, this means that there must exist a permission partition that leads to a full evaluation along  that particular execution path.

\begin{lem}[Corrective Reductions]\label{cor:corrective-reductions}
  \begin{displaymath}
    S\evaluate \text{ and }S\reduc^n T \text{ and } T\not\evaluate\quad\text{implies}\quad\exists\ R\text{ such that }S\reduc^n R \text{ and } R\quaseq T \text{ and }R\evaluate
  \end{displaymath}
\end{lem}

\begin{proof}
  Immediate from Lemma~\ref{lem:eval-preserve-quaseq-multi-reduc} from Appendix~\ref{sec:conf-proc} and the fact that $S\quaseq S$.
\end{proof}

\begin{lem}[Reduction Correspondence]\label{lem:eval-correspondence}
  \begin{displaymath}
    S \evaluate  \text{ and } |S|\reduc Q \quad \text{implies}\quad \exists R\text{ such that } S\reduc^{+} R \text{ and } |R| \steq Q
  \end{displaymath}
\end{lem}
\proof By rule induction on $|S|\reduc Q$; see Appendix~\ref{sec:conf-proc} \qed

\begin{lem}[Multi-step Reduction Correspondence]  \label{thm:reduc-determinism}
  \begin{displaymath}
  |S| \reduc^n Q \text{ and } S \evaluate  \quad\text{implies} \quad  S\reduc^{n+m} R \text{ such that } R\evaluate \text{ and } |R|\steq Q.
\end{displaymath}
\end{lem}

\proof
  Proof by induction on the number of reduction steps that lead to a stable process $|S|\reduc^n Q$:
  \begin{desCription}
  \item\noindent{\hskip-12 pt\bf $n=0\,$:}\ \ Immediate since $Q=|S|$ and $S\reduc^0 S$ where $S\evaluate$.
  \item\noindent{\hskip-12 pt\bf $n=k+1\,$:}\ \ This means that $\exists P$ such that $|S|\reduc P\reduc^k Q$.  By $S\evaluate $ and Lemma~\ref{lem:eval-correspondence} we know:
    \begin{align}
      \label{eq:22}
      & \exists T\text{ such that } S \reduc^{l} T, l>0 \text{ and } |T| \steq P
    \end{align}
   Thus by $P\reduc^k Q$, $|T| \steq P$ from \eqref{eq:22} and \rtit{rStr} we have
   \begin{align}
     \label{eq:57}
     & |T|\reduc^k Q
   \end{align}
    At this point we have two cases:
    \begin{description}
    \item[$T\evaluate$] By  I.H. implies we deduce that $T\reduc^{k+m} R \text{ such that } R\evaluate \text{ and } |R|\steq Q$, and by $ S \reduc^{l} T$ from \eqref{eq:22} we obtain
      \begin{align*}
        &S\reduc^{k+m+l} R \text{ such that } R\evaluate \text{ and } |R|\steq Q.
      \end{align*}
    \item[$T\not\evaluate$]  By $S \reduc^{l} T$ from \eqref{eq:22} and Lemma~\ref{cor:corrective-reductions}, we know
      \begin{align}
        \label{eq:38}
        & \exists T'\text{ such that } S\reduc^{l} T' \text{ and } T'\quaseq T \text{ and } T'\evaluate 
      \end{align}
       Now, by Proposition~\ref{lem:quas-impl-unconst-steq},  $T'\quaseq T$ and  implies $|T'| \steq |T|$.  Thus by \eqref{eq:57} and \rtit{rStr} we deduce $|T'|\reduc^k Q$.  Thus by $T'\evaluate$ and I.H. we obtain $T'\reduc^{k+m} R \text{ such that } R\evaluate \text{ and } |R|\steq Q$, and by $S \reduc^{l} T'$ from \eqref{eq:38} we obtain
      \begin{math}
        S\reduc^{k+m+l} R \text{ such that } R\evaluate \text{ and } |R|\steq Q. 
      \end{math} \qed\smallskip
    \end{description}
  \end{desCription}

\noindent Lemma~\ref{thm:eval-determinism} uses also Lemma~\ref{lem:eval-determ-no-reduc}, which maps stable processes to safely stable systems.

\begin{lem}[Correspondence]\quad\label{thm:correspondence}
    \begin{math}
   S \reduc T \quad\text{implies}\quad |S| \reduc |T|    \;\text{or}\; |S| \steq |T|  
  \end{math}
\end{lem}

\proof
  The proof
  is by rule induction on $S \reduc T$ and we relegate this  to Appendix \ref{sec:conf-proc}.  \qed

\begin{lem}[Correspondence and Termination] \label{lem:eval-determ-no-reduc}
  \begin{math}
      |S|\reducNot \text{ and } S\evaluate T\quad\text{implies}\quad |T|\steq |S| 
  \end{math} 
\end{lem}
\proof By induction on $n$ where $S \reduc^n T$.  The inductive case uses the contrapositive of Lemma~\ref{thm:correspondence}.  See Appendix~\ref{sec:conf-proc} \qed

\begin{lem}[Evaluation Determinism]  \label{thm:eval-determinism}
  $|S| \evaluate Q \text{ and } S \evaluate T \quad\text{implies} \quad  Q \steq |T| $.
\end{lem}

\proof    $|S| \evaluate Q$  implies that
\begin{align}
  &|S| \reduc^n Q\reducNot \text{ for some $n$}\label{eq:74}
\end{align}
By $|S| \reduc^n Q$, $S \evaluate T$ and Lemma~\ref{thm:reduc-determinism}  we know that $S \reduc^{n+m} R$ such that $R\evaluate$ and $|R|\steq Q$.  Since $Q\reducNot$, \eqref{eq:74}, then by Corollary~\ref{cor:steq-reduction-unconstrained} we obtain $|R|\reducNot$ and thus, by $R\evaluate$ and Lemma~\ref{lem:eval-determ-no-reduc}  we know that
\begin{align}
  & \text{ $R\evaluate T'$ for some $T'$ where $|T'|\steq|R|$}\label{eq:75}
\end{align}
By $S \reduc^{n+m} R$ and $R\evaluate T'$ of~\eqref{eq:75} we deduce that $S\evaluate T'$ and by $S\evaluate T$ and \thmref{thm:eval-eq-upto-permissions} from Section~\ref{sec:system-determinism} we know $|T|\steq|T'|$.  Thus by transitivity we obtain $|T|\steq|T'|\steq|R|\steq Q$ as required. \qed

\begin{theorembis}{\ref{thm:eval-implies-determinism}}{Process Evaluation Determinism} \quad
  \begin{math}
    S \evaluate T,\ |S| \evaluate Q_1,\ |S|\evaluate Q_2 \;\text{implies}\; Q_1 \steq Q_2 \steq |T| 
  \end{math}
\end{theorembis}

\begin{proof}
  By Lemma~\ref{thm:eval-determinism} we know $Q_1\steq |T|$ and $Q_2\steq |T|$ and the required result follows by transitivity of \steq.
\end{proof}

\medskip

The theorem relating system evaluation and process convergence uses the following corollary, obtained from Proposition~\ref{prop:safe-stab-vs-system-structure} of Section~\ref{sec:system-determinism}.

\begin{cor}\label{cor:correspondence-1} \quad
  \begin{math}
     S \errNot \text{ and } S\reducNot \quad \text{implies} \quad |S|\reducNot
  \end{math}
\end{cor}

\proof Follows from Proposition~\ref{prop:safe-stab-vs-system-structure}. \qed

\begin{theorembis}{\ref{thm:process-convergence}}{Process Convergence}
  \begin{math}
    S\evaluate \text{ implies } |S|\converge
  \end{math}
\end{theorembis}

\proof By contradiction.  Assume that $|S|\diverge$. Since, by $S\evaluate$ and \thmref{thm:sys-eval-implies-sys-conv}, any reduction sequence starting from $S$ is finite, by $|S|\diverge$ there must exists a long enough  reduction sequence
\begin{align*}
  &|S|\reduc^n Q \reduc\ldots
\end{align*}
where, by Lemma~\ref{thm:reduc-determinism}, $S\evaluate T$  and $|T|\steq Q$.  Now since $T\checkmark$, then by Corollary~\ref{cor:correspondence-1} we must have $Q\reducNot$ which contradicts our assumption. Thus $|S|\converge$.\qed

\subsection{Confined Semantics Application}
\label{sec:conf-semant-appl}

The following examples expound on the use of  linear  permission allocations for reasoning about deterministic code.  

\begin{exa}\label{ex:running-confined} $\ptit{Prg}\;\paralS  \pioutA{c_1}{v_1}\paralS\pioutA{c_2}{v_2}$ can be shown to be deterministic by finding a permission assignment for every process below that permits a safe evaluation.
  \begin{align*}
    & \con{\eV_1}{\ptit{Prg}}\paralS 
    \con{\eV_2}{\pioutA{c_1}{2}}\paralS\con{\eV_3}{\pioutA{c_2}{5}} \quad \evaluate\quad  \con{\eVV_1}{\pioutA{c_1}{(2,4)}}\paralS \con{\eVV_2}{\pioutA{c_4}{}}
  \end{align*}
  Two possible assignments for $\eV_1$, $\eV_2$ and $\eV_3$ that permit the above evaluation are: 
  \begin{align}\label{eq:83}
     \eV_1  &= \sset{\permin{c_1},\permin{c_2},\permout{c_4}}, & \eV_2 &=\sset{\permout{c_1}},&  \eV_3&=\sset{\permout{c_2}}  && \text{or};\\
     \label{eq:84}
    \eV_1  &= \sset{\permin{c_1},\permin{c_2}}, & \eV_2 &=\sset{\permout{c_1},\permout{c_4}}, & \eV_3 &=\sset{\permout{c_2}}
  \end{align}
Stated otherwise, we have  at least two possible linear-permission based narratives explaining why $\ptit{Prg}\;\paralS  \pioutA{c_1}{v_1}\paralS\pioutA{c_2}{v_2}$ is deterministic.
For both assignments  $\permout{c_1}\in\eVV_1$ and $\permout{c_4}\in\eVV_2$ must hold for the resulting safely-stable system $\con{\eVV_1}{\pioutA{c_1}{(2,4)}}\paralS \con{\eVV_2}{\pioutA{c_4}{}}$, but the remaining permissions $\permin{c_1},\,\permin{c_2}$ and $\,\permout{c_2}$, which are redundant at that point, can arbitrarily be split amongst $\eVV_1$ and $\eVV_2$.  More specifically, recall from \exref{ex:running} that 
\begin{align*}
      \ptit{Prg} &\deftri \restB{c_3}{ \ptit{Fltr} \paral \ptit{Dbl}}     \qquad \qquad \qquad \qquad
      \ptit{Dbl}   \deftri \piin{c_2}{x_2}{\piin{c_3}{x_4}{\pioutA{c_1}{(x_4\!+\!x_4)}}}\\
     \ptit{Fltr} & \deftri \piin{c_1}{x_1}{\cmp{x_1}{9}{\,\pioutA{c_3}{x_1}\paralS\piinBB{c_1}{x_3}{\pioutA{c_1}{(x_1,x_3)}\paralS \pioutA{c_4}{}} \,}{\,\pioutA{c_4}{x_1}}}
\end{align*}
Using the permission assignment in \eqref{eq:83} we can have the reduction sequence below.  Reduction \eqref{eq:85} can be derived using the rules \rtit{cLcl}, \rtit{cStr} and \rtit{cPar} from (\cf \figref{fig:confined-lang})  whereas reduction~\eqref{eq:86} is derived using \rtit{cSpl}, \rtit{cPar} and \rtit{cRes}; other reductions can be derived in similar fashion.  For the most part, we have abstract away from structural manipulation of terms, with the exception of reduction \eqref{eq:110}  which employs \rtit{cTgh} and \rtit{cStr} to discard the redundant scoped channel name $c_3$ and the permissions associated with it. 

\begin{align}
  \label{eq:85}
  & \con{\sset{\permin{c_1},\permin{c_2},\permout{c_4}}}{\ptit{Prg}}\;\paralS 
    \con{\sset{\permout{c_1}}}{\pioutA{c_1}{2}}\paralS\con{\sset{\permout{c_2}}}{\pioutA{c_2}{5}} \reduc \\
    \label{eq:86}
    & \restB{c_3}{\con{\sset{\permin{c_1},\permin{c_2},\permout{c_4},\permout{c_3},\permin{c_3}}}{\ptit{Fltr} \paral \ptit{Dbl}}\;\paralS 
    \con{\sset{\permout{c_1}}}{\pioutA{c_1}{2}}\paralS\con{\sset{\permout{c_2}}}{\pioutA{c_2}{5}}} \reduc \\
  \label{eq:104}
  & \restB{c_3}{\con{\sset{\permin{c_1},\permout{c_4},\permout{c_3}}}{\ptit{Fltr}} \paralS \con{\sset{\permin{c_2},\permin{c_3}}}{\ptit{Dbl}}\;\paralS 
    \con{\sset{\permout{c_1}}}{\pioutA{c_1}{2}}\paralS\con{\sset{\permout{c_2}}}{\pioutA{c_2}{5}}} \reduc \\
  \label{eq:105}
  & \restB{c_3}{\con{\sset{\permin{c_1},\permout{c_4},\permout{c_3},\permout{c_1}}}{
     \!\! \begin{array}{l}
\cmpA{2}{9}{}\\
{\;\pioutA{c_3}{2}\paralS\piinBB{c_1}{x_3}{\pioutA{c_1}{(2,x_3)}\paralS \pioutA{c_4}{}} }\\
\;\elseA{\,\pioutA{c_4}{2}}
\end{array}\!\!
} \paralS \con{\sset{\permin{c_2},\permin{c_3}}}{\ptit{Dbl}}\paralS\con{\sset{\permout{c_2}}}{\pioutA{c_2}{5}}} \reduc \\
\label{eq:106}
 & \restB{c_3}{\con{\sset{\permin{c_1},\permout{c_4},\permout{c_3},\permout{c_1}}}{
\pioutA{c_3}{2}\paralS\piinBB{c_1}{x_3}{\pioutA{c_1}{(2,x_3)}\paralS \pioutA{c_4}{}} } \paralS \con{\sset{\permin{c_2},\permin{c_3}}}{\ptit{Dbl}}\paralS\con{\sset{\permout{c_2}}}{\pioutA{c_2}{5}}} \reduc \\
\label{eq:107}
 & \restB{c_3}{\con{\sset{\permout{c_3},\permout{c_1}}}{
\pioutA{c_3}{2}} \paralS \con{\sset{\permin{c_1},\permout{c_4}}}{\piinBB{c_1}{x_3}{\pioutA{c_1}{(2,x_3)}\paralS \pioutA{c_4}{}} } \paralS \con{\sset{\permin{c_2},\permin{c_3}}}{\ptit{Dbl}}\paralS\con{\sset{\permout{c_2}}}{\pioutA{c_2}{5}}} \reduc \\
\label{eq:108}
& \restB{c_3}{
  \!\!\begin{array}{l}
\con{\sset{\permout{c_3},\permout{c_1}}}{
\pioutA{c_3}{2}} \paralS \con{\sset{\permin{c_1},\permout{c_4}}}{\piinBB{c_1}{x_3}{\pioutA{c_1}{(2,x_3)}\paralS \pioutA{c_4}{}} } \\
\quad \paralS \con{\sset{\permin{c_2},\permin{c_3},\permout{c_2}}}{\piin{c_3}{x_4}{\pioutA{c_1}{(x_4\!+\!x_4)}}}
\end{array}\!\!
} \reduc \\
&\label{eq:109}
\restB{c_3}{
\con{\sset{\permin{c_1},\permout{c_4}}}{\piinBB{c_1}{x_3}{\pioutA{c_1}{(2,x_3)}\paralS \pioutA{c_4}{}} } 
 \paralS \con{\sset{\permin{c_2},\permin{c_3},\permout{c_2},\permout{c_3},\permout{c_1}}}{\pioutA{c_1}{(2\!+\!2)}}
} \reduc \\
&
\restB{c_3}{
\con{\sset{\permin{c_1},\permout{c_4},\permin{c_2},\permin{c_3},\permout{c_2},\permout{c_3},\permout{c_1}}}{\pioutA{c_1}{(2,4)}\paralS \pioutA{c_4}{}} 
} \steq  \nonumber\\
\label{eq:110}
&\qquad\restB{c_3}{
\con{\sset{\permin{c_1},\permout{c_4},\permin{c_2},\permin{c_3},\permout{c_2},\permout{c_3},\permout{c_1}}}{\pioutA{c_1}{(2,4)}\paralS \pioutA{c_4}{}}\paralS \con{\emptyset}{\inert} 
}\reduc \\
&
\con{\sset{\permin{c_1},\permout{c_4},\permin{c_2},\permout{c_2},\permout{c_1}}}{\pioutA{c_1}{(2,4)}\paralS \pioutA{c_4}{}} 
 \paralS \restB{c_3}{\con{\emptyset}{\inert} }   \steq \nonumber\\
 \label{eq:111}
 & \qquad\con{\sset{\permin{c_1},\permout{c_4},\permin{c_2},\permout{c_2},\permout{c_1}}}{\pioutA{c_1}{(2,4)}\paralS \pioutA{c_4}{}}\reduc \\
 &\label{eq:112}
 \con{\sset{\permout{c_1},\permin{c_2},\permout{c_2}}}{\pioutA{c_1}{(2,4)}}\paralS \con{\sset{\permin{c_1},\permout{c_4}}}{\pioutA{c_4}{}} \reducNot \errNot
\end{align}\vspace{0 pt}

\noindent We highlight two important aspects of this reduction sequence.  First, reduction \eqref{eq:107} could have been interleaved with any of the reductions \eqref{eq:104},~\eqref{eq:105} and \eqref{eq:106} while still yielding the same safely-stable system; this holds because these reductions are confluent, as the separate permissions held by each subsystem attest. Second, we could have opted for a different permission partitioning in the reductions  \eqref{eq:86}, \eqref{eq:106} and   \eqref{eq:111},  and still attained a safely-stable system.  For instance, in  \eqref{eq:86} we could have allocated permission \permout{c_4} to the process \ptit{Dbl} and, similarly, in the case of \eqref{eq:106}  permission \permout{c_4} could have been allocated to the process \pioutA{c_3}{2}, without altering the eventual safely-stable system reached.

From the fact that \eqref{eq:112} is safely-stable and the contrapositive of Lemma~\ref{lem:violation-pres} we know that permissions were never violated throughout the reduction sequence.
\thmref{thm:eval-eq-upto-permissions} guarantees that the process part of any system evaluation will be structurally equivalent to $\pioutA{c_1}{(2,4)} \paralS \pioutA{c_4}{}$ and, by \thmref{thm:eval-implies-determinism} and  \thmref{thm:process-convergence}, this implies that $\ptit{Prg}\paralS  \pioutA{c_1}{v_1}\paralS\pioutA{c_2}{v_2}$ \emph{deterministically} evaluates to  $\pioutA{c_1}{(2,4)} \paralS \pioutA{c_4}{}$ \ie it always converges.

From a compositional perspective, permission-sets also delineate the footprint of every process and, indirectly, the requirement for well-resourcing of \defref{def:wf-systems} defines an interface for detecting race conditions.  Consider for example the system:
\begin{align*}
  & \con{\sset{\permin{c_1},\permin{c_2},\permout{c_4}}}{\ptit{Prg}}
\end{align*}
 In order for this system to be safe, it needs the permission \permin{c_1} (otherwise it would yield a permission violation through rule \rtit{eIn}).  Recall the context $\pioutA{c_1}{1}\paral \pioutA{c_2}{v_2} \paral\piin{c_1}{x}{\inert} $ from \exref{ex:running} which had introduced a race condition on inputs on channel $c_1$.   In order for this system not to violate permissions itself, it must own a permission set $\eVV$ \ie $\con{\eVV}{\pioutA{c_1}{1}\paral \pioutA{c_2}{v_2} \paral\piin{c_1}{x}{\inert}}$, where $\permin{c_1} \in \eVV$ as well.  However, the separation condition  for well-resourcing  prohibits us from composing these two systems together because  their respective permissions are not disjoint \ie $\sset{\permin{c_1},\permin{c_2},\permout{c_4}} \not\perp \eVV$.  

\end{exa}

\begin{exa} \label{ex:quicksort-determ} If, in the  array ${\pioutA{a_1}{v_1}\paral\ldots\paral\pioutA{a_n}{v_n}}$ to be sorted, we assign the permission set $\eVV_i = \sset{\permout{a_i}}$ to every element $\pioutA{a_i}{v_i}$  and assign the permission set $\eV = \sset{\permin{a_1},\ldots,\permin{a_n},\permout{r}}$ to $\quick{1,n}$  then it turns out that we can show that
  \begin{align*}
    \conC{\quick{1,n}} \paral \con{\eVV_1}{\pioutA{a_1}{v_1}} \paral\ldots\paral\con{\eVV_n}{\pioutA{a_n}{v_n}} & \quad \evaluate \quad T
  \end{align*}
  for some safely stable system $T$ where 
  \begin{align*}
    T & \steq  \con{\eVV_1}{\pioutA{a_1}{u_1}} \paral\ldots\paral\con{\eVV_n}{\pioutA{a_n}{u_n}} \paral \conC{\pioutA{r}{}}
  \end{align*}
Note how, as in \exref{ex:running-confined}, $\eV$ in \conC{\quick{1,n}} defines an interface that parallel processes to be composed with it to respect, in order for it to evaluate deterministically.  
\end{exa}

\subsection{Discussion}
\label{sec:discussion}

Process spawning, \rtit{cSpl}, is  intentionally
non-deterministic: apart from alleviating permission annotation,\footnote{ The current formulation leads to a more lightweight form of annotation for confined processes. The other alternative would have been to extend the definition of parallel composition at the process level and have systems of the form \conC{P \paral_{(\eVV_1, \eVV_2)} Q}. whereby $\eVV_1$ and $\eVV_2$ specify \emph{deterministically} how $\eV$ is to be apportioned amongst $P$ and $Q$.} its non-deterministic nature is in line with the unspecified way that permissions 
can be allocated  in a confined system.  Correspondingly, through \thmref{thm:eval-eq-upto-permissions} and Corollary~\ref{cor:sys-eval-implies-proc-deter}, we have seen how there may be more than one way how to validly distribute permissions across processes so as to prove determinacy.  

Since we eventually plan to use confined processes as part of the model for our logic (\cf \secref{sec:logic}), 
 we here opt for the most flexible solution \ie non-deterministic splits for parallel composition, which permits more narratives explaining process determinism while still restricting the permission allocations that can be used.  This setup gives better separation of concerns between confined process reduction and the model used for our logic.   In particular, this model  
 incorporates environments describing permission-transfer invariants, apart from confined processes. These environments are however orthogonal to the properties of confined processes derived in this section.  In fact, their purpose is that of  allowing for better compositional analysis when determining assertion satisfactions, as we shall see in \secref{sec:logic} and \secref{sec:proof-system}.


\section{Logic}
\label{sec:logic}

We define a separation-based logic that enables us to reason about programs that deterministically evaluate to stable systems satisfying assertions describing their state.  Our logic concentrates more on describing data held at asynchronous outputs in stable systems, and abstracts away from issues dealing with control for deterministic evaluation. 
 For this reason, the logic semantics is not defined directly on bare processes.  Instead, the confined processes  of \secref{sec:conf-value-pass}  together with the definitions for safe-stability and evaluations, \defref{def:safe-stability-eval},    provide the basis for a model to our separation logic whereby the permissions owned constitute our units of separation (\cf \defref{def:separation}).   
 Together with the associated proof system of \secref{sec:proof-system}, this amounts to our proposal for a  logical framework for reasoning over non-interfering concurrent programs.

\subsection{Permission Environments}
\label{sec:perm-envir}

In our logic, channels have a dual role.  Apart from acting as a mechanism for communicating data, they also act  
as delimiters of \emph{mutual-exclusion groups of resources},  modeling condition-critical regions\cite{Ohearn04}.   Each input process \piin{c}{\lst{x}}{P}  abides 
to use certain permissions in $P$ only \emph{after}  it synchronises on channel $c$  whereas each output-process $\pioutA{c}{\lst{e}}$    obliges 
to own the permissions \emph{guarded} by $c$; these guarded permissions are  transferred \emph{dynamically} upon communication on $c$ using rule \rtit{cCom} of \figref{fig:confined-lang} and enable us to reason about channel reuse in deterministic systems.

The invariants relating to  permission mutual-exclusion are characterised 
as \emph{permission environments}, $\env \in \Chans\pmap \pset{\Perm}$,  partial maps associating 
channels $c$ to permission-sets $\eV$.  They require abiding processes to own all the permissions in $\eV$ when  outputting on $c$ and, dually, allow processes to assume the acquisition of all permissions in $\eV$ when inputting on $c$. 
The constraints in \defref{def:permission-environment} ensure that 
\begin{inparaenum}[\itshape \upshape(1\upshape)]
 \item permission transfer always includes  the permission
   $\permout{c}$ to output over the communicating channel, but never 
  the capability
 \permin{c} to input over it, as this must 
  already belong to the receiving process;
 \item environments are suitably closed.
\end{inparaenum}

\begin{defi} [\rm\textit{Permission Environment}] 
\label{def:permission-environment} \rm
 $\env$ is a finite map  from names to permission sets
 such that: 
\begin{enumerate}[(1)]
\item   
  \begin{math}
    \text{forall } c\in\dom{\env}\; \permin{c}\not\in\env(c)\text{ and
    }\permout{c}\in\env(c),
  \end{math}
\item  
  \begin{math}
    \eV\in\cod{\env}\;\text{ implies
    }\;\names{\eV}\subseteq\dom{\env},
  \end{math}
\end{enumerate}
 where 
 $\names{\eV}\deftxt\sset{c\mid\permin{c}\in\eV\text{ or }\permout{c}\in\eV}$.
\end{defi}

\subsection{Logical Formulas}
\label{sec:logical-formulas}

Our logic formulas, ranged over by the meta-variables \fV,\fVV,  characterise a
 `spatial' notion of \emph{state} for deterministic processes in terms of the data held on asynchronous channels at stable processes.  In order to simplify our conceptual process interpretations, we limit ourselves to describing only the states of stable processes, abstracting away from the intermediary reductions that lead to stability. For this we require
 asynchronous output data assertions, \fstate{c}{\lst{e}},  the `separated conjunction', \fcons{\fV}{\fVV}, and its unit , \femp;  formulas constructed using just these constructs are denoted by the metavariable \fVVVV and are called \emph{state} formulas. 
  Guided by Proposition~\ref{prop:safe-stab-vs-system-structure}, stability  requires our formulas to  describe (input) blocked processes,  \fblk{c}.  Finally, we  also describe unrestricted terminating process by \fany\ whenever we want to abstract away completely from the structure of a terminating process.

\begin{defi}[Formulas]\label{def:formula} \rm
\begin{align*} 
  \fVVVV,\fVVV \in \AFrm & \bnfdef \;\femp \bnfsepp \fstate{c}{\lst{e}} \bnfsepp  \fcons{\fVVVV}{\fVVVV} \\
   \fV,\fVV  \in\Frm & \bnfdef \;\femp \bnfsepp \fany \bnfsepp \fstate{c}{\lst{e}} \bnfsepp \fblk{c} \bnfsepp  \fcons{\fV}{\fV} 
\end{align*}
 \end{defi}

Our formulas are interpreted over permission environments and well-formed
systems,  \ie $\env,S\,\sat\,\fV$.  They are defined in \figref{fig:assertion-satisfaction}, inductively  on the 
 structure of \emph{closed} formulas  \ie  formulas with no free variables in the expressions \lst{e} of \fstate{c}{\lst{e}}.   
Our definition of formula satisfaction relies heavily on the evaluation judgement, $S \evaluate T$, which is only defined for closed systems (\defref{def:safe-stability-eval}); recall that system evaluation \emph{existentialises} over a reduction path leading to a stable system .  


\begin{small_display}{Formula Satisfaction}{fig:assertion-satisfaction} 
 \begin{align*}
   \\
    \env,S  &\sat \femp && \text{iff }  S\evaluate\con{\emptyset}{\inert};\\[6pt]
    \env,S  &\sat \fany && \text{iff } S\evaluate T  ;\\[6pt]
    \env,\,S  &\sat \fstate{c}{\lst{e}} && \text{iff }
       S\evaluate \conC{\pioutA{c}{\lst{e'}}}
       \;\text{with}\;\lst{e}\evaluate \lst{v},\;\lst{e'}\evaluate \lst{v}\,\text{and}\, \env(c)\subseteq\eV;
    \\[5pt]
    \env,\,S &\sat  {\fblk{c}} &&  \text{iff }  S\evaluate\,\rest{\lst{d}}{\conC{\piin{c}{\lst{x}}{P}}} \;\text{with} \;c\not\in\lst{d} \;\text{and} \; c \in \dom{\env}
    ;\\[6pt]
    \env,\,S  &\sat  \fcons{\fV_1}{\fV_2} &&\text{iff }
      S\evaluate\,
      \restB{\lst{d}}{S_1\paral S_2}\text{ with }
      \lst{d}\not\in 
      \dom{\env}
      \text{ and }\env,S_1\sat\fV_1\text{ and } \env,S_2\sat\fV_2;\\
  \end{align*}
\end{small_display}

The satisfaction relation in \figref{fig:assertion-satisfaction} describes the state of a system once it stabilises.
The main assertion satisfaction is that for  data assertions, \fstate{c}{\lst{e}}, as it relates the data held on asynchronous outputs of a stable system with the data stated in the assertion.  To do this, the definition  relies on the assumption that  $S$ is closed to establish the equality between the two expressions \lst{e} and \lst{e'}.  Moreover, it uses the environment, \env, to ensure that the (stable) asynchronous output 
owns the permissions imposed by the permission guarding invariants.  Its use has already been discussed in \secref{sec:perm-envir} and will be elaborated further when we consider compositional analysis of satisfaction in \secref{sec:proof-system}.   Data assertions are typically composed together using the separating conjunction assertion,  $\fcons{\fV_1}{\fV_2}$, and the empty assertion, \femp.  For the satisfaction for \femp, 
the system $\con{\emptyset}{\inert}$ is chosen to be the identity interpretation for our model \wrt separation, thereby making the interpretation  for just these constructs a commutative monoid (\cf Lemma~\ref{lem:formula-equiv}).

 The satisfaction definition of the separating conjunction, 
  $\fcons{\fV_1}{\fV_2}$, is however more complicated than one would have expected, as it needs to handle  conjunctions with \fblk{c} and \fany\ formulas as well; the interpretation for the latter two formulas is rather straightforward. 
Thus, apart from relying on the system well-resourcing assumption to guarantee that the partitioned sub-systems are separate, $S_1 \perp S_2$ (\cf \defref{def:separation}), satisfaction for the separating conjunction also enforces that a system is stable before it is split, \ie $S\evaluate S_1\paral S_2$.  This condition rules out systems whose subcomponents satisfy the sub-formulas of a conjunction  $\fcons{\fV_1}{\fV_2}$, but then violate stability once composed together; we return to this later in \exref{ex:unsatisfiability}.  The fact that separating conjunction ranges over input-blocked processes also requires  a satisfaction definition that ignores scoping of channel names across separation \ie  $S\evaluate \restB{\lst{d}}{S_1\paral S_2}$; these scoped names \lst{d} refer to channels used  in the continuations of blocked processes, as explained later in \exref{ex:satisfiability}, and cannot be abstracted away using structural equivalence rules such as \rtit{scExt} and \rtit{scNew} from \figref{fig:confined-lang}.

 \begin{exa}[Satisfiability]\label{ex:satisfiability} Recall the process definitions 
\begin{align*}
      \ptit{Prg} &\deftri \restB{c_3}{ \ptit{Fltr} \paral \ptit{Dbl}}  \\  
      \ptit{Dbl}  & \deftri \piin{c_2}{x_2}{\piin{c_3}{x_4}{\pioutA{c_1}{(x_4\!+\!x_4)}}}\\
     \ptit{Fltr} & \deftri \piin{c_1}{x_1}{\cmp{x_1}{9}{\,\pioutA{c_3}{x_1}\paralS\piinBB{c_1}{x_3}{\pioutA{c_1}{(x_1,x_3)}\paralS \pioutA{c_4}{}} \,}{\,\pioutA{c_4}{x_1}}}
   \end{align*}
from \exref{ex:running}. Assuming the environment 
 \begin{align*}
\env &= \emap{c_1}{\sset{\permout{c_1}}},\,\emap{c_2}{\sset{\permout{c_2}}},\,
   \emap{c_4}{\sset{\permout{c_4},\permin{c_1}}}      
   \end{align*}
we have the following satisfactions:
\begin{align}\label{eq:63}
  \env, \con{\sset{\permin{c_1},\permin{c_2},\permout{c_4}}}{\ptit{Prg}}\paralS 
    \con{\sset{\permout{c_1}}}{\pioutA{c_1}{2}}\paralS\con{\sset{\permout{c_2}}}{\pioutA{c_2}{5}} & \models \fcons{\fstate{c_1}{2,4}}{\fstate{c_4}{}} \\
    \label{eq:64}
\env, \con{\sset{\permin{c_1},\permin{c_2},\permout{c_4},\permout{c_1},\permout{c_2}}}{\ptit{Prg}\paralS\pioutA{c_1}{2}\paralS\pioutA{c_2}{5}} & \models \fcons{\fstate{c_1}{2,4}}{\fstate{c_4}{}} \\
\label{eq:65}
  \env, \con{\sset{\permout{c_1}}}{\pioutA{c_1}{(5\!-\!3,3\!+\!1)}} \paralS \con{\sset{\permout{c_4},\permin{c_1}}}{\pioutA{c_4}} & \models \fcons{\fstate{c_1}{2,4}}{\fstate{c_4}{}} 
\end{align}
whereby, according to the definition in \figref{fig:assertion-satisfaction}, satisfaction is only concerned with the existence of a reduction path to a stable system, where the outputs corresponding to data assertions are required to own the permissions expected by permission environment \env;  the reduction path \eqref{eq:63} and \eqref{eq:64} has already been discussed in \exref{ex:running-confined}.  Satisfaction for \eqref{eq:65} is more straightforward to determine as the system is stable.  On the other hand, for \env defined above, the following do not satisfy their respective assertions:
\begin{align}
  \label{eq:66}
  \env, \con{\sset{\permin{c_1},\permin{c_2},\permout{c_4}}}{\ptit{Prg}}\paralS 
    \con{\emptyset}{\pioutA{c_1}{2}}\paralS\con{\sset{\permout{c_2}}}{\pioutA{c_2}{5}} & \not\models \fcons{\fstate{c_1}{2,4}}{\fstate{c_4}{}} \\
    \label{eq:67}
   \env, \con{\sset{\permin{c_2},\permout{c_4}}}{\ptit{Prg}}\paralS 
    \con{\sset{\permout{c_1}}}{\pioutA{c_1}{2}}\paralS\con{\sset{\permout{c_2}}}{\pioutA{c_2}{5}} & \not\models \fcons{\fstate{c_1}{2,4}}{\fstate{c_4}{}} \\
    \label{eq:68}
    \env, \con{\sset{\permout{c_1}}}{\sset{\pioutA{c_1}}{(5\!-\!3,3\!+\!1)}} \paralS \con{\sset{\permout{c_4}}}{\pioutA{c_4}} & \not\models \fcons{\fstate{c_1}{2,4}}{\fstate{c_4}{}}\\
    \label{eq:69}
    \env, \con{\sset{\permout{c_1}}}{\pioutA{c_1}{(2,3)}} \paralS \con{\sset{\permout{c_4},\permin{c_1}}}{\pioutA{c_4}} & \not\models \fcons{\fstate{c_1}{2,4}}{\fstate{c_4}{}}  
\end{align}
The first two systems fail to satisfy the assertion because they cannot evaluate to safely-stable systems due to lack of permission. In particular, in \eqref{eq:66} process \pioutA{c_1}{2} does not own permission \permout{c_1} required for communication (\cf \rtit{cCom} in \figref{fig:confined-lang})  whereas in \eqref{eq:67} \ptit{Prg} is missing permission \permin{c_1} .  The third system, \eqref{eq:68},  fails to satisfy the assertion although it is already a safely-stable system, as it violates the permission obligations for outputs imposed by \env \ie output \pioutA{c_4} does not own permission \permin{c_1}.  Finally, the fourth system  \eqref{eq:69} fails to satisfy the assertion due to a mismatch between the data expected by the assertions and the data communicated by the outputs.  We also have the following satisfactions involving the other assertion forms of the logic:
\begin{align}\label{eq:70}
     (\env, \emap{c_3}{\sset{\permout{c_3}}} ), \con{\sset{\permin{c_1},\permin{c_2},\permout{c_4}, \permin{c_3}}}{\ptit{Flt}\paral\ptit{Dbl}}\paralS 
    \con{\sset{\permout{c_1}}}{\pioutA{c_1}{10}}\paralS\con{\sset{\permout{c_2}}}{\pioutA{c_2}{5}} & \models \fcons{\fstate{c_4}{10}}{\fblk{c_3}} \\
    \label{eq:71}
 \env, \con{\sset{\permin{c_1},\permin{c_2},\permout{c_4}}}{\ptit{Prg}}\paralS 
    \con{\sset{\permout{c_1}}}{\pioutA{c_1}{10}}\paralS\con{\sset{\permout{c_2}}}{\pioutA{c_2}{5}} & \models \fcons{\fstate{c_4}{10}}{\fany} \\  
    \label{eq:72}
  \env, \con{\sset{\permin{c_1},\permin{c_2},\permout{c_4}}}{\ptit{Prg}}\paralS 
    \con{\sset{\permout{c_1}}}{\pioutA{c_1}{10}}\paralS\con{\sset{\permout{c_2}}}{\pioutA{c_2}{5}} & \models \fany \\
    \label{eq:73}
    \env, \restB{c_3}{\con{\sset{\permin{c_1}}}{\piin{c_1}{}{\pioutA{c_3}{}}} \paral \con{\sset{\permin{c_2}}}{\piin{c_2}{}{\piin{c_3}{}{\inert}}} } & \models \fcons{\fblk{c_1}}{\fblk{c_2}} 
\end{align}
Satisfaction \eqref{eq:70} requires us to extend \env to account for the permission invariants of channel $c_3$, which is not scoped. We also need the input permission \permin{c_3} as dictated by the satisfaction of the sub-assertion  $\fblk{c_3}$  in \figref{fig:assertion-satisfaction}.  In  the subsequent satisfaction, \eqref{eq:71}, \fany\ is used to describe the input-blocked  process on a scoped channel $c_3$  that is scoped in \ptit{Prg} (recall that $\ptit{Prg} \deftri \restB{c_3}{\ptit{Fltr}\paral\ptit{Dbl}}$).   Note also how, in \eqref{eq:73}, since $c_3 \not\in \dom{\env} 
$ (\cf satisfaction for $\fcons{\fV_1}{\fV_2}$ in \figref{fig:assertion-satisfaction}), the scoping of $c_3$  does not prohibit us from splitting the system to determine the satisfaction of the subcomponents of the formula \ie $\fblk{c_1}$ and $\fblk{c_2}$.
\end{exa}

The requirement that satisfaction is limited to \emph{safe} evaluations in \figref{fig:assertion-satisfaction} intentionally makes certain formulas unsatisfiable.   Alternative definitions could have been possible whereby  we allow systems to temporarily satisfy a formula but then fail to satisfy it as computation progresses, meaning that the eventual stable system would not necessarily satisfy the formula.  However, as discussed briefly in the Introduction, in our eventual framework of \secref{sec:proof-system}, systems will have the dual role of acting both  as state as well as state-transformers.  We therefore opted for the simpler interpretation that is conceptually easier to work with  and chose a satisfaction interpretation  that can be easily reasoned about in terms of the eventual stable systems reached.

\begin{exa}[Unsatisfiability]\label{ex:unsatisfiability} 
  Formulas such as the ones below are unsatisfiable under the interpretation given in \figref{fig:assertion-satisfaction}.  
  \begin{align*} 
    &\rm\fcons{\fstate{c}{5}}{\fstate{c}{6}}
   && \rm \fcons{\fstate{c}{1}}{\fblk{c}}
  \end{align*}
In the first case, \ie $\fcons{\fstate{c}{5}}{\fstate{c}{6}}$,  sub-systems respectively satisfying  \fstate{c}{5} and \fstate{c}{6} 
can never be merged into a well-resourced system as they must conflict on the permission $\permout{c}$ irrespective of the narrative chosen, due to the environment constraints set out in  \defref{def:permission-environment}.    This is desirable because any system satisfying the first formula will create  a race condition for any inputs on the channel $c$.

In later case, \ie \fcons{\fstate{c}{1}}{\fblk{c}}, sub-systems satisfying the sub-formulas of the separating conjunction  become \emph{unstable} once they are composed  in parallel  violating their respective sub-formula satisfaction.  Hence any such satisfying system would violate the evaluation condition imposed on the satisfaction of the conjunct formula $\fcons{\fV_1}{\fV_2}$ in \figref{fig:assertion-satisfaction}.  In fact, any sub-system $S_1$ satisfying $\fstate{c}{1}$ must evaluate to a stable system of the form \conC{\pioutA{c}{e}} where $e \evaluate 1$.  Similarly any sub-system $S_2$ satisfying $\fblk{c}$ must evaluate to a stable system that is structurally equivalent to $\rest{\lst{d}}{\conCC{\piin{c}{x}{P}}}$ (where $c\not\in\lst{d}$).  This means that, by the semantics of \secref{sec:conf-value-pass}, $\conC{\pioutA{c}{e}} \paral \rest{\lst{d}}{\conCC{\piin{c}{x}{P}}}$ is not stable, even if it is well-resourced (\ie $\eV \cap \eVV = \emptyset$).  Our satisfaction definition $\fcons{\fV_1}{\fV_2}$ rules out this possibility by first requiring the composite system evaluates to a stable system before splitting.  There are two reasons for this stricter interpretation.  First, once the reduction happens leading to an evaluation to some other stable state  $S_3$
\begin{displaymath}
\conC{\pioutA{c}{e}} \paral \rest{\lst{d}}{\conCC{\piin{c}{x}{P}}} \quad \reduc\quad \rest{\lst{d}}{\con{\eVV\cup\eV}{P\subC{1}{x}}}\;\evaluate S_3
\end{displaymath}
it may be the case that $S_3$ does not satisfy \fcons{\fstate{c}{1}}{\fblk{c}} anymore.  Second, and perhaps more importantly, the above reduction can potentially trigger permission-violating or non-terminating behaviour in \rest{\lst{d}}{\con{\eVV\cup\eV}{P\subC{1}{x}}}.  For instance, process $P$ may be of the form $\pioutA{d}{1} \paral \pioutA{d}{2} \paral \piin{d}{y}{\pioutA{c}{(x+y)}}$ \ie it has two competing outputs on channel $d$.  This implies that, whereas \rest{\lst{d}}{\conCC{\piin{c}{x}{P}}} is safely-stable, its continuation is permission-violating, irrespective of the permissions held at that point, because it can hold at most one permission to output on channel $d$.    
\end{exa}

\medskip

 Since structural equivalence is central to \defref{fig:assertion-satisfaction} ($\evaluate$ in \defref{def:safe-stability-eval} incorporates it), satisfaction abstracts over structurally equivalent systems, which allows us to work up-to structural equivalence when reasoning about systems.   Moreover, we can also reason about formula satisfaction from existing system-formula satisfaction and systems that reduce (converge) to them in zero or more steps. 

\begin{prop}[Satisfaction and Evaluation]\label{cor:satisfaction-evaluation}\rm
 $\env, S \models \fV$  implies $\exists T. \ S \evaluate T$ and $\env,T\models \fV$  
\end{prop}

\begin{prop}[Structural Eq. and Satisfaction]\rm\label{lemma:proc-struct-eq}
  \begin{math}
    \env,S\sat\fV \text{ and } S\steq T\text{ implies } \env,T\sat\fV
  \end{math}
\end{prop}

\begin{prop}[Satisfaction and Convergence]\label{cor:satisfaction-convergence}\rm
 $\env, S \models \fV$ and $T \reduc^\ast S$ implies $\env,T\models \fV$  
\end{prop}


\medskip
We overload $\models$ to denote semantic implication amongst formulas in standard fashion.  We then are able to prove certain properties about our logic, stated in Lemma~\ref{lem:formula-equiv}.    

\begin{defi}[Semantic Implication]\label{def:logic-impl}\rm 
  \begin{math}
    \fV \models \fVV \;\deftxt\quad   \env, S \models \fV\text{ implies } \env, S \models\fVV
  \end{math}
\end{defi}

\begin{lem}[Formula equivalence] \rm\label{lem:formula-equiv} The following bidirectional implications hold:
  \begin{enumerate}[(1)]
  \item $\fcons{\femp}{\fV}   \quad  \modeledby \models\quad  \fV  $ 
  \item $  \fconsBr{\fV_1}{\fcons{\fV_2}{\fV_3}}  \quad  \modeledby \models\quad  \fconsBl{\fcons{\fV_1}{\fV_2}}{\fV_3}$ 
  \item $ \fcons{\fV}{\fVV}  \quad  \modeledby \models\quad \fcons{\fVV}{\fV}\qquad$ 
  \end{enumerate}
\end{lem}

\subsection{Composing satisfactions}
\label{sec:comp-satisf}
Recall, from \exref{ex:unsatisfiability}, that the satisfaction of the sub-assertions  $\fV_1$ and $\fV_2$ does not necessarily imply the satisfaction of the composite assertion,  \fcons{\fV_1}{\fV_2}.  Nevertheless  it is possible to determine when it is safe to infer this by analysing the structure of the sub-formulas.  This analysis is formalised as the formula separation judgement, denoted as $\sepproc{}{\fV_1}{\fV_2}$ and defined in \defref{def:fromula-sep}.
This judgement relies on the functions $\edg{}{}$ and $\trg{}{}$ to  conservatively approximate matching outputs and inputs across sub-systems satisfying the formulas $\fV_1$, $\fV_2$ and, by prohibiting such matching channel operations, it ensures that no new reductions are introduced when sub-systems are composed in parallel. As a result, sub-systems that  satisfy sub-formulas in a separating conjunction formulas must still satisfy the conjunction formula once composed, as stated in Lemma~\ref{lem:merge-assert}.   This formula separation judgement is used later on by the proof system in Section~\ref{sec:proof-system} to circumvent the construction of problematic formulas such as those discussed  in \exref{ex:unsatisfiability}.

\begin{defi}[Formula  Edges, Triggers and Separation]\label{def:fromula-sep}\rm
  \begin{align*}
    \edg{}{\fV} & \deftxt
    \begin{cases}
      \emptyset &\text{if } \fV = \femp \text{ or } \fV=\fblk{c} \\
      \eset{\permout{c}} &\text{if } \fV = \fstate{c}{\lst{e}}\\
      \edg{}{\fV_1} \cup \edg{}{\fV_2} &\text{if } \fV=\fcons{\fV_1}{\fV_2}\\
      \text{undefined} & \text{otherwise}
    \end{cases}\\
    \trg{}{\fV} & \deftxt
    \begin{cases}
      \emptyset &\text{if }  \fV = \femp \text{ or } \fV=\fstate{c}{\lst{e}} \\
      \eset{\permout{c}} &\text{if } \fV = \fblk{c}\\
      \trg{}{\fV_1} \cup \trg{}{\fV_2} &\text{if } \fV=\fcons{\fV_1}{\fV_2}\\
      \text{undefined} & \text{otherwise}
    \end{cases}
  \end{align*}
  \begin{align*}
    \sepproc{}{\fV}{\fVV}  & \deftxt \quad\edg{}{\fV}\cap\trg{}{\fVV} = \emptyset \quad\wedge\quad \edg{}{\fVV}\cap\trg{}{\fV} = \emptyset
  \end{align*}
\end{defi}

\begin{lem}[Merging Assertions]\label{lem:merge-assert}
  \begin{displaymath}
    \env,S \sat \fV\text{ and }\env,T \sat \fVV  \text{ and } S\perp T \text{ and } \sepproc{}{\fV}{\fVV} \quad\text{implies}\quad \env,\, S\paral T \satS \fcons{\fV}{\fVV}
  \end{displaymath}
\end{lem}

\proof See Appendix~\ref{sec:logic-app}.\qed

Note that, for a number of conjunctions, the sub-formulas are trivially separate making formula separation checks superfluous.  For instance, \femp is separate from any formula; also state formulas \fcons{\fVVVV_1}{\fVVVV_2} are 
trivially separate, \sepproc{}{\fVVVV_1}{\fVVVV_2} as stated in Proposition \ref{cor:sep-formulas}.  

\begin{prop}  \label{cor:sep-formulas} For any environment, \env, state formulas, $\fVVVV, \fVVV$ and formula $\fV$ we have:
\begin{enumerate}[\em(1)]
 \item \begin{math}
    \sepproc{}{\fVVVV}{\fVVV} 
  \end{math}
 \item \begin{math}
    \sepproc{}{\fV}{\femp}
  \end{math}
\end{enumerate}
\end{prop}

\proof Immediate from \ref{def:fromula-sep} \qed

\section{Proof System}
\label{sec:proof-system}

\renewcommand{\edgE}[1]{\ensuremath{\textbf{edg}(#1)}}
\renewcommand{\trgE}[1]{\ensuremath{\textbf{trg}(#1)}}
\renewcommand{\sepproc}[3]{{\ensuremath{ #2 \esep #3}}}

We complete our framework by developing a compositional 
 proof-system for the logic of \S\ref{sec:logic}, interpreted according 
 to the satisfaction of \figref{fig:assertion-satisfaction}.   Our sequents, inspired by Hoare triples, 
  have the format
 \begin{align*}
  \seqEB{\fV}{S}{\fVV}\,,
 \end{align*}
 where $S$ is a well-resourced system, $\fV$ and $\fVV$ are respectively the 
 pre-condition and post-condition, \env is a permission environment, and $\bV$ is a boolean expression defined  in \figref{fig:structural-reduc}, now serving as a  boolean formula 
 over our value domain.   The system, formulas and boolean condition in a sequent  are potentially open \ie that may have free variables.  Thus, the meaning of our sequents quantifies over all substitutions, $\sigma\in \Subs$ that make the boolean condition evaluate to true, and also over all systems $T\in\Sys$ which are separate from $S$ and which satisfy the precondition in the following way.
\begin{defi}[\rm\textit{Sequent satisfaction}]%
\label{defn:sequent-satisfaction}\rm  
  \begin{align*}
    \seqsemEB{\fV}{S}{\fVV} &\quad\deftxt\quad
       \forall \sigma, T. \ \bV\sigma\evaluate\boolT ,\;\; 
       \env,T\sigma\models \fV\sigma,\;\; 
       T\sigma\perp S\sigma
       \quad\text{implies}\quad 
\env,\, (T \paral S)\sigma \models \fVV\sigma 
  \end{align*}
\end{defi}\vspace{5 pt}

\noindent As in \cite{HennessyLin96}, our sequents 
tease apart auxiliary reasoning about our value domain, since  determining the truth (or otherwise) of these  boolean
formulas is \emph{process-independent}.   Such disentangling also allows us to  make refined claims about derivations in our system.  For instance, if we limit value expressions to Presburger arithmetic, we know that our  boolean formula derivations exists and are decidable \cite{presburger29}.

We note that our sequents deal with \emph{total-correctness}.  Formula satisfaction, defined in \figref{fig:assertion-satisfaction}, centers around system evaluation, $S \evaluate T$, which \emph{existentially} quantifies over one sequence of system reductions.
 The strength of what may, at first, seem a rather weak behaviour assertion comes from the determinism properties  afforded by our model of confined processes. In fact, \thmref{thm:eval-eq-upto-permissions} (Evaluation Determinism) allows us to extend such behaviour assertions to universal system behaviour, up-to redundant permissions.  What we are ultimately interested in however is universal \emph{processes} behaviour.  This can then be retrieved  in immediate fashion through \defref{def:proc-sat} (Process Satisfaction), defined later in \secref{sec:process-sequents}, \thmref{thm:process-convergence} (Process Convergence), and ultimately, \thmref{thm:eval-implies-determinism} (Process Evaluation Determinism).


\begin{small_display}{Sequent Rules} {fig:proof-system-1}
\begin{align*}
& \lefteqn{\textbf{Logical Rules} }\\[12pt]
& \inference[\rtit{lNil}]{}
    {\seqEB{\fV}{\conC{\inert}}{\fV}}\qquad 
 \inference[\rtit{lFls}]{  }
     {\seqE{\bfalse}{\fV}{\!S\!}{\fVV}}
    \qquad
 \inference[\rtit{lBlk}]
    {\permin{c}\in\eV}
    {\seqEB{\femp}{\conC{\piin{c}{\lst{x}}{P}}}{\fblk{c}}}
    \\[12pt]
& \inference[\rtit{lOut}]
    {
     \env(c)\subseteq\eV}
    {\seqEB{\femp}{\conC{\pioutA{c}{\lst{e}}}}{\fstate{c}{\lst{e}}}}
\qquad \qquad
\inference[\rtit{lIn}] 
     {\permin{c}\in\eV & \seqEB{\fV} {\conC{P\subC{\lst{e}}{\lst{x}}} \paral S} {\fVV} }
     {\seqEB{\fcons{\fV}{\fstate{c}{\lst{e}}}} {\con{\eV\setminus\env(c)} {\piin{c}{\lst{x}}{P}} \paral S } {\fVV}}
\\[12pt]  
&\inference[\rtit{lIf}]
    {\seqE{\band{\bV_1}{\bV_2}}{\fV}{\conC{P}\paral S}{\fVV} 
      \\\seqE{\band{\bV_1}{\bnot{\bV_2}}}{\fV}{\conC{Q}\paral S }{\fVV}}
    {\seqE{\bV_1}{\fV}{\conC{\cmpC{\bV_2}{P}{Q}}\paral S}{\fVV}} 
\qquad\quad \inference[\rtit{lDef}]
    {\pCall{\pV}{\lst{x}}\deftri P \quad \seqEB{\fV}{\conC{\,P\subC{\lst{e}}{\lst{x}}\subC{\lst{c}}{\lst{d}}\,}\paral S}{\fVV}}
   {\seqEB{\fV}{\conC{\,\pCall{\pV}{\lst{e}}\subS{\lst{c}}{\lst{d}}\,}\paral S}{\fVV}}
\\[12pt]
& \inference[\rtit{lPar}]
    {\seqEB{\fV_1}{S}{\fcons{\fVV_1}{\fV_3}} & \sepprocE{\fV_2}{\fV_3}\\
    \seqEB{\fcons{\fV_2}{\fV_3}}{T}{\fVV_2}  &
    \sepprocE{\fVV_1}{\fVV_2} }
    {\seqEB{\fcons{\fV_1}{\fV_2}}
    	{S\paralS T}{\fcons{\fVV_1}{\fVV_2}}}
 \qquad \qquad 
  \inference[\rtit{lSpl}]
    {\seqEB{\fV}{\con{\eV}{P}\paralS\con{\eVV}{Q}\paral S}{\fVV} }
    {\seqEB{\fV}
    	{\con{\eV\uplus\eVV}{P\paralS Q}\paral S}{{\fVV}}}
\\[12pt]
&\;  
  \inference[\rtit{lRes}]
    {\seqEB{\fV}{S}{\fVV}}
    {\seqB{\env\!\setminus\!\lst{c}}{\fV}{\rest{\lst{c}}{S}}{\fVV\!\setminus\!\lst{c}}}
  \qquad\qquad
  \inference[\rtit{lLcl}]
    {\seqEB{\fV}{\rest{c}{\con{\eV\uplus\eset{\permin{c},\permout{c}}}{P}}\paral S}{\fVV}}
    {\seqEB{\fV}{\conC{\rest{c}{P}}\paral S}{\fVV}}
   \\[10pt]
&\lefteqn{\textbf{Structural Rules} }\\[12pt]
&\inference[\rtit{lInst}]{\seqEB{\fV}{S}{\fVV}}{\seqE{\bV\subC{e}{x}}{\fV\subC{e}{x}}{S\subC{e}{x}}{\fVV\subC{e}{x}}}
\quad\;
   \inference[\rtit{lSub}]{\bV \models x = e \quad \seqEB{\fV\subC{e}{x}}{S\subC{e}{x}}{\fVV\subC{e}{x}}}{ \seqEB{\fV}{S}{\fVV} }\\[12pt]
& 
 \inference[\rtit{lImp}]
     {\seqE{\bVV}{\fV_1}{T}{\fVV_1}\\ 
\bV\models\bVV\quad \fV\models\fV_1\quad
     S\steq T \quad \fVV_1\models\fVV }
     {\seqEB{\fV}{S}{\fVV}} 
\qquad \qquad
\inference[\rtit{lRen}]
{d\not\in\fn{\env,\fV,\fVV,S} \\ \seqEB{\fV}{S}{\fVV}}
{\seqB{\env\subC{d}{c}}{\fV\subC{d}{c}}{S\subC{d}{c}}{\fVV\subC{d}{c}}}
\\
 \end{align*}
\end{small_display}

\medskip

The proof system, defined by the rules in \figref{fig:proof-system-1},
 assumes the  derivation
 judgement $\bV_1\models\bV_2$ between two (possibly open) boolean formulas, with the expected property that
   \begin{align*}
     \forall \sigma:\Subs.\;  \bV_1\models\bV_2 \;\text{ and }\; \bV_1\sigma\evaluate\boolT\quad\text{implies}\quad \bV_2\sigma\evaluate\boolT
   \end{align*}
 Most of the logical rules are rather intuitive and their `naturality' is, in part, due to the strong substratum provided by process confinement, in terms of absence of races. 
We have four logical axioms where \rtit{lNil}, \rtit{lBlk} and \rtit{lOut} deal with stable systems.  More precisely,  \rtit{lNil} acts as a wire between the precondition and the postcondition,  \rtit{lFls}  trivialises proofs with an unsatisfiable boolean condition, \rtit{lBlk} generates input-blocked process assertions, and \rtit{lOut} generates data assertions.     

The rule \rtit{lIn}   is central to the proof system as it is the only rule that consumes part of the precondition.     Together with \rtit{lOut} and \rtit{lPar} they capture process communication in our proof system.  In particular, they observe the permission mutual-exclusion invariants dictated by the environment, whereby the side-condition  in \rtit{lOut}, \ie  $\env(c)\subseteq\eV$,  forces outputs to own the permissions guarded by the mutual exclusion through the side-condition $\env(c)\in\eV$, whereas the premise in \rtit{lIn} permit inputs to assume ownership of  these guarded permissions after communication, through the masking of these permissions in the conclusion,  \ie  $\eV\setminus\env(c)$.  The permission checking side-conditions in the axioms \rtit{lOut} and \rtit{lBlk} ensure that stable systems are safe; similarly, the permission checking side-condition in \rtit{lIn} ensures that evaluations are also safe - recall that any permission violation is propagated down to the eventual stable system by Lemma~\ref{lem:violation-pres}.

The system  parallel composition rule (\rtit{lPar}) is central to our proof system. It is the only rule that allows us to introduce a cut-middle formula in the hypotheses, $\fV_3$. The asymmetry in the hypotheses of this rule guarantees the existence of a reduction sequence across two independently verified sub-systems since the unidirectional cut disallows mutual dependencies across the premise sequents; this prevents deadlocks and ensures \emph{total} correctness. 
\rtit{lPar}  also carries two side-conditions, $\sepprocE{\fVV_1}{\fVV_2}$ and $\sepprocE{\fV_2}{\fV_3}$, denoting  formula separation, defined in \defref{def:fromula-sep}.

The proof system also has a  rule for process parallel composition, (\rtit{lSpl}), which forces a partitioning of permission-resources, analogously to \rtit{cSpl} from \figref{fig:confined-lang}; similarly, the  process scoping rule (\rtit{lLcl}) follows rule \rtit{cLcl} from \figref{fig:confined-lang}. The
system scoping rule  (\rtit{lRes}) restricts the permission-guarding invariants relating to the scoped channels and  filters assertions blocked by the scoping  using the function $\fVV \setminus\lst{c}$, as defined in \defref{def:formula-rest} ; in particular this function over-approximates to \fany\ any message state assertions and  input-blocked assertions affected by the name scoping of the restriction.   \rtit{lRes} also uses an environment restriction operation $\env\setminus c$ defined in \defref{def:env-rest}.

\begin{defi}[Formula Restriction]\label{def:formula-rest}\rm
  \begin{align*}
   \fV\setminus\lst{c} & \quad\deftxt\quad
   \begin{cases}
     \fstate{d}{\lst{e}} &\text{if } \fV = \fstate{d}{\lst{e}} \text{ and } d\not\in\lst{c}\\
     \fblk{d} &\text{if } \fV = \fblk{d} \text{ and } d\not\in\lst{c}\\
     \femp   &\text{if } \fV = \femp 
     \\
     \fconsBB{\fV_1\setminus \lst{c} }{\fV_2\setminus \lst{c} } &\text{if } \fV = \fcons{\fV_1}{\fV_2}\\
     \fany
     & \text{otherwise}
   \end{cases}
 \end{align*}
\end{defi}

\begin{defi}[Environment Restriction]\label{def:env-rest}\rm
  \begin{align*}
    \env \setminus c & \deftxt
    \begin{cases}
      \emptyset & \text{if } \env = \emptyset\\
      \env'\setminus c & \text{if } \env=\env',\emap{c}{\eV}\\
       (\env'\setminus c),\emap{d}{(\eV\setminus \sset{\permin{c},\permout{c}}} & \text{if } \env=\env',\emap{d}{\eV} \text{ and } c\neq d
    \end{cases}
  \end{align*}
\end{defi}

\begin{prop} \label{prop:env-well-formed}
  If $\env$ is a permission environment then $\env\setminus c$ is as well.
\end{prop}

\proof It is immediate to check that  \defref{def:permission-environment} is still observed by $\env\setminus c$, in particular that it is suitably closed (\defref{def:permission-environment}.2).

The remaining logical rules are fairly straightforward. In the conditional proof rule \rtit{lIf}, the hypotheses on each
 branch are augmented with the corresponding assertion, as usual in Hoare logics; this mechanism works in pairs with the structural rule \rtit{lFls} 
 which trivialises the proof obligations on unreachable branches. \rtit{lDef} completes the treatment of the logical rules in the obvious way. Note that rules \rtit{lIn} and \rtit{lDef} abuse the substitution notation, extending it from values to (possibly open) expressions.

The proof system  also
has  a number of structural rules.    The rule  (\rtit{lInst}) permits instantiations of generic sequents whereas (\rtit{lSub}) permits substitutions of expressions to variables that can be inferred to be equivalent from the sequent boolean expression. The rule \rtit{lRen} renames channel names in sequents; the rule side-condition guarantees that the name $d$ is fresh which make renaming injective.    Finally, (\rtit{lImp}) endows proofs with a basic understanding  of structural  equivalence, $\steq$, and of logical   implication, $\models$.


\subsection{Derived Rules}
\label{sec:derived-rules}

 Although \rtit{lPar} is used extensively when proving properties of parallel communicating processes, it turns out that we often do not require its full power which makes it somewhat cumbersome to use.
 We therefore derive lightweight versions  of \rtit{lPar}, enabling parallel code to be either logically sequenced thereby focussing on cutting intermediary 
  formulas (\rtit{lCut}),
  or else considered totally separate, where composite  pre-conditions are assumed to produce composite post-conditions (\rtit{lSep}).  
 These derived rules require fewer side-conditions relating to formula separation. For instance, \rtit{lCut} disposes of the side-conditions entirely, and \rtit{lSep} limits them to one check. 
\begin{align*}
&\rule{2pc}{0pt}
 \inference[\rtit{lCut}]
    {\seqEB{\fV_1}{S}{\fVV} \\ 
     \seqEB{\fVV}{T}{\fV_2}}
  {\seqEB{\fV_1}{S\paralS T}{\fV_2}}
   \quad\;
  \inference[\rtit{lSep}]
     {\seqEB{\fV_1}{S}{\fVV_1} & \phantom{\sepprocE{\fVV_1}{\fVV_2}}\\
    \seqEB{\fV_2}{T}{\fVV_2}  & \sepprocE{\fVV_1}{\fVV_2}}
    {\seqEB{\fcons{\fV_1}{\fV_2}}
    	{S\paralS T}{\fcons{\fVV_1}{\fVV_2}}}
 \end{align*}
For state formula pre and postconditions, an even simpler version \rtit{lSep} is obtained by Corollary. \ref{cor:sep-formulas}, \ie \rtit{lSepSt}, which requires no side-conditions at all. 
\begin{align*}
& \inference[\rtit{lSepSt}]
     {\seqEB{\fVVV_1}{S}{\fVVVV_1} &\;\;
    \seqEB{\fVVV_2}{T}{\fVVVV_2}  }
    {\seqEB{\fcons{\fVVV_1}{\fVVV_2}}
    	{S\paralS T}{\fcons{\fVVVV_1}{\fVVVV_2}}}
 \end{align*}

\noindent The derivations of these lightweight parallel rules are straightforward and use formula semantic implications from Lemma \ref{lem:formula-equiv} together with properties for formula separation from Proposition \ref{cor:sep-formulas}; See Appendix~\ref{sec:proof-system-app}.

\medskip
The output axiom rule \rtit{lOut} appears frequently in most derivations using our proof system.  We find it convenient to formulate another derived rule that facilitates comparisons between the expression outputted by the process and that specified by the state formula, even when these expressions do not syntactically match.
\begin{align*}
  &   \inference[\rtit{lOutD}]
    { \bV \models \beq{\lst{e_1}}{\lst{e_2}} \quad 
     \env(c)\subseteq\eV}
    {\seqEB{\femp}{\conC{\pioutA{c}{\lst{e_1}}}}{\fstate{c}{\lst{e_2}}}}
 \end{align*}
Dually, the rule \rtit{lIn} is used  frequently to dispose of cut-formulas. However the direct use of this rule can become unwieldy due to necessary system structural manipulations required to get the system in form required by the rule.    A more convenient version can be derived that abstracts away from structural equivalence manipulations.
\begin{align*}
  & \inference[\rtit{lInD}] 
     {T \steq \con{\eV\setminus\env(c)} {\piin{c}{\lst{x}}{P}} \paral S & \permin{c}\in\eV & \seqEB{\fV} {\conC{P\subC{\lst{e}}{\lst{x}}}\paral S} {\fVV} }
     {\seqEB{\fcons{\fV}{\fstate{c}{\lst{e}}}} { T } {\fVV}}
\end{align*}
The proofs for these derived rules are straightforward and relegated to Appendix~\ref{sec:proof-system-app}.

Derived rules similar to \rtit{lIn} can be obtained for \rtit{lDef}, \rtit{lIf} \rtit{lSpl} and \rtit{lLcl} using an analogous derivation.  In \secref{sec:application} we shall often abuse this fact and use the derived rule named as the respective proof rule while at the same time abstracting away from structural manipulations.


\subsection{Frame Rule} \label{rem:frame}\rm  The frame rule embodies local reasoning in separation-based logics~\cite{Reynolds02}. 
 For satisfiable post-conditions, a variant of the frame rule  
can be derived in our proof system. 
  \begin{align*}
       &  \inference[\rtit{lFrm}]
  {\seqEB{\fV_1}{S}{\fV_2} 
   & \sepprocE{\fV_2}{\fVV}}
  {\seqEB{\fcons{\fV_1}{\fVV}}{S}{\fcons{\fV_2}{\fVV}}}
 \end{align*} 
Moreover, for the special case when the pre and post conditions are state formulas, the frame rule eliminates the need for the side condition as stated below.
  \begin{align*}
       &  \inference[\rtit{lFrmSt}]
  {\seqEB{\fVVVV_1}{S}{\fVVVV_2} 
   }
  {\seqEB{\fcons{\fVVVV_1}{\fVVV}}{S}{\fcons{\fVVVV_2}{\fVVV}}}
 \end{align*} 
We here show the derivation for the more general version of frame rule, \ie \rtit{lFrm}, using the proof rules (\rtit{lNil}), 
  (\rtit{lPar}) and (\rtit{lImpl}) and the structural rule $S\paralS\con{\emptyset}{\inert}\steq S$.
\begin{align*}
  & \infer[\rtit{lImp}]{\seqEB{\fcons{\fV_1}{\fVV}}{S}{\fcons{\fV_2}{\fVV}}}{
        \infer[\rtit{lSep}]{\seqEB{\fcons{\fV_1}{\fVV}}{S\paral \con{\emptyset}{\inert}}{\fcons{\fV_2}{\fVV}} \qquad S \steq S\paral \con{\emptyset}{\inert}}{
          \seqEB{\fV_1}{S}{\fV_2}  &\qquad
          \infer[\rtit{lNil}]{\seqEB{\fVV}{\con{\emptyset}{\inert}}{\fVV}}{}  &\qquad \sepprocE{\fV_2}{\fVV}
        }
      }
\end{align*}

\bigskip

\noindent Our proof-system is sound with respect to \defref{defn:sequent-satisfaction}.

\begin{thm}[\rm\textit{Soundness}]\label{thm:logic-soundness}
  \begin{math}
    \quad\seqEB{\fV}{S}{\fVV}\quad
    	\text{implies}\quad\seqsemEB{\fV}{S}{\fVV}\,.
  \end{math}
\end{thm}

\proof
  By rule induction on \seqEB{\fV}{S}{\fVV}.  We here show the main rules:
  \begin{desCription}
  \item\noindent{\hskip-12 pt\bf\rtit{lOut}:}\ \ For arbitrary $\sigma,\, T$ we have:
    \begin{align}
      \label{eq:39}
      &\bV\sigma\evaluate\boolT \\
      \label{eq:40}
      & \env, T\sigma \;\sat\; \femp\sigma \\
      \label{eq:41}
      & T\sigma \perp \conC{\pioutA{c}{\lst{e}}}\sigma
    \end{align}
    and the side-condition
    \begin{align}
      \label{eq:87}
      \env(c)\subseteq\eV
    \end{align}
    By \figref{fig:assertion-satisfaction}  and \eqref{eq:40} we know
    \begin{align}
     & T\sigma\evaluate \con{\emptyset}{\inert}\label{eq:60a}
    \end{align}
    By \eqref{eq:41} we know that $T\sigma\paral\conC{\pioutA{c}{\lst{e}}}\sigma$ is well-resourced.  Moreover, by \eqref{eq:60a}  and \rtit{cPar} and \rtit{scNil} of \figref{fig:confined-lang} we deduce
    \begin{align}\label{eq:88}
      &T\sigma\paral\conC{\pioutA{c}{\lst{e}}}\sigma \reduc^\ast
      \con{\emptyset}{\inert}\paral\conC{\pioutA{c}{\lst{e}}}\sigma
      \steq \conC{\pioutA{c}{\lst{e}}}\sigma
    \end{align}
    Clearly,  $\conC{\pioutA{c}{\lst{e}}}\sigma \reducNot$.  Moreover by the conditions imposed on environment mappings in \defref{def:permission-environment}, we know $\permout{c}\in \env(c)$ and thus by \eqref{eq:87} we deduce that $\permout{c}\in \eV$ and hence that $\conC{\pioutA{c}{\lst{e}}}\sigma \errNot$.   As a result, from \eqref{eq:88} we obtain $T\sigma\paral\conC{\pioutA{c}{\lst{e}}}\sigma \,\evaluate\, \conC{\pioutA{c}{\lst{e}}}\sigma$ and for some $\lst{v}$ where $\lst{e}\sigma\evaluate \lst{v}$ and by \eqref{eq:87} and \figref{fig:assertion-satisfaction} we obtain  ${\env,\;(T\paral\conC{\pioutA{c}{\lst{e}}})\sigma\;\sat\;(\fstate{c}{\lst{e}}\sigma)}$.

  \item\noindent{\hskip-12 pt\bf \rtit{lIn}:}\ \ For arbitrary $\sigma,\, T$ we have:
    \begin{align}
      \label{eq:45}
      &\bV\sigma\evaluate\boolT \\
      \label{eq:46}
      &\env, T\sigma \satS (\fcons{\fV}{\fstate{c}{\lst{e}}})\sigma \\
      \label{eq:47}
      &  T\sigma \perp (\con{\eV\setminus\env(c)}{\piin{c}{\lst{x}}{P}}\paral S)\sigma
    \end{align}
    and the side-condition
    \begin{align}
      \label{eq:89}
      \permin{c}\in\eV
    \end{align}
    By \eqref{eq:46} and \figref{fig:assertion-satisfaction} we know
    \begin{align}
       \label{eq:48}
      & T\evaluate \restB{\lst{d}}{T_1 \paralS T_2} \\
      \label{eq:90}
      \text{ where }& \lst{d}\not\in \dom{\env} \\
      \label{eq:60}
      &\env,T_1\sat\fV\sigma \\
      \label{eq:91}
      \text{ and } &\env,T_2\sat \fstate{c}{\lst{e}}\sigma
    \end{align}
   By $\env,T_2\sat \fstate{c}{\lst{e}}\sigma$ and \figref{fig:assertion-satisfaction} we know
   \begin{align}\label{eq:49}
    & T_2\evaluate\conCC{\pioutA{c}{\lst{e'}}}\text{ where }\lst{e}\sigma \evaluate \lst{v}, \lst{e'}\evaluate \lst{v} \text{ and } \env(c)\subseteq\eVV 
   \end{align}
   By \eqref{eq:47} we know $T\sigma \perp (\con{\eV\setminus\env(c)}{\piin{c}{\lst{x}}{P}}\paral S)\sigma$ is well-resourced and by \eqref{eq:90} and $\env(c)\subseteq\eVV$ of \eqref{eq:49} we know that $c\not\in\lst{d}$ and that $\lst{d}\not\in\nm{\eVV}$.  Thus by \eqref{eq:48},~\eqref{eq:49} and \rtit{cPar}, \rtit{cCom}, (\ref{eq:90}) and \rtit{scExt} we obtain
   \begin{align}
     \label{eq:50}
     &T\sigma \paralS (\con{\eV\setminus\env(c)}{\piin{c}{\lst{x}}{P}}\paral S)\sigma \;\reduc^\ast\steq\; \restB{\lst{d}}{T_1}\paralS \conCC{\pioutA{c}{\lst{e'}}} \paralS (\con{\eV\setminus\env(c)}{\piin{c}{\lst{x}}{P}}\paral S)\sigma \\
     \label{eq:1}
     &\restB{\lst{d}}{T_1}\paralS \conCC{\pioutA{c}{\lst{e'}}} \paralS (\con{\eV\setminus\env(c)}{\piin{c}{\lst{x}}{P}}\paral S)\sigma \;\reduc\; \rest{\lst{d}}{T_1}\paralS  (\con{\eV\setminus\env(c)}{\piin{c}{\lst{x}}{P}}\paral S)\sigma
   \end{align}
   By \eqref{eq:50} and Lemma~\ref{lem:resourcing} we know that $\restB{\lst{d}}{T_1}\paralS (\con{\eV\setminus\env(c)}{\piin{c}{\lst{x}}{P}}\paral S)\sigma $ is well-resourced, and by  $\env(c)\subseteq\eVV$ of\eqref{eq:49} we deduce that 
    \begin{align}
      &\rest{\lst{d}}{T_1} \perp (\conC{P\subC{\lst{e}}{\lst{x}}}\paral S)\sigma \label{eq:92}
    \end{align}
By \eqref{eq:60},  (\ref{eq:90}) and Lemma~\ref{lem:sat-rest} we obtain
\begin{align*}
  \env, \rest{\lst{d}}{T_1} \sat \fV\sigma 
\end{align*}
and thus by\eqref{eq:45} ,~\eqref{eq:92}, the premise  $\seqEB{\fV} {\conC{P\subC{\lst{e}}{\lst{x}}}\paral S} {\fVV}$ and I.H. we obtain
  \begin{align}\label{eq:93}
    &\env,\;\rest{\lst{d}}{T_1}\paralS(\conC{P\subC{\lst{e}}{\lst{x}}}\paral S)\sigma\satS \fVV\sigma
  \end{align}
By $\lst{e}\sigma\evaluate \lst{v}$ of \eqref{eq:49} and Lemma~\ref{lem:subs} we get   
  \begin{math}
    \env,\;\rest{\lst{d}}{T_1}\paralS (\conC{P\subC{\lst{v}}{\lst{x}}}\paral S)\sigma\satS \fVV\sigma
  \end{math}.    Moreover by Lemma~\ref{lem:rest-satisfaction} we also obtain 
  \begin{displaymath}
    \env,\;\rest{\lst{d}}{T_1}\paralS (\con{\eV\cup\eVV}{P\subC{\lst{v}}{\lst{x}}}\paral S)\sigma\satS \fVV\sigma
  \end{displaymath}
Thus by \eqref{eq:50},~\eqref{eq:1} and Proposition~\ref{cor:satisfaction-convergence} we obtain $\env,\;T\sigma \paralS (\con{\eV\setminus\env(c)}{\piin{c}{\lst{x}}{P}}\paral S)\sigma\satS \fVV\sigma$ as required.

\item\noindent{\hskip-12 pt\bf \rtit{lPar}:}\ \ For arbitrary $\sigma,\, R$ we have:
  \begin{align}
    \label{eq:76}
    & \bV\sigma\evaluate\boolT \\
    \label{eq:77}
    &  \env,\, R\sigma \satS (\fcons{\fV_1}{\fV_2})\sigma \\
    \label{eq:78}
    &R\sigma \,\perp\, S\sigma\paral T\sigma
  \end{align}
and side-conditions
\begin{align}
  \label{eq:100}
  & \sepprocE{\fV_2}{\fV_3} \\
  \label{eq:101}
  & \sepprocE{\fVV_1}{\fVV_2}
\end{align}
By \eqref{eq:77} we know
\begin{align}
  \label{eq:79}
  & R\sigma\,\evaluate\, \restB{\lst{c}}{R_1\paral R_2} \\
  \label{eq:7}
  \text{ where } & \lst{c}\not\in\dom{\env}\\
  \label{eq:43}
  &\env,\, R_1\satS \fV_1\sigma \\
  \label{eq:44}
  \text{ and } &\env,\, R_2\satS \fV_2\sigma
\end{align}
By  \eqref{eq:79},~\eqref{eq:78} and Lemma~\ref{lem:resourcing}  we know
\begin{align}
  \label{eq:80}
  & R_1\perp R_2 \\
  \label{eq:94}
  \quad\text{ and }& R_1 \perp  S\sigma\paral T\sigma \\
  \label{eq:95}
  \quad \text{ and }&  R_2 \perp  S\sigma\paral T\sigma
\end{align}  
By \eqref{eq:76}, \eqref{eq:43}, $R_1 \perp  S\sigma$ from \eqref{eq:94} and I.H. we have $\env,\, R_1\paral S\sigma\satS (\fcons{\fVV_1}{\fV_3})\sigma$ and from the satisfaction definition of  \figref{fig:assertion-satisfaction} we obtain
\begin{align}
  \label{eq:81}
  &R_1\paral S\sigma\, \evaluate\, \restB{\lst{d}}{S_1\paral S_2} \\
  \label{eq:54}
  \text{ where }& \lst{d}\not\in\dom{\env}\\
  \label{eq:96}
  &\env,\, S_1\satS \fVV_1\sigma \\
  \label{eq:97}
  \text{ and } &\env,\, S_2\satS \fV_3\sigma
\end{align}
By \eqref{eq:80} and \eqref{eq:95} we know $R_1\paral S\sigma \perp R_2$.  Thus, by  \eqref{eq:81}  and Lemma~\ref{lem:locality}  we derive $S_1 \perp R_2$, and by \eqref{eq:44},  \eqref{eq:97}, the rule side-condition  \eqref{eq:100} and Lemma~\ref{lem:merge-assert} we obtain
\begin{align}\label{eq:98}
  &\env,\,R_2\paral S_2 \satS (\fcons{\fV_2}{\fV_3})\sigma
\end{align}
Using \eqref{eq:80} and \eqref{eq:81} we can also derive $R_2\paral S_2 \perp T\sigma$ and by \eqref{eq:76},~\eqref{eq:98} and I.H. we derive
\begin{align}\label{eq:99}
  &\env,\,R_2\paral S_2 \paral T\sigma \satS \fVV_2\sigma
\end{align}
By  \eqref{eq:80} and \eqref{eq:81} we also derive $S_1 \perp (R_2\paral S_2 \paral T)$ and by the rule side-condition \eqref{eq:101} and Lemma~\ref{lem:merge-assert} we obtain
\begin{align*}
  &\env,\,S_1\paral R_2\paral S_2 \paral T\sigma \satS (\fcons{\fVV_1}{\fVV_2})\sigma
\end{align*}
Thus by \eqref{eq:7},~\eqref{eq:54} and Lemma~\ref{lem:sat-rest} we deduce
\begin{align}\label{eq:82}
  &\env,\,\restB{\lst{c},\lst{d}}{S_1\paral R_2\paral S_2 \paral T}\sigma \satS (\fcons{\fVV_1}{\fVV_2})\sigma
\end{align}
From  \eqref{eq:79}, \eqref{eq:81}, \rtit{cPar}, \rtit{cRes}, \rtit{cStr} and \rtit{scExt} we derive $$R\sigma\paral S\sigma\paral T\sigma \reduc^\ast\steq \restB{\lst{c}}{R_1\paral R_2\paral S\sigma\paral T\sigma} \reduc^\ast\steq  \restB{\lst{c},\lst{d}}{S_1\paral R_2\paral  S_2 \paral T\sigma} $$ and by \eqref{eq:82}, Proposition \ref{cor:satisfaction-convergence} and Proposition~\ref{lemma:proc-struct-eq} we obtain
\begin{math}
  \env,\,R\sigma\paral S\sigma \paral T\sigma \satS (\fcons{\fVV_1}{\fVV_2})\sigma
\end{math}
 as required.
 \qed
\end{desCription}

\bigskip

\subsection{Process Sequent Satisfaction}
\label{sec:process-sequents}

 We conclude this section with \defref{def:proc-sat}, which extends sequent satisfaction to processes by assuming the \emph{existence} of a permission
environment  and the respective permission-set, required by the satisfaction definition
of  \figref{fig:assertion-satisfaction}.  This allows for the possibility of having multiple narratives explaining determinism,  and is
in line with the \emph{``ownership is in the eye of the asserter''} principle \cite{Ohearn04}.

\begin{defi}[Process Sequent Satisfaction]\label{def:proc-sat} \rm
  \begin{align*}
    \seqsemNE{\bV}{\fV}{P}{\fVV}  &\; \deftxt\; \text{exists } \env,\,\eV \text{ such that } \;\;\seqsemEB{\fV}{\conCP}{\fVV}
  \end{align*}
\end{defi}

\begin{exa}\label{ex:proc-sat} According to \ref{def:proc-sat}, we can now state that \ptit{Prg}, from \exref{ex:running} satisfies the property
  \begin{align} \label{eq:3}
     &\seqsemNES{x\leq 9}{\fcons{\fstate{c_1}{x}}{\fstate{c_2}{y}}}{\ptit{Prg}}{\fcons{\fstate{c_1}{x,2x}}{\fstate{c_4}{}}},
  \end{align}
  while abstracting over the narrative  as to why $\ptit{Prg}$ is deterministic.  It can be read as saying that, given two values $x$ and $y$ on channels $c_1$ and $c_2$ respectively, \ptit{Prg} returns the value of $x$ together with its double on $c_1$ and a signal on $c_4$, provided that the value of $x$ is less than $10$.  Mirroring the previous discussion in  \exref{ex:running}, \ptit{Prg} also satisfies the property
  \begin{align} \label{eq:4}
     &\seqsemNES{x > 9}{\fcons{\fstate{c_1}{x}}{\fstate{c_2}{y}}}{\ptit{Prg}}{\fcons{\fstate{c_4}{x}}{\fany}},
  \end{align}
  where \fany\ abstracts over the blocked code \restB{c_3}{\piin{c_3}{x_4}{\pioutA{c_1}{(x_4\!+\!x_4)}}}, as described earlier in \exref{ex:satisfiability}.
\end{exa}

\newcommand{\vord}[1]{\textbf{ord}(#1)}
 \newcommand{\veq}[2]{#1 \doteq #2}

We are also in a position to specify the correctness of our quicksort algorithm  through some macro definitions for compactness.

\begin{exa}[Specifying Correctness for Parallel Quicksort]\label{ex:quicksort-spec}    The expected behaviour of  \quick{i,j} from \exref{ex:quicksort} can be expressed through the sequent satisfaction
   \begin{align}\label{eq:58}
     \seqsemNESS{\left(\vord{\vlst{y}{i}{j}} \wedge \veq{\vlst{x}{i}{j}}{\vlst{y}{i}{j}}\right)}{\arrX{i}{j}{i}{j}} {\quick{i,j}} {\fcons{\arr{i}{j}{y}{i}{j}}{\fstate{r}{}}}
   \end{align}
  using the following macro definitions, whereby \vlst{x}{i}{j} denotes lists of variables $x_i\ldots x_j$ when $i \leq j$ and  the empty-list $\epsilon$ otherwise
   \begin{align*}
    \arrX{i}{j}{i}{j}  &\deftxt
  \begin{cases}
    \femp &\text{if }i > j\\
    \fcons{\fstate{a_i}{x_i}}{\arrX{i+1}{j}{i+1}{j}} &\text{if }i\leq j
  \end{cases}\\
     \vord{\vlst{x}{i}{j}}
     & \deftxt
     \begin{cases}
       \btrue
       &\text{if } i = j\\
       x_i \leq x_{i+1} \wedge   \vord{\vlst{x}{i+1}{j}}
       &\text{if } i < j
     \end{cases} \\
      \veq{\vlst{x}{i}{j}}{\vlst{y}{i}{j}} & \deftxt
     \begin{cases}
        \btrue &\text{if } i=j\\[5pt]
        {\displaystyle \bigvee_{i\leq k\leq j}}   \bigl(\vlst{y}{i}{j}= \vlst{y}{i}{k-1} x_i\, \vlst{y}{k+1}{j}\bigr) \wedge \bigl(\veq{\vlst{x}{i+1}{j}}{\vlst{y}{i}{k-1}\vlst{y}{k+1}{j}}\bigr)  &\text{if }  i < j
       \end{cases}  
   \end{align*}
    The specification of \eqref{eq:58} above  states that when \quick{i,j} is composed with an array of arbitrary values on channels $a_i\ldots a_j$, denoted by the assertion macro  \arr{i}{j}{x}{i}{j}, it returns another array of values on the same channel list, \arr{i}{j}{y}{i}{j}, together with a signal on channel $r$ denoting completion. Moreover,  the values returned are 
   \begin{enumerate}[(1)]
   \item ordered, expressed as the predicate $\vord{\vlst{y}{i}{j}}$
   \item equal, up to reordering, to the original values, expressed as the predicate  $\veq{\vlst{x}{i}{j}}{\vlst{y}{i}{j}}$.
     \end{enumerate}
\end{exa}

\section{Application}
\label{sec:application}

We conclude by revisiting the properties stated in  \secref{sec:process-sequents}  and show how our proof-system can be used to prove properties about them.   In \exref{ex:running-proof} we see how proofs about concurrent code are performed by running through only one possible reduction trace, even when other interleavings are possible.  The main appeal of these proofs is however their amenability to \emph{compositionality} as shown in \exref{ex:-quicksort-correct}.  In this example proof, the behaviour of sub-programs is verified in terms of their pre and post conditions only, without any concern towards external interference from other concurrent code.  Independently verified sub-programs are then merged together using \rtit{lPar} (and its variants \rtit{lCut}, \rtit{lSep} and \rtit{lSepSt}), as long as the sub-programs are separate \wrt the permissions that they own.

\begin{exa}[Proving Satisfiability]\label{ex:running-proof}  \rm We  prove the specifications \eqref{eq:3} and \eqref{eq:4} stated earlier in \exref{ex:proc-sat} by first augmenting the satisfaction specification with an appropriate narrative for determinism as stated in \defref{def:proc-sat}.  One possible narrative is the permission-set $\sset{\permin{c_1},\permin{c_2},\permout{c_4}}$ together with the permission-transfer invariants 
   \begin{align*}
\env &= \emap{c_1}{\sset{\permout{c_1}}},\,\emap{c_2}{\sset{\permout{c_2}}},\,
   \emap{c_4}{\sset{\permout{c_4},\permin{c_1}}}      
   \end{align*}
yielding the system specification 
\begin{align} \label{eq:51}
     &\seqsemES{x\leq 9}{\fcons{\fstate{c_1}{x}}{\fstate{c_2}{y}}}{\con{\sset{\permin{c_1},\permin{c_2},\permout{c_4}}}{\ptit{Prg}}}{\fcons{\fstate{c_1}{x,2x}}{\fstate{c_4}{}}}
\end{align}
Another possible narrative is the permission-set $\sset{\permin{c_1},\permin{c_2}}$ and the environment
\begin{align*}
  \env' &=\emap{c_1}{\sset{\permout{c_1},\permout{c_4}}},\,\emap{c_2}{\sset{\permout{c_2}}},\,\emap{c_4}{\sset{\permout{c_4},\permin{c_1}}} 
\end{align*}
yielding a different intensional specification explaining the process determinism below:
\begin{align*} 
     &\seqsemS{\env'}{x\leq 9}{\fcons{\fstate{c_1}{x}}{\fstate{c_2}{y}}}{\con{\sset{\permin{c_1},\permin{c_2}}}{\ptit{Prg}}}{\fcons{\fstate{c_1}{x,2x}}{\fstate{c_4}{}}}
\end{align*}
We here focus on the specification with the first narrative,~\eqref{eq:51},  which by \thmref{thm:logic-soundness}, 
follows from the proof of the sequent
\begin{align} \label{eq:5}
     &\seqES{x\leq 9}{\fcons{\fstate{c_1}{x}}{\fstate{c_2}{y}}}{\con{\sset{\permin{c_1},\permin{c_2},\permout{c_4}}}{\ptit{Prg}}}{\fcons{\fstate{c_1}{x,2x}}{\fstate{c_4}{}}}
\end{align}
 Since $\ptit{Prg} \deftri \restB{c_3}{ \ptit{Fltr} \paral \ptit{Dbl}}$, we prove \eqref{eq:5} by applying the proof rules \rtit{lDef} followed by \rtit{lLcl} and \rtit{lRes}, which leaves us with the following sequent to prove 
\begin{align}\label{eq:6}
  &\seqSS{\env''}{x\leq 9}{\fcons{\fstate{c_1}{x}}{\fstate{c_2}{y}}}{\con{\sset{\permin{c_1},\permin{c_2},\permin{c_3},\permout{c_3}
        ,\permout{c_4}}}{\ptit{Fltr} \paral \ptit{Dbl}}}{\fcons{\fstate{c_1}{x,2x}}{\fstate{c_4}{}}}
\end{align}
where $\env''$ is the extended environment \begin{math}
\env''= \env,\emap{c_3}{\sset{\permout{c_3},\permout{c_1}
  }}
\end{math}.  Note  that, through \rtit{lRes}, in \eqref{eq:6} we have also increased the permissions owned by the system with $\permin{c_3}$ and $\permout{c_3}$, the permissions relevant to the scope of $c_3$, opened by \rtit{lRes}.  Moreover for \rtit{lRes}, the post-condition is unaffected in this case, \ie  according to \defref{def:formula-rest} $(\fcons{\fstate{c_1}{x,2x}}{\fstate{c_4}{}}) \setminus c_3 = \fcons{\fstate{c_1}{x,2x}}{\fstate{c_4}{}}$.  After applying the logical rule \rtit{lSpl}, followed by two applications of \rtit{lDef} for \ptit{Fltr} and \ptit{Dbl} we are left with  
\begin{align*}
  &\seqS{\env''}{x\leq 9}{\fcons{\fstate{c_1}{x}}{\fstate{c_2}{y}}}{
    \left(\begin{array}{l}
      \con{\sset{\permin{c_1},\permout{c_3},\permout{c_4}}}{
        \begin{array}{l}
              \piin{c_1}{x_1}{\cmpA{x_1}{9}{}\\  
                \quad \pioutA{c_3}{x_1}\paralS\piinBB{c_1}{x_3}{\pioutA{c_1}{(x_1,x_3)}\paral \pioutA{c_4}{}} \\
                \qquad\elseA{\pioutA{c_4}{x_1}}}
        \end{array}
}\\
      \;\paralS \quad \con{\sset{\permin{c_2},\permin{c_3}}}{\piin{c_2}{x_2}{\piin{c_3}{x_4}{\pioutA{c_1}{(x_4\!+\!x_4)}}}}
    \end{array}\right)
}{\fcons{\fstate{c_1}{x,2x}}{\fstate{c_4}{}}}
\end{align*}
We proceed by applying \rtit{lIn} twice for $c_1$ and $c_2$ (in any order) and then by applying \rtit{lIf}, which gives us one unreachable branch since $x\leq 9 \wedge \neg(x\leq 9) \Rightarrow \bfalse$; this can be discharged by \rtit{lImpl} and the axiom \rtit{lFls}.  The reachable premise can be proved as follows; we elide the environment and boolean condition from the sequents below as they remain unchanged throughout:
\begin{align*}
  \footnotesize
\infer[\rtit{lSpl}]{
\seqNEBSS{\femp}{
    \con{\sset{\permin{c_1},\permout{c_1},\permout{c_3},\permout{c_4}}}{\pioutA{c_3}{x}\paralS\piinBB{c_1}{x_3}{\pioutA{c_1}{(x,x_3)}\paral \pioutA{c_4}{}} }\paralS
    \con{\sset{\permin{c_2},\permin{c_3},\permout{c_2}}}{\piin{c_3}{x_4}{\pioutA{c_1}{(x_4\!+\!x_4)}}}
}{\fcons{\fstate{c_1}{x,2x}}{\fstate{c_4}{}}}
}{\infer[\rtit{lCut}]{
\seqNEBSS{\femp}{\con{\sset{\permout{c_1},\permout{c_3}}}{\pioutA{c_3}{x}}\paralS \con{\sset{\permin{c_1},\permout{c_4}}}{\piinBB{c_1}{x_3}{\pioutA{c_1}{(x,x_3)}\paral \pioutA{c_4}{}} } \paralS
    \con{\sset{\permin{c_2},\permin{c_3},\permout{c_2}}}{\piin{c_3}{x_4}{\pioutA{c_1}{(x_4\!+\!x_4)}}} }{\fcons{\fstate{c_1}{x,2x}}{\fstate{c_4}{}}}
}{
\infer[\rtit{lOut}]{\seqNEBSS{\femp}{\con{\sset{\permout{c_1},\permout{c_3}}}{\pioutA{c_3}{x}}}{\fstate{c_3}{x}} }{\env''(c_3) 
  \subseteq \sset{\permout{c_1},\permout{c_3}} }
\infer[\rtit{lIn}]{
\seqNEBSS{\fstate{c_3}{x}}{
  \begin{array}{l}
     \con{\sset{\permin{c_1},\permout{c_4}}}{\piinBB{c_1}{x_3}{\pioutA{c_1}{(x,x_3)}\paral \pioutA{c_4}{}} }\\
     \paral\; \con{\sset{\permin{c_2},\permin{c_3},\permout{c_2}}}{\piin{c_3}{x_4}{\pioutA{c_1}{(x_4\!+\!x_4)}}} 
  \end{array}
}{\fcons{\fstate{c_1}{x,2x}}{\fstate{c_4}{}}}
}{
\infer[\rtit{lCut}]{
\seqNEBSS{\femp}{ \begin{array}{l}
     \con{\sset{\permin{c_1},\permout{c_4}}}{\piinBB{c_1}{x_3}{\pioutA{c_1}{(x,x_3)}\paral \pioutA{c_4}{}} }\\
     \paral\; \con{\sset{\permin{c_2},\permin{c_3},\permout{c_2},\permout{c_1},\permout{c_3}}}{\pioutA{c_1}{(x\!+\!x)}} 
  \end{array} 
}{\fcons{\fstate{c_1}{x,2x}}{\fstate{c_4}{}}}
}{
  \begin{array}{l}
 \qquad\qquad\qquad\infer[\rtit{lIn}]{\seqNEBSS{\fstate{c_1}{2x}}{ \con{\sset{\permin{c_1},\permout{c_4}}}{\piinBB{c_1}{x_3}{\pioutA{c_1}{(x,x_3)}\paral \pioutA{c_4}{}} } }{\fcons{\fstate{c_1}{x,2x}}{\fstate{c_4}{}}}}{
\infer[\rtit{lSpl}]{
\seqNEBSS{\femp}{ \con{\sset{\permin{c_1},\permout{c_4},\permout{c_1}}}{\pioutA{c_1}{(x,2x)}\paral \pioutA{c_4}{}}  }{\fcons{\fstate{c_1}{x,2x}}{\fstate{c_4}{}}}
}{
\infer[\rtit{lSepSt}]{
\seqNEBSS{\femp}{ \con{\sset{\permout{c_1}}}{\pioutA{c_1}{(x,2x)}}\paral\con{\sset{\permin{c_1},\permout{c_4}}}{\pioutA{c_4}{}}  }{\fcons{\fstate{c_1}{x,2x}}{\fstate{c_4}{}}}
}{
  \begin{array}{l}
\infer[\rtit{lOut}]{\seqNEBSS{\femp}{ \con{\sset{\permout{c_1}}}{\pioutA{c_1}{x,2x}} }{\fstate{c_1}{x,2x}}}{\env''(c_1) \subseteq \sset{\permout{c_1}}}\\\\
\qquad \vdots   \qquad  \infer[\rtit{lOut}]{\seqNEBSS{\femp}{ \con{\sset{\permin{c_1},\permout{c_4}}}{\pioutA{c_4}{}}  }{\fstate{c_4}{}}}{\env''(c_4) \subseteq \sset{\permin{c_1},\permout{c_4}}}
  \end{array}
}}}\\\\
    \infer[\rtit{lOutD}]{\seqNEBSS{\femp}{\con{\sset{\permin{c_2},\permin{c_3},\permout{c_2},\permout{c_1},\permout{c_3}}}{\pioutA{c_1}{(x\!+\!x)}} }{\fstate{c_1}{2x}}}{x+x = 2x &\env''(c_1) \subseteq \sset{\permin{c_2},\permin{c_3},\permout{c_2},\permout{c_1},\permout{c_3}}} \qquad \quad \vdots\quad\qquad
  \end{array}
}}
}}
\end{align*}
\medskip
Similarly, the proof for the second specification \eqref{eq:4} in \exref{ex:proc-sat} can also be proved by the sequent:
\begin{align} \label{eq:52}
     &\seqES{x > 9}{\fcons{\fstate{c_1}{x}}{\fstate{c_2}{y}}}{\con{\sset{\permin{c_1},\permin{c_2},\permout{c_4}}}{\ptit{Prg}}}{\fcons{\fstate{c_4}{x}}{\fany}},
  \end{align}
The proof is similar to that of \eqref{eq:5}, where we first apply \rtit{lDef}, \rtit{lLcl} and \rtit{lRes}, which leaves us with the following sequent:
\begin{align}\label{eq:53}
  &\seqSS{\env''}{x > 9}{\fcons{\fstate{c_1}{x}}{\fstate{c_2}{y}}}{\con{\sset{\permin{c_1},\permin{c_2},\permin{c_3},\permout{c_3}
        ,\permout{c_4}}}{\ptit{Fltr} \paral \ptit{Dbl}}}{\fcons{\fstate{c_4}{x}}{\fblk{c_3}}}
\end{align}
where, this time, we have the premise postcondition obtained as $\fcons{\fstate{c_4}{x}}{\fblk{c_3}} \setminus c_3 = \fcons{\fstate{c_4}{x}}{\fany}$ according to \defref{def:formula-rest}.  Again, similar to the proof for \eqref{eq:5},  we apply \rtit{lSpl} to \eqref{eq:53}  followed by two applications of \rtit{lDef} for \ptit{Fltr} and \ptit{Dbl}.   Then we apply \rtit{lIn} twice for $c_1$ and $c_2$ to consume the state formula in the precondition, and then by applying \rtit{lIf}.  This time, the rule for conditional gives us a different unreachable branch since $x > 9 \wedge x\leq 9 \Rightarrow \bfalse$.  The reachable premise can be proved as follows:
\begin{align*}
  \footnotesize
\infer[\rtit{lSep}]{
\seqNEBSS{\femp}{
    \con{\sset{\permin{c_1},\permout{c_1},\permout{c_3},\permout{c_4}}}{\pioutA{c_4}{x} }\paralS
    \con{\sset{\permin{c_2},\permin{c_3},\permout{c_2}}}{\piin{c_3}{x_4}{\pioutA{c_1}{(x_4\!+\!x_4)}}}
}{\fcons{\fstate{c_4}{x}}{\fblk{c_3}}}}
{
  \begin{array}{l}
\infer[\rtit{lOut}]{\seqNEBSS{\femp}{
    \con{\sset{\permin{c_1},\permout{c_1},\permout{c_3},\permout{c_4}}}{\pioutA{c_4}{x} }}{\fstate{c_4}{x}}}
{\env''(c_4)\subseteq \sset{\permin{c_1},\permout{c_1},\permout{c_3},\permout{c_4}}} \\\\
\qquad\qquad\vdots\qquad\qquad  
\infer[\rtit{lBlk}]{
\seqNEBSS{\femp}{
    \con{\sset{\permin{c_2},\permin{c_3},\permout{c_2}}}{\piin{c_3}{x_4}{\pioutA{c_1}{(x_4\!+\!x_4)}}}}{\fblk{c_3}}
}{\permin{c_3} \in \sset{\permin{c_2},\permin{c_3},\permout{c_2}}}     
  \end{array}
& 
\begin{array}{l}\\\\\\\\
\sepprocE{\fstate{c_4}{x}}{\fblk{c_3}}  
\end{array}
 }
\end{align*}

\medskip

\end{exa}

 \begin{exa}[Proving Correctness for Parallel Quicksort] \label{ex:-quicksort-correct} \rm 
To prove the correctness property \eqref{eq:58} for \quick{i,j}, as stated in \exref{ex:quicksort-spec}, we choose a narrative where the environment is 
\begin{align*}
  \env & = \envmap{a_i}{\eset{\permout{a_i}}},\;\ldots,\;\envmap{a_j}{\eset{\permout{a_j}}},\;\envmap{r}{\eV(r,i,j)}
\end{align*}
and \quick{i,j} owns the permission set $\eV(r,i,j)$ defined  as 
\begin{align*}
       \eV(x,i,j) & \deftxt \eset{\permout{x},\permin{a_i},\ldots\permin{a_j}}\,.
\end{align*}
The permissions associated with $r$ express the fact that the array can only be read \emph{after} the signal denoting completion is consumed.

\smallskip
We argue, by induction on $n=j-i$  (where $i\leq j$), that if we show that the following sequent holds for arbitrary $i$ and $j$, 
 \begin{align}
   & \seqES{\;\left(\vord{\vlst{y}{i}{j}} \wedge \veq{\vlst{x}{i}{j}}{\vlst{y}{i}{j}}\right)\;}{\arrX{i}{j}{i}{j}} {\;\con{\eV(r,i,j)}{\quick{i,j}}\;} {\fcons{\arr{i}{j}{y}{i}{j}}{\fstate{r}{}}} \label{eq:12}
  \end{align}
this would imply correctness for \quick{i,j} with the above narrative \ie
\begin{align*}
   & \seqsemES{\;\left(\vord{\vlst{y}{i}{j}} \wedge \veq{\vlst{x}{i}{j}}{\vlst{y}{i}{j}}\right)\;}{\arrX{i}{j}{i}{j}} {\;\con{\eV(r,i,j)}{\quick{i,j}}\;} {\fcons{\arr{i}{j}{y}{i}{j}}{\fstate{r}{}}}
\end{align*}
which, by \defref{def:proc-sat}, would prove the satisfaction \eqref{eq:58}.  

For the \emph{base case} of \eqref{eq:12},  \ie $n=0$  assuming $i=j$ as part of the sequent boolean expression, we trivially prove the sequent using \rtit{lIf}, the state frame rule, \rtit{lFrmSt}, and \rtit{lOut} as shown below.  In what follows, we often elide the sequent environment and boolean condition from our proofs.

\begin{align*} \small
  \infer[\rtit{lDef}]{\seqNEB{\arrX{i}{j}{i}{j}} {\;\con{\eV(r,i,j)}{\quick{i,j}}\;} {\fcons{\arr{i}{j}{y}{i}{j}}{\fstate{r}{}}}}{  
   \infer[\rtit{lIf}]{\seqNEB{\arrX{i}{j}{i}{j}} {\;\con{\eV(r,i,j)}{ \cmpC{i= j}{\pioutA{r}{}}{\ldots} \;}} {\fcons{\arr{i}{j}{y}{i}{j}}{\fstate{r}{}}}}{ 
     \infer[\rtit{lSub}]{\seqNEB{\arrX{i}{j}{i}{j}} {\;\con{\eV(r,i,j)}{ \pioutA{r}{}\;}} {\fcons{\arr{i}{j}{y}{i}{j}}{\fstate{r}{}}}}{
        \infer[\rtit{lFrmSt}]{\seqNEB{\arrX{i}{j}{i}{j}} {\;\con{\eV(r,i,j)}{ \pioutA{r}{}\;}} {\fcons{\arr{i}{j}{x}{i}{j}}{\fstate{r}{}}}}{
          \infer[\rtit{lOut}]{\seqNEB{\femp} {\;\con{\eV(r,i,j)}{ \pioutA{r}{}\;}} {\fstate{r}{}}}{} 
        }
     }
     &
     \infer[\rtit{lFls}]{ \seqNEB{\arrX{i}{j}{i}{j}} {\;\con{\eV(r,i,j)}{ \ldots \;}} {\fcons{\arr{i}{j}{y}{i}{j}}{\fstate{r}{}}} }{}
      }
   }
\end{align*}

\smallskip

The \emph{inductive case}, $n\!+\!1=j\!-\!i$,  \ie adding $i < j$ to the sequent boolean expression, assumes that the property holds for all $ m\leq n$, \ie   all $m < j-i$ (the inductive hypothesis), and follows from proving the following two sequents
\begin{align}
 \seqB{\env_1}{\arrX{i}{j}{i}{j}}{\quad\con{\eV(r_3,i,j)}{\prtn{i,j}}\quad}{ 
\begin{array}{l}
\fcons{\arr{i}{p-1}{z}{i}{p-1}}{\fstate{a_p}{y_p}}\\
\quad\fcons{}{\fcons{\arr{p+1}{j}{z}{p+1}{j}}{\fstate{r_3}{p}}}
\end{array}} &\label{eq:2}\\
 \seqB{\env_1}{
  \begin{array}{l}
\fcons{\arr{i}{p-1}{z}{i}{p-1}}{\fstate{a_p}{y_p}}\\
\quad\fcons{}{\fcons{\arr{p+1}{j}{z}{p+1}{j}}{\fstate{r_3}{p}}}
\end{array}
} {\;\con{\eset{\permin{r_3}\permout{r}}}{
    \begin{array}{l}
      \piin{r_3}{x}{\rest{r_1,r_2}{}\\
       \quad          \begin{array}{l}
                    \quick{i,x-1}\subS{r_1}{r} \\
                    \paralS\; \quick{x\!+\!1,j}\subS{r_2}{r}  \\
                    \paralS\; \piin{r_1}{}{\piin{r_2}{}{\pioutA{r}{}}}
                 \end{array}}
    \end{array}
}} {\fcons{\arr{i}{j}{y}{i}{j}}{\fstate{r}{}}}  \quad&\label{eq:59}
\end{align}
where $\env_1$ extends \env with the mapping for $r_3$ \ie $\env_1 =  \env,\envmap{r_3}{\eV(r_3,i,j)} $ and \bV\ is a stronger boolean condition defined as:
\begin{align*}
   \bV & =  \vord{\vlst{y}{i}{j}} \,\wedge\,  
   {\veq{\vlst{x}{i}{j}\!}{\vlst{y}{i}{j}}} \wedge\, \underbrace{\bigl(\veq{\vlst{y}{i}{p-1}\!}{\vlst{z}{i}{p-1}} \wedge\, \veq{\vlst{y}{p+1}{j}\!}{\vlst{z}{p+1}{j}} \bigr)}_{(i)}
   \wedge \underbrace{\bigl(\textstyle \bigwedge_{k=i}^ {p-1} z_k < y_p \bigr)}_{(ii)} 
    \wedge 
 \underbrace{\bigl( \textstyle\bigwedge_{k=p+1}^{j} y_p \leq z_k \bigr)}_{(iii)}
\end{align*}
 It requires intermediary lists of values  $ \small \vlst{z}{i}{p-1}$ and $\small \vlst{z}{p+1}{j}$, returned by partitioning \prtn{i,j}, to be reorderings of the final values $\small \vlst{y}{i}{p-1}$ and $\small  \vlst{y}{p+1}{j}$, $(i)$, that the values in $\small  \vlst{z}{i}{p-1}$ are less than the pivot, $(ii)$,  and also that the values $\small \vlst{z}{p+1}{j}$ are greater than or equal to the pivot, $(iii)$.  

 The proof for sequent \eqref{eq:12} is  derived from \eqref{eq:2} and \eqref{eq:59} by applying the derived rule \rtit{lCut} which logically sequentialises the two systems; then we apply \rtit{lInst} to substitute $\small \vlst{y}{i}{p-1}\vlst{y}{p+1}{j}$ for $\small \vlst{z}{i}{p-1}\vlst{z}{p+1}{j}$  in $\bV$ (notice that the substitution leaves  the pre/post-conditions and  the system unchanged as $\small \vlst{z}{i}{p-1}\vlst{z}{p+1}{j}$ are not free in them), then \rtit{lImpl} to recover the boolean condition $\left(\vord{\vlst{y}{i}{j}} \wedge \veq{\vlst{x}{i}{j}}{\vlst{y}{i}{j}}\right)$, then \rtit{lRes} to recover $\env$ from $\env_1$, and finally \rtit{lLcl} and \rtit{lDef} to recover $\con{\eV(r,i,j)}{\quick{i,j}}$.

The proof of sequent \eqref{eq:59} follows from the following three sequents \eqref{eq:9}, \eqref{eq:10} and \eqref{eq:11} below, where \rtit{lRes} is used to  extend 
 $\env_1$  as
\begin{align*}
   \env_2 = \env_1, \envmap{r_1}{\eV(r_1,i,p-1)}, \envmap{r_2}{\eV(r_2,p+1,j)} 
\end{align*}
 to account for the mappings associated with the channels  $r_1$ and $r_2$.  Notice how this rule allows us to choose the permission association relating to $r_1$ and $r_2$ \emph{dynamically}, depending on the index $p$ returned by the partitioning phase of sequent \eqref{eq:2}.  Such data dependencies normally complicate similar dependency analyses based on type systems such as \cite{TerauchiAiken08, BergerHondaYoshida08}.

\begin{align}
& \seqB{\env_2}{\arr{i}{p-1}{z}{i}{p-1}}{\;\con{\eV(r_1,i,p-1)}{\quick{i,p-1}}\;}{\fcons{\arr{i}{p-1}{y}{i}{p-1}}{\fstate{r_1}{}}} \label{eq:9}\\
& \seqB{\env_2}{\arr{p+1}{j}{z}{p+1}{j}}{\;\con{\eV(r_2,p+1,j)}{\quick{p+1,j}}\;}{\fcons{\arr{p+1}{j}{y}{p+1}{j}}{\fstate{r_2}{}}} \label{eq:10}\\
& \seqB{\env_2}{\fcons{\arr{i}{j}{y}{i}{j}}{\fcons{\fstate{r_1}{}}{\fstate{r_2}{}}}} {\;\con{\eset{\permin{a_p},\permin{r_1},\permin{r_2},\permout{r}}}{\piin{r_1}{}{\piin{r_2}{}{\pioutA{r}{}}}}\;} {\fcons{\arr{i}{j}{y}{i}{j}}{\fstate{r}{}}} \label{eq:11}
\end{align}\vspace{3 pt}

\noindent Sequents \eqref{eq:9} and \eqref{eq:10} follow from the inductive hypotheses.   Sequent \eqref{eq:11} can be easily derived using \rtit{lFrmSt}, which eliminates \arr{i}{j}{y}{i}{j} from the pre and post conditions, and then applying \rtit{lIn} twice for $r_1$ and $r_2$ respectively, followed by applying \rtit{lOut} once for $r$; the two inputs on $r_1$ and $r_2$ would hand over the permissions $\permin{a_i},\ldots,\permin{a_{p-1}}$ and $\permin{a_{p+1}},\ldots,\permin{a_{j}}$ respectively; these are necessary for the output on $r$ to be derived.  

We recover the proof of sequent \eqref{eq:59} as follows. Sequents \eqref{eq:9} and \eqref{eq:10} can be  composed together  as separate parallel code using \rtit{lSep}, and then extended to include \fstate{a_p}{y_p} in the pre and post-conditions using \rtit{lFrmSt}.  This allows us to logically sequence these two systems \emph{before} the system \con{\eset{\permin{a_p},\permin{r_1},\permin{r_2},\permout{r}}}{\piin{r_1}{}{\piin{r_2}{}{\pioutA{r}{}}}} of sequent~\eqref{eq:11}, thereby cutting the pre-condition of this sequent, using \rtit{lCut}.   Then we  scope the two channels $r_1$ and $r_2$ using a combination of \rtit{lRes}, \rtit{lLcl} and \rtit{lSpl}  (which  leaves the pre and post conditions intact since they do not contain any mention of the channels $r_1$ and $r_2$), and finally precede this whole system by an input on $r_3$ using \rtit{lIn}, which adds \fstate{r_3}{p} to the precondition.

This leaves us with only sequent \eqref{eq:2} to prove to complete the main proof. 
This sequent proof follows immediately from  a proof for the  following sequent 
\begin{align}
  & \small \seqB{\env_1}{\small \arrX{i+1}{j}{i+1}{j}}{\con{\eV(r_3,i,j)\cup\eset{\permout{a_i}}}{\prtnlp{\small i,j,x_i,i,i+1}}}{
\begin{array}{l}
\fcons{\arr{i}{p-1}{z}{i}{p-1}}{\fstate{a_p}{y_p}}\\
\quad\fcons{}{\fcons{\arr{p+1}{j}{z}{p+1}{j}}{\fstate{r_3}{p}}}
\end{array}}\label{eq:13}
\end{align}
through one application of \rtit{lIn}, which  reinstates \fstate{a_i}{x_i} in the pre-condition,  then an application of \rtit{lDef} to recover  $\con{\eV(r_3,i,j)}{\prtn{i,j}}$.

 We prove (\ref{eq:13}) by proving the more general sequent 
\begin{align}
  & \small \seq{\env_1}{\bV'}{
    \begin{array}{l}
      \overbrace{\arr{i+1}{q}{w}{i+1}{q}}^{(i)}\\
      \quad\fcons{}{\fcons{\underbrace{\arr{q+1}{c-1}{w}{q+1}{c-1}}_{(ii)}}{\underbrace{\arrX{c}{j}{c}{j}}_{(iii)}}}
    \end{array}
}{\con{\eV(r_3,i,j)\cup\eset{\permout{a_i}}}{\prtnlp{\small i,j,x_i,q,c}}}{
\begin{array}{l}
\fcons{\arr{i}{p-1}{z}{i}{p-1}}{\fstate{a_p}{y_p}}\\
\quad\fcons{}{\fcons{\arr{p+1}{j}{z}{p+1}{j}}{\fstate{r_3}{p}}}
\end{array}
} \label{eq:14}
\end{align}
 where $\bV' = \bandS{\bV}{\bandS{(i\leq q < c \leq j+1)} {\underbrace{\bigl(\veq{\vlst{x}{i+1}{c-1}}{\vlst{w}{i+1}{c-1}}\bigr)}_{(iii)} \wedge \underbrace{\bigl( \bigwedge_{k=i+1}^{q} w_k < x_i \bigr)}_{(i)} \wedge \underbrace{\bigl( \bigwedge_{k=q+1}^{c-1} x_i  \leq w_k \bigr)}_{(ii)}}}$.   

Sequent \eqref{eq:14}  allows us to  stratify every iteration of the traversal, thereby proving the sequent by  induction on $n=(j+1)-c$. At each iteration, $c$, with pivot index $q$ and pivot value $x_i$, \eqref{eq:14} expects a precondition split into 3 parts:  $\small\arr{i+1}{q}{w}{i+1}{q}$ holds processed values that are less than  the pivot 
$x_i$, $(i)$, $\small\arr{q+1}{c-1}{w}{q+1}{c-1}$ holds processed values that are greater than or equal to  the pivot
$x_i$, $(ii)$, and $\small\arrX{c}{j}{c}{j}$ is the part of the array that still needs to be traversed.   Note also that the values preceding the current counter, $\vlst{w}{i+1}{c-1}$, must be equal, up to reordering, of the values already processed $\vlst{x}{i+1}{c-1}$, $(iii)$.      The base case, \ie when $c = j+1$ (and thus $\small \arrX{c}{j}{c}{j}=\arrX{j+1}{j}{j+1}{j}=\femp$), establishes the post-condition in \eqref{eq:14}  whereas the inductive case works up towards the base case, whereby the value comparison at every iteration adds to the ordering information expressed by $\bV'$. Both proof cases use a mixture of rules \rtit{lIn}, \rtit{lOut}, \rtit{lIf} and, \rtit{lSepSt} and \rtit{lCut} in a manner similar to that discussed already above; the details are left for the interested reader.

To obtain \eqref{eq:13} from \eqref{eq:14}, we take $q$ and $c$ to be $i$ and $i+1$ respectively. This case makes the array assertions $\small \arr{i+1}{q}{w}{i+1}{q}$ and $\small \arr{q+1}{c-1}{w}{q+1}{c-1}$ in the precondition of \eqref{eq:14} empty, \ie $\small \arr{i+1}{q}{w}{i+1}{q} =  \arr{q+1}{c-1}{w}{q+1}{c-1}= \arr{i+1}{i}{w}{i+1}{i} = \femp$, which  by Lemma \ref{lem:formula-equiv} and \rtit{lImp},
 leaves us with $\small \arrX{i+1}{j}{i+1}{j}$ \ie the precondition of \eqref{eq:13}.  Moreover, for this case  the boolean expression  $\bV'$ is of the form
 \begin{math}
   \bandS{\bV}{(i\leq i < i+1 \leq j+1)}
 \end{math} which is implied by \bV,  \ie $\bV \models \bV'$.  This means that we can recover $\bV$ for our sequent simply by applying \rtit{lImp} as well. 
\end{exa}

\section{Conclusion}
\label{sec:conclusion}

We have developed a logic for deterministic processes,   interpreted over systems whose behaviour is confined by sets of linear permissions.    We also developed a sound proof system through which we can determine, in compositional fashion,  the satisfaction of formulas in this logic. 
We applied this logic and proof system to specify and prove the correctness of an in-place parallel quicksort. 

\subsection{Related Work}
\label{sec:related-work}

Modal logics have traditionally been used in process calculi for the specification
of behavioural properties. Proof systems for these logics have been
developed in a variety of settings
(\eg~\cite{Hliu93,AnStWi94,AmDam96, Dam94,
  Dam02,BergerHondaYoshida08})
and some of these have focused on compositional reasoning as a means
of dealing with the scalability problem (\eg \cite{AnStWi94,Dam02,BergerHondaYoshida08}). However, there has been
little focus on locality of reasoning in these efforts. Approaching
compositionality without necessarily modelling locality does seem to
have been at the expense of general, but long-winded proof rules for
parallel composition (e.g.~\cite{Dam02}).    In addition, termination is often not a major focus in these logics; in fact, the bisimulation proof technique, often associated with these logics, is insensitive to divergence.  Termination is central to the logical characterisations that we give in this work. 

Despite the apparent resemblance, spatial logics for process calculi  such as
\cite{Caires01,Caires04} differ from our interpretation of the separating conjunction: we separate on \emph{permissions},  logical embellishments on processes, whereas their logical separation is more intensional and operates on  the structure of processes, describing parallel composition directly.  Moreover, their aims appear to differ from ours since they model mobility and channel privacy; we focus on data, non-interference and locality, and deal with implicit transfer of permissions.

Following \cite{Ohearn04}, the use of separation logic to support local reasoning for concurrent
programs has been studied intensively  over the past few years
for the shared-variable model of concurrency.  The
 initial
main idea of ownership transfer of resources between threads
impacting upon local reasoning already appears in \cite{Ohearn04}.
This was then extended to co-exist with Rely/Guarantee
reasoning  \cite{Vafeiadis07amarriage,Feng} and recently refined through fractional permissions as Deny/Guarantee reasoning \cite{dodds09:denyguarantee}.  The latter is interesting to us as a means of widening
our class of programs under analysis.  For instance, \cite{Gotsman} uses  this approach  for dealing with dynamically allocated resource  locks.

Separation Logic has been applied to process calculi on at least two occasions.  In \cite{HoareOhearn}, they give a separation semantics for a variant of the \pic, based on traces.  Their work differs from ours in a number of respects in that  they only deal with \emph{explicit} ownership transfer of resources and are not concerned with developing a proof system.   In \cite{PymTofts}, they also use a process calculus as a model for a separation logic.  They are quite general \wrt the form of resources and how these are transferred across processes and, as a result, our model of confined processes seems related to theirs.  However, aspects such as the use of SCCS on their part, where processes evolve in synchrony, and the focus on value passing and stability on ours, lead to a substantially different satisfaction relation of the logics.  The aim of their work is also different from ours; they establish a correspondence between strong bisimulation and logic satisfaction whereas we focus on developing a compositional proof system.    Separation logic has also been applied to an imperative concurrent language with message passing in \cite{Villard09} where the main focus is the implementability of message-passing communication as a copy-less communication over a shared memory model.  Although their technical development is considerably different from ours, this work can be seen as complementary to ours if implementation aspects of our language are considered.

\subsection{Future Work}
\label{sec:future-work}

There is much further work to be done in the area of local reasoning
for message-passing concurrency.  

With respect to the work presented here, there are a number of design decisions 
that are worth exploring.  For instance, at the level of the proof system, a partial correctness interpretation of our sequents (as opposed to total correctness) would probably allow us to design a version of the parallel proof rule, \rtit{lPar}, that is more symmetric.  Another avenue worth exploring is that of relaxing the interpretation of our logical assertions so as to not  limit them to safely-stabilising systems.  This would simplify the verification of certain formulas, such as \fany, and would also allow us to have models where formulas such as $\fcons{\fstate{c}{v}}{\fblk{c}}$ are satisfiable.  At the same time, this satisfaction weakening would also entail that our existing assertion interpretation changes to one where systems satisfy a formula at some point during their evaluation but may then fail to satisfy it as computation progresses.  Although it is not yet clear whether this is a desirable property to have from the point of view of the application of the logic, it has appealing benefits in terms of the assertion satisfaction definition,  as it streamlines  the satisfaction of core formulas like the separating conjunction with existing interpretations.  Moreover, we also conjecture that this altered interpretation would eliminate the need for  the side conditions present in the existing parallel rule, \rtit{lPar}.

At a more general level, we also seek to widen the
class of programs we can treat by introducing non-confluent behaviour
in a controlled way. We intend to extend our setting to allow for more
interesting forms of data to be communicated, including  say
channel names.  We also need to develop algorithms for inferring the
permission-set maps, develop tools to support the proof-system reasoning.
 Finally, and perhaps most importantly, we need to
 expand our suite of case studies and
consider larger example proofs.

\section*{Acknowledgment}
  The authors wish to acknowledge numerous referees for their incisive comments on a preliminary version of this paper.


\small
\bibliographystyle{plain}
\bibliography{logic}

\begin{thebibliography}{10}

\bibitem{AmDam96}
R.~Amadio and M.~Dam.
\newblock Toward a modal theory of types for the $\pi$-calculus.
\newblock In {\em FTRTFT}, volume 1135 of {\em LNCS}, 1996.

\bibitem{AnStWi94}
H.~Andersen, C.~Stirling, and G.~Winskel.
\newblock A compositional proof system for the modal $\mu$-calculus.
\newblock In {\em LICS}, volume 4-7, pages 144--153, 1994.

\bibitem{Armstrong07}
J.~Armstrong.
\newblock {\em Programming Erlang}.
\newblock Pragmatic Bookshelf, 2007.

\bibitem{BergerHondaYoshida08}
M.~Berger, K.~Honda, and N.~Yoshida.
\newblock Completeness and full abstraction in modal logics for typed mobile
  processes.
\newblock In {\em ICALP}, volume 5126 of {\em LNCS}, pages 99--111, 2008.

\bibitem{permission}
R.~Bornat, C.~Calcagno, P.~O'Hearn, and M.~Parkinson.
\newblock Permission accounting in separation logic.
\newblock In {\em POPL}, pages 259--270, 2005.

\bibitem{Boyland03}
J.~Boyland.
\newblock Checking interference with fractional permissions.
\newblock In {\em SAS}, volume 2694 of {\em LNCS}, pages 55--72, 2003.

\bibitem{Brookes07}
S.~Brookes.
\newblock A semantics for concurrent separation logic.
\newblock {\em TCS}, 375(1-3):227--270, 2007.

\bibitem{Caires01}
L.~Caires and L.~Cardelli.
\newblock A spatial logic for concurrency (i).
\newblock {\em I\&C}, pages 1--37, 2001.

\bibitem{Caires04}
L.~Caires and L.~Cardelli.
\newblock A spatial logic for concurrency (ii).
\newblock {\em TCS}, 322(3):517--565, 2004.

\bibitem{Calcagno07localaction}
C.~Calcagno, P.~W. O'{H}earn, and H.~Yang.
\newblock Local action and abstract separation logic.
\newblock In {\em LICS}, pages 366--378, 2007.

\bibitem{Calcagno07modularsafety}
C.~Calcagno, M.~Parkinson, and V.~Vafeiadis.
\newblock Modular safety checking for fine-grained concurrency.
\newblock In {\em SAS}, volume 4634 of {\em LNCS}, pages 233--238, 2007.

\bibitem{Dam94}
M.~Dam.
\newblock Model checking mobile processes.
\newblock In {\em CONCUR}, volume 715 of {\em LNCS}, pages 22--36, 1993.

\bibitem{Dam02}
M.~Dam.
\newblock Proof systems for pi-calculus logics.
\newblock In {\em Logic for Concurrency and Synchronisation}, Trends in Logic,
  Studia Logica Library, pages 145--212, 2003.

\bibitem{dodds09:denyguarantee}
Mike Dodds, Xinyu Feng, Matthew Parkinson, and Viktor Vafeiadis.
\newblock Deny-guarantee reasoning.
\newblock In {\em ESOP 2009}, volume 5502 of {\em LNCS}, pages 363--377, 2009.

\bibitem{Feng}
X.~Feng, R.~Ferreira, and Z.~Shao.
\newblock On the relationship between concurrent separation logic and
  assume-guarantee reasoning.
\newblock In {\em ESOP}, pages 173--188, 2007.

\bibitem{Galletly97}
J.~Galletly.
\newblock {\em Occam-2 (2nd ed.): including occam.2.1}.
\newblock UCL Press, 1997.

\bibitem{Girard:1989}
Jean-Yves Girard, Paul Taylor, and Yves Lafont.
\newblock {\em Proofs and types}.
\newblock Cambridge University Press, New York, NY, USA, 1989.

\bibitem{Gotsman}
A.~Gotsman, J.~Berdine, B.~Cook, N.~Rinetzky, and M.~Sagiv.
\newblock Local reasoning for storable locks and threads.
\newblock In {\em APLAS}, volume 4807 of {\em LNCS}, pages 19--37, 2007.

\bibitem{HennessyLin96}
M.~Hennessy and H.~Lin.
\newblock Proof systems for message-passing process algebras.
\newblock In {\em Formal Aspects of Computing}, pages 379--407, 1996.

\bibitem{Hliu93}
M.~Hennessy and H.~Liu.
\newblock A modal logic for message passing processes.
\newblock {\em Act. Inf.}, 32:375--393, 1995.

\bibitem{Hoare78}
C.~A.~R. Hoare.
\newblock Communicating sequential processes.
\newblock {\em Comm. ACM}, 21(8):666--677, 1978.

\bibitem{HoareOhearn}
T.~Hoare and P.~O'Hearn.
\newblock Separation logic semantics for communicating processes.
\newblock {\em ENTCS}, 212:3--25, 2008.

\bibitem{Knuth98}
Donald~E. Knuth.
\newblock {\em The art of computer programming, volume 3: (2nd ed.) sorting and
  searching}.
\newblock Addison Wesley Longman Publishing Co., Inc., Redwood City, CA, USA,
  1998.

\bibitem{KobayashiPT:linearity}
Naoki Kobayashi, Benjamin~C. Pierce, and David~N. Turner.
\newblock Linearity and the pi-calculus.
\newblock {\em ACM ToPLaS}, 21(5):914--947, 1999.

\bibitem{Lynch96}
Nancy~A. Lynch.
\newblock {\em Distributed Algorithms}.
\newblock Morgan Kaufmann, 1996.

\bibitem{Milner89}
R.~Milner.
\newblock {\em Communication and concurrency}.
\newblock Prentice-Hall, 1989.

\bibitem{Milner99}
R.~Milner.
\newblock {\em Communicating and mobile systems: the {$\pi$}-calculus}.
\newblock Cambridge Univ., 1999.

\bibitem{Ohearn04}
P.~W. O'{H}earn.
\newblock Resources, concurrency and local reasoning.
\newblock {\em TCS}, pages 49--67, 2004.

\bibitem{ParkinsonBornat}
M.~Parkinson, R.~Bornat, and P.~W. O'{H}earn.
\newblock Modular verification of a non-blocking stack.
\newblock {\em SIGPLAN Not.}, 42(1):297--302, 2007.

\bibitem{presburger29}
Moj\.{z}esz Presburger.
\newblock \u{U}ber die vollst\u{a}ndigkeit eines gewissen systems der
  arithmetik ganzer zahlen, in welchem die addition als einzige operation
  hervortritt.
\newblock In {\em Comptes Rendus du I congr\`{e}s de Math\'{e}maticiens des
  Pays Slaves}, pages 92--101, 1929.

\bibitem{PymTofts}
D.~Pym and C.~Tofts.
\newblock A calculus \& logic of resources \& processes.
\newblock {\em FAC}, 18:495--517, 2006.

\bibitem{Reppy99}
J.~H. Reppy.
\newblock {\em Concurrent programming in ML}.
\newblock Cambridge University Press, 1999.

\bibitem{Reynolds02}
J.~C. Reynolds.
\newblock Separation logic: A logic for shared mutable data structures.
\newblock In {\em LICS}, pages 55--74, 2002.

\bibitem{SangiorgiWalker01}
D.~Sangiorgi and D.~Walker.
\newblock {\em PI-Calculus: A Theory of Mobile Processes}.
\newblock Cambridge University Press, 2001.

\bibitem{Stirling01}
C.~Stirling.
\newblock {\em Modal and temporal properties of processes}.
\newblock Springer-Verlag, 2001.

\bibitem{TerauchiAiken08}
T.~Terauchi and A.~Aiken.
\newblock A capability calculus for concurrency and determinism.
\newblock {\em TOPLAS}, 30(5):1--30, 2008.

\bibitem{Vafeiadis07amarriage}
V.~Vafeiadis and M.~Parkinson.
\newblock A marriage of rely/guarantee and separation logic.
\newblock In {\em CONCUR}, pages 256--271, 2007.

\bibitem{Villard09}
J.~Villard, {\'E}.~Lozes, and C.~Calcagno.
\newblock Proving copyless message passing.
\newblock In {\em {APLAS}'09}, volume 5904 of {\em LNCS}, pages 194--209,
  Seoul, Korea, 2009.

\end{thebibliography}


\newpage

\appendix


\section{Proofs}
\label{sec:auxilliary-lemmas}

\subsection{Processes}\label{sec:processes}

\begin{lem}[Structural Equivalence and Reductions]\label{lem:steq-reduction-unconstrained}
  \begin{displaymath}
    P\steq Q\text{ and } P \reduc P' \text{ implies } \exists Q'. \ Q \reduc Q' \text{ and }P'\steq Q'
  \end{displaymath}
\end{lem}
\proof  By rule induction on $P\steq Q$.  \qed

\begin{cor}[Structural Equivalence and Reductions]\label{cor:steq-reduction-unconstrained}
  $P\steq Q \text{ and }P\reducNot\text{ implies }Q\reducNot$
\end{cor}

\subsection{Confined Processes}
\label{sec:conf-proc}

\begin{lem}\label{lem:struct-corresp-1}
  \begin{math}
    S\steq T\quad\text{implies}\quad|S|\steq |T|
  \end{math}
\end{lem}
\proof
By Rule induction on $S\steq T$
\begin{desCription}
\item\noindent{\hskip-12 pt\bf \rtit{scNil}:}\  $|S \paral \con{\emptyset}{\inert}| = |S| \paralS |\con{\emptyset}{\inert}|  = |S| \paralS \inert$  and $|S| \paralS \inert \steq |S|$ by \rtit{sNil}.
\item\noindent{\hskip-12 pt\bf \rtit{scCom}, \rtit{scAss}, \rtit{scNew}, \rtit{scSwp}:}\  By the corresponding structural rules \rtit{sCom}, \rtit{sAss}, \rtit{sNew}, \rtit{sSwp}.
\item\noindent{\hskip-12 pt\bf \rtit{scExt}:}\ By \rtit{sExt} and the fact that $c\not\in\fn{S}$ implies $c\not\in\fn{|S|}$.\qed
\end{desCription}

\begin{lemmabis}{\ref{thm:correspondence}}{Correspondence}\quad
    \begin{math}
   S \reduc T \quad\text{implies}\quad |S| \reduc |T|    \;\text{or}\; |S| \steq |T|  
  \end{math}
\end{lemmabis}

\proof
  The proof 
  is by rule induction on $S \reduc T$.
  \begin{desCription}
  \item\noindent{\hskip-12 pt\bf\rtit{cThn}, \rtit{cEls}, \rtit{cCom}, \rtit{cPrc}:}\   There is a corresponding reduction rule in the semantics of \figref{fig:structural-reduc}.
  \item\noindent{\hskip-12 pt\bf\rtit{cSpl}, \rtit{cRst}, \rtit{cDsc}:}\ Satisfies $|S| \steq |T|$.
  \item\noindent{\hskip-12 pt\bf\rtit{cPar}, \rtit{cRes}:}\ Follows by I.H.
  \item\noindent{\hskip-12 pt\bf\rtit{cStr}:}\ By \rtit{rStr} and Lemma~\ref{lem:struct-corresp-1} \qed
  \end{desCription}

\begin{cor}
  \label{cor:correspondence-2}
  \begin{math}
    |S|\reducNot \quad\text{implies}\quad S\reducNot \text{ or } (\exists T. \ S\reduc T \text{ and }  |S|\steq|T|)
  \end{math}
\end{cor}

\begin{lemmabis}{\ref{lem:quaseq-preserve-absence-err}}{Properties of \quaseq\ with respect to reductions}\quad
  \begin{enumerate}[\em(1)]
  \item   
   \begin{math}
    S\quaseq T\text{ and } T\reduc T'\text{ and }\text{S\errNot}\quad\text{implies}\quad \exists S'. S\reduc S'\text{ and } S'\quaseq T'
  \end{math}
 \item 
     \begin{math}
    S\quaseq T\text{ and } \text{S\checkmark} \quad\text{implies}\quad T\reducNot
  \end{math}
  \end{enumerate}
\end{lemmabis}

\begin{proof}
  The first clause is proved by case analysis of $T\reduc T'$ using Lemma~\ref{lem:reduc-proc-struct} to infer the structure of $T$, then use the definition $S\quaseq T$ to determine the structure of $S$.  The second clause is proved by assuming that $\exists T'$ such that $T\reduc T'$  and then use the first clause to show that this leads to a contradiction. 
\end{proof}

\begin{lem}[Reduction and System Structure] \label{lem:reduc-proc-struct} $S \reduc T$ implies
  \begin{enumerate}[\em(1)]
  \item  $S \steq \restB{\lst{c}}{\conC{\pioutA{c}{\lst{e}}}\paral \conCC{\piin{c}{\lst{x}}{P}}\paral R}$, $T \steq \restB{\lst{c}}{\con{\eV\cup\eVV}{P\subC{\lst{v}}{\lst{x}}}\paral R}$, $\permout{c}\in\eV$, $\permin{c}\in\eVV$, $\lst{e}\evaluate\lst{v}$  or;
     \item $S \steq \restB{\lst{c}}{\conC{\cmpC{\bV}{P}{Q}}\paral R}$,  $T \steq \restB{\lst{c}}{\conC{P}\paral R}$ or $T \steq \restB{\lst{c}}{\conC{Q}\paral R}$  or; 
     \item $S \steq \restB{\lst{c}}{\conC{\pCall{\pV}{\lst{e}}\subS{\lst{d_1}}{\lst{d_2}}}\paral R}$, $T \steq \restB{\lst{c}}{\conC{P\subC{\lst{v}}{\lst{x}}\subC{\lst{d_1}}{\lst{d_2}}}\paral R}$,  $\lst{e}\evaluate\lst{v}$  or;
     \item $S \steq \restB{\lst{c}}{\con{\eV\uplus\eVV}{P\paral Q} \paral R}$, $T \steq \restB{\lst{c}}{\conC{P}\paral \conCC{Q} \paral R}$ or;
     \item $S \steq \restB{\lst{c}}{\conC{\rest{c}{P}} \paral R}$, $T \steq \restB{\lst{c}}{\restB{c}{\con{\eV\cup\sset{\permin{c},\permout{c}}}{P}} \paral R}$ or;
      \item $S \steq \restB{\lst{c}}{\conC{\inert} \paral R}$, $S \steq \restB{\lst{c}}{\con{\emptyset}{\inert} \paral R}$, $\eV\neq \emptyset$ 
  \end{enumerate}
\end{lem}

\begin{proof}
  By rule induction on $S \reduc T$.
\end{proof}

\begin{propbis}{\ref{prop:safe-stab-vs-system-structure}}{Safe-Stability and System Structure}
  \begin{align*}
    S\checkmark & \quad\text{iff}\quad  S \steq\restB{\lst{d}}{ \;\paral_{i=0}^n \con{\eV_i}{\pioutA{c_i}{\lst{e_i}}} \; \paral_{j=0}^m \con{\eVV_j}{\piin{c'_j}{\lst{x_j}}{P_j}}}
  \end{align*}
  where
  \begin{iteMize}{$\bullet$}
  \item $\sset{c_1,\ldots, c_n} \cap \sset{c'_1,\ldots,c'_m} = \emptyset$.
  \item $\bigwedge_{i=0}^n \permout{c_i} \in \eV_i$  and $\bigwedge_{j=0}^m \permin{c_j} \in \eVV_j$
  \end{iteMize}
  and where $\;\paral_{i=0}^0 \con{\eV_i}{\pioutA{c_i}{\lst{e_i}}}$ and  $\;\paral_{j=0}^0 \con{\eVV_j}{\piin{c'_j}{\lst{x_j}}{P_j}}$ denote = $\con{\emptyset}{\inert}$.
\end{propbis}

\begin{proof}
  Immediate by case analysis of Lemma~\ref{lem:reduc-proc-struct} and then the conditions for $S\err$ from \figref{fig:confined-lang}.  
\end{proof}

\begin{lemmabis}{\ref{lem:confluence}}{Partial Confluence}
    \begin{math}
    S \reduc T_1 \text{ and }  S \reduc T_2 \; \text{implies  either of the following:}
   \end{math}
    \begin{enumerate}[\em(1)]
    \item  $T_1\quaseq T_2 \text{ or;}$
    \item $\exists T_3. \, T_1 \reduc T_3 \text{ and } T_2 \reduc T_3$
    \end{enumerate}
\end{lemmabis}

\proof
 By case analysis of the possible forms of $S$ using Lemma~\ref{lem:reduc-proc-struct}, then restricting the possibilities using properties of well-formed systems.  We here overview the two main cases.
 \begin{iteMize}{$\bullet$}
 \item For $S\reduc T_1$ we have 
\[S \steq \restB{\lst{c}}{\con{\eV_1}{\pioutA{c_1}{\lst{e_1}}}\paral
  \con{\eVV_1}{\piin{c_1}{\lst{x}}{P_1}}\paral R_1}
  \enbox{,}
  T_1 \steq\restB{\lst{c}}{\con{\eV_1\cup\eVV_1}{P\subC{\lst{v}}{\lst{x}}}\paral
  R_1}
  \enbox{,}
  \permout{c_1}\in\eV_1
  \enbox{,}\permin{c_1}\in\eVV_1\ .
\]
  Also for $S\reduc T_2$ we have 
\[S \steq \restB{\lst{c}}{\con{\eV_2}{\pioutA{c_2}{\lst{e_2}}}\paral
  \con{\eVV_2}{\piin{c_2}{\lst{x}}{P_2}}\paral R_2}
  \enbox{,}
  T_1 \steq
  \restB{\lst{c}}{\con{\eV_2\cup\eVV_2}{P\subC{\lst{v}}{\lst{x}}}\paral
    R_2}
  \enbox{,}
\permout{c_2}\in\eV_2
  \enbox{,}
\permin{c_2}\in\eVV_2\ .
\]
  We have two sub-cases
   \begin{description}
   \item[$c_1\neq c_2$] The two redexes in $S$ are distinct and, for
     some system $R$, we have 
\[R_1 \steq \restB{\lst{d_2}}{\con{\eV_2}{\pioutA{c_2}{\lst{e_2}}}\paral \con{\eVV_2}{\piin{c_2}{\lst{x}}{P_2}}\paral R}
   \quad\mbox{and}\quad
 R_2 \steq
 \restB{\lst{d_1}}{\con{\eV_1}{\pioutA{c_1}{\lst{e_1}}}\paral
   \con{\eVV_1}{\piin{c_1}{\lst{x}}{P_1}}\paral R}
\]
  from which we can then find a common $T_3$  that both $T_1$ and $T_2$ reduce to.
 \item[$c_1= c_2$]  The conditions that $\permout{c_1}\in\eV_1$, $\permin{c_1}\in\eVV_1$, $\permout{c_2}\in\eV_2$ and $\permin{c_2}\in\eVV_2$ and the fact that $S$ is well-formed ensure that $S\reduc T_1$ and $S\reduc T_2$ refer to the same reduction (modulo structural equivalence) \ie $\eV_1=\eV_2$, $\eVV_1=\eVV_2$, $\lst{e_1}=\lst{e_2}$, $P_1=P_2$ and $R_1=R_2$ which implies $T_1\steq T_2$, thus $T_1\quaseq T_2$ by Proposition~\ref{lem:quas-impl-unconst-steq}.
   \end{description}
   \item For $S\reduc T_1$ we have $S \steq \restB{\lst{c}}{\con{\eV_1\uplus\eVV_1}{P_1\paral Q_1} \paral R_1}$, $T_1 \steq \restB{\lst{c}}{\con{\eV_1}{P_1}\paral \con{\eVV_1}{Q_1} \paral R_1}$ and for $S\reduc T_2$ we have $S \steq \restB{\lst{c}}{\con{\eV_2\uplus\eVV_2}{P_2\paral Q_2} \paral R_2}$, $T_2 \steq \restB{\lst{c}}{\con{\eV_2}{P_2}\paral \con{\eVV_2}{Q_2} \paral R_2}$.  By the assumption that $S$ is well-formed, we have the following sub-cases:
     \begin{description}
       \item[$(\eV_1\uplus\eVV_1) \neq (\eV_2\uplus\eVV_2)$]  Then we
         have different redexes meaning that for some $R$ we have 
\[R_1\steq \restB{\lst{d_2}}{\con{\eV_2\uplus\eVV_2}{P_2\paral Q_2}
  \paral R}
  \quad\mbox{and}\quad
 R_2\steq \restB{\lst{d_1}}{\con{\eV_1\uplus\eVV_1}{P_1\paral Q_1}
   \paral R}\ ,
\]
  which guarantees the existence of a common system $T_3$ that $T_1$ and $T_2$ can reduce to.
         \item[$(\eV_1\uplus\eVV_1) = (\eV_2\uplus\eVV_2)$] Then we must have the same redexes, \ie $P_1=P_2$, $Q_1=Q_2$ and $R_1=R_2$.  This implies $T_1\quaseq T_2$.  \qed
     \end{description}
 \end{iteMize}

\noindent The following technical Lemmas deal with the restricted
non-determinism of confined processes and how it can be characterised
using the relation \quaseq.  In particular,
Lemma~\ref{lem:corrective-reduc} is useful because it allows us to
correct reductions that lead to systems that do not evaluate by
instead reducing to systems that are related to them by \quaseq, which
in turn means, by Proposition~\ref{lem:quas-impl-unconst-steq}, that
they contain the same process structure.

\begin{lem} \label{lem:non-det-implies-struct}
  \begin{math}
    S \evaluate \text{ and } S \reduc T \text{ and } T \not\evaluate \quad\text{implies}\quad \exists P, Q, R. \quad S \steq \restB{\lst{c}}{\con{\eV\uplus\eVV}{P\paral Q}\paral R}\text{ and } \newline T \steq \restB{\lst{c}}{\conC{P}\paral \conCC{Q}\paral R}
  \end{math}
\end{lem}

\proof
  By induction on the number of reductions in $S\evaluate$ leading to a safely-stable system \ie $S\reduc^n S'\checkmark$ for some $S'$.
  \begin{desCription}
  \item\noindent{\hskip-12 pt\bf $n=1\,$:}\ \ By Lemma~\ref{lem:confluence} and $S'\checkmark$ (\ie $S'\reducNot$) it must be the case that $T\quaseq S'$.  By Lemma~\ref{lem:quaseq-preserve-absence-err}.2 this also implies $T\reducNot$ and since $T\not\evaluate$ it must be the case that $T\err$.  Now by case analysis of Lemma~\ref{lem:reduc-proc-struct}, the only system structure that allows this is when $S\steq\restB{\lst{c}}{\con{\eV\uplus\eVV}{P\paral Q}\paral S''}$ and $T\steq \restB{\lst{c}}{\conC{P}\paral \conCC{Q}\paral S''}$.
  \item\noindent{\hskip-12 pt\bf $n=k+1\,$:}\ \ We have $S\evaluate$ $S\reduc S'' \reduc^k S'\checkmark$.  By Lemma~\ref{lem:confluence} we have two sub-cases.  The first case subsumes the second in some cases, so we here consider the mutually exclusive    variants:
    \begin{description}
    \item[$\exists T'. T\reduc T'\text{ and }S''\reduc T'$]  $T\not\evaluate$ implies $T'\not\evaluate$, and by $S''\reduc T'$, $S'' \reduc^k S'\checkmark$ and I.H. we obtain $S''\steq\restB{\lst{c}}{\con{\eV\uplus\eVV}{P\paral Q}\paral S'''}$ and $T'\steq \restB{\lst{c}}{\conC{P}\paral \conCC{Q}\paral S'''}$. Now $T \not\evaluate$ and $S'' \reduc^k S'\checkmark$ implies $T\neq S''$.   Thus by that fact that $S\reduc T\reduc T'$ and the uniqueness of linear permissions, it must be the case that $S\steq\restB{\lst{c}}{\con{\eV\uplus\eVV}{P\paral Q}\paral S''''}$ and $T\steq \restB{\lst{c}}{\conC{P}\paral \conCC{Q}\paral S''''}$ for some $S''''$ such that $S''''\reduc S'''$.
    \item[$T\quaseq S'' \text{ where } \not\exists T'. T\reduc T'\text{ and }S''\reduc T'$] Clearly, since  $T \not\evaluate$ we have $T\neq S''$.  Also the fact that there is no common system $T$ and $S''$ can reduce to means that the reductions from $S$ where not from separate redexes. By case analysis of Lemma~\ref{lem:reduc-proc-struct}, the only possible option for having non-deterministic reductions from the same redex is the case where  $S\steq\restB{\lst{c}}{\con{\eV\uplus\eVV}{P\paral Q}\paral S'''}$ and $T\steq \restB{\lst{c}}{\conC{P}\paral \conCC{Q}\paral S'''}$. \qed
    \end{description}
  \end{desCription}

\begin{lem}\label{lem:eval-and-splt-implies-good-split}
  \begin{math}
    S\evaluate \text{ and } S\steq\restB{\lst{c}}{\conC{P\paral Q}\paral R}\quad \text{implies}\quad \exists \eVV_1,\eVV_2\text{ such that } \eVV_1\uplus\eVV_2=\eV\text{ and } \newline\restB{\lst{c}}{\con{\eVV_1}{P}\paral \con{\eVV_2}{Q}\paral R}\evaluate
  \end{math}
\end{lem}

\proof
    By induction on the number of reductions in $S\evaluate$ leading to a safely-stable system \ie $S\reduc^n S'\checkmark$ for some $S'$.
    \begin{desCription}
    \item\noindent{\hskip-12 pt\bf $n=1\,$:}\ \  By \rtit{cSpl}, $S\steq\restB{\lst{c}}{\conC{P\paral Q}\paral R}$ can reduce to  $\restB{\lst{c}}{\con{\eV_1}{P}\paral \con{\eV_2}{Q}\paral R}$,    for some $\eV_1,\eV_2$, and by  Lemma~\ref{lem:confluence} and  $S'\reducNot$ we must have  $S'\quaseq \restB{\lst{c}}{\con{\eV_1}{P}\paral \con{\eV_2}{Q}\paral R}$, and since $S'\errNot$, this implies $\exists \eVV_1,\eVV_2\text{ such that } \eVV_1\uplus\eVV_2=\eV\text{ and } \restB{\lst{c}}{\con{\eVV_1}{P}\paral \con{\eVV_2}{Q}\paral R}\evaluate$.
    \item\noindent{\hskip-12 pt\bf $n=k+1\,$:}\ \ We have $S\reduc S'\reduc^k S''\checkmark$ for some $S',S''$.  Lemma~\ref{lem:reduc-proc-struct} gives us two sub-cases:
      \begin{description}
      \item[$S'\steq \restB{\lst{c}}{\conC{P\paral Q}\paral R'}$ where $R\reduc R'$]  By  $S'\reduc^k S''\checkmark$ and  I.H. we obtain  $\exists \eVV_1,\eVV_2$ such that $\eVV_1\uplus\eVV_2=\eV\text{ and } \restB{\lst{c}}{\con{\eVV_1}{P}\paral \con{\eVV_2}{Q}\paral R'}\evaluate$ which implies that $\exists \eVV_1,\eVV_2$  such that ${\eVV_1\uplus\eVV_2=\eV}$  and  $\restB{\lst{c}}{\con{\eVV_1}{P}\paral \con{\eVV_2}{Q}\paral R}\evaluate$.
      \item[$\restB{\lst{c}}{\con{\eV_1}{P}\paral \con{\eV_2}{Q}\paral R}$] Immediate. \qed
      \end{description}
    \end{desCription}

\begin{lem}[Corrective Reductions]\label{lem:corrective-reduc}
  \begin{math}
    S\evaluate \text{ and } S\reduc T\text{ and }T\not\evaluate \quad \text{implies}\quad \exists R\text{ such that } S\reduc R\text{ and } R\quaseq T \text{ and } R\evaluate
  \end{math}
\end{lem}

\begin{proof}
  By Lemma~\ref{lem:non-det-implies-struct} we know $S \steq
  \restB{\lst{c}}{\con{\eV\uplus\eVV}{P\paral Q}\paral S'}$ as well as
  $T \steq \restB{\lst{c}}{\conC{P}\paral \conCC{Q}\paral S'}$.  By
  Lemma~\ref{lem:eval-and-splt-implies-good-split} we know $\exists\eVV_1,\eVV_2$ such that $\eVV_1\uplus\eVV_2=\eV$  and $\restB{\lst{c}}{\con{\eVV_1}{P}\paral \con{\eVV_2}{Q}\paral S'}\evaluate$.   Since  $T\quaseq \restB{\lst{c}}{\con{\eVV_1}{P}\paral \con{\eVV_2}{Q}\paral S'}$ this implies that we can correct the permission split and be able to reduce to a safely-stable state.  
\end{proof}

In order to apply corrective actions to multiple reduction steps, we need to extend Lemma~\ref{lem:corrective-reduc} to systems that are related by \quaseq, due to reductions of type $(1)$ of Lemma~\ref{lem:confluence}.  The next Lemmas deal with this.    Lemma~\ref{lem:eval-preserve-quaseq} states that there exist matching reductions for systems related by $\quaseq$ preserving the evaluation property and Lemma~\ref{lem:eval-preserve-quaseq-multi-reduc} extends this to multiple reductions.  This allows us to prove the existence of corrective reductions over multiple reductions.

\begin{lemmabis}{\ref{lem:eval-preserve-quaseq}}{Evaluation Preservation for \quaseq}\quad
  \begin{math}
    S\quaseq T \text{ and } S\evaluate \text{ and } T\reduc T'\quad\text{implies}\newline \quad\exists S'\text{ such that } S\reduc S' \text{ where } S'\quaseq T'\text{ and } S'\evaluate
  \end{math}
\end{lemmabis}
\proof
By   $S\evaluate$ and Lemma~\ref{lem:violation-pres} we have $S\errNot$ and by Lemma~\ref{lem:quaseq-preserve-absence-err}.1 we know $\exists\ S_1$ such that $S\reduc S_1\text{ and } S_1\quaseq T'$.  At this point we have two sub-cases:  if $S_1\evaluate$ then the result follows immediately.  Otherwise, if $S_1\not\evaluate$,  then Lemma~\ref{lem:corrective-reduc} states that $\exists S_2\text{ such that } S\reduc S_2\text{ and } S_2\quaseq S_1 \text{ and } S_2\evaluate$.  By transitivity we have $S_2\quaseq S_1\quaseq T'$. \qed

\begin{lem}[Evaluation Preservation for \quaseq]\label{lem:eval-preserve-quaseq-multi-reduc}\quad
  \begin{math}
    S\quaseq T \text{ and } S\evaluate \text{ and } T\reduc^n T'\quad\text{implies}\newline \quad\exists S'\text{ such that } S\reduc^n S' \text{ where } S'\quaseq T'\text{ and } S'\evaluate
  \end{math}
\end{lem}
\proof
  By induction on $n$, the number of reductions in $T\reduc^n T'$.
  \begin{desCription}
  \item[$n=0$:] Immediate.
  \item[$n=k+1$:] We have $T\reduc T''\reduc^k T'$.  From $T\reduc T''$ and Lemma~\ref{lem:eval-preserve-quaseq} we obtain $\exists S'\text{ such that } S\reduc S' \text{ where } S'\quaseq T'\text{ and } S'\evaluate$. By I.H. we know $S'\reduc^k S''$ for some $S''$ such that $S''\quaseq T'$ and $S''\evaluate$ and $S\reduc S'\reduc^k S''$ gives us the required reduction sequence. \qed
  \end{desCription}

\begin{lem}\label{lem:eval-corresp-struct}
  $|S| \steq Q \text{ and } S\evaluate \quad \text{implies}\quad \exists T\text{ such that } S\reduc^\ast \steq T \text{ and } T\evaluate \text{ and }|T| = Q $. 
\end{lem}

\begin{proof}
  By rule induction on $|S| \steq Q$. 
\end{proof}

\begin{lem}\label{lem:struct-corresp-2}
  \begin{math}
    |S| \steq Q\quad \text{implies}\quad \exists T.\;  S \reduc^{\ast} T  \text{ or }  S \steq T \;S\text{ where }\;   |T| =Q   
  \end{math}  
\end{lem}

\proof
  By rule induction on $|S|\steq Q$ and then a tedious consideration of all the possible permutations of $S$ that may lead to $|S|$.  
  \begin{desCription}
  \item[\bf\rtit{sAss}:] If $|S| =P_1\paral (P_2 \paral P_3)$ then $Q = (P_1\paral P_2) \paral P_3$ and $S$ can be either of the following:
    \begin{description}
    \item[$S=\conC{P_1\paral (P_2 \paral P_3)}$] By $2$ applications of \rtit{cSpl} and then an application of \rtit{cStr} using \rtit{scAss} we obtain $S \reduc^{+} (\con{\eV_1}{P_1} \paral \con{\eV_2}{P_2}) \paral \con{\eV_3}{P_3}$ where $\eV_1\uplus\eV_2\uplus\eV_3 = \eV$ and $|(\con{\eV_1}{P_1} \paral \con{\eV_2}{P_2}) \paral \con{\eV_3}{P_3}| = Q$. 
     \item[$S=\con{\eV_1}{P_1}\paral \conC{(P_2 \paral P_3)}$] By one application of \rtit{cSpl} and one application of \rtit{cStr} using \rtit{scAss} we obtain $S \reduc^{+} (\con{\eV_1}{P_1} \paral \con{\eV_2}{P_2}) \paral \con{\eV_3}{P_3}$ where $\eV_2\uplus\eV_3 = \eV$ and $|(\con{\eV_1}{P_1} \paral \con{\eV_2}{P_2}) \paral \con{\eV_3}{P_3}| = Q$. 
       \item[$S=\con{\eV_1}{P_1}\paral \con{\eV_2}{(P_2} \paral \con{\eV_3}{P_3)}$] By  \rtit{scAss} we obtain $S \steq (\con{\eV_1}{P_1} \paral \con{\eV_2}{P_2}) \paral \con{\eV_3}{P_3}$ where  $|(\con{\eV_1}{P_1} \paral \con{\eV_2}{P_2}) \paral \con{\eV_3}{P_3}| = Q$. 
    \end{description}
    The symmetric case where $|S| =(P_1\paral P_2) \paral P_3$ and $Q = P_1\paral (P_2 \paral P_3)$ is similar.
  \item[\bf\rtit{sCom}:]  Similar to \rtit{sAss} case.
  \item[\bf\rtit{sNil}:] If $|S| =P\paral \inert$ and $Q=P$ then we have two cases:
    \begin{description}
    \item[$S=\conC{P \paral \inert}$] By \rtit{cSpl}, \rtit{cStr} and \rtit{scNil} we obtain $S \reduc^{+} \conC{P}$ and $|\conC{P}| = P = Q$.
    \item[$S=\con{\eV_1}{P}\paral\con{\eV_2}{\inert}$] By \rtit{cDsc}, \rtit{cStr} and \rtit{scNil} we obtain $S \reduc^{+} \conC{P}$ and $|\conC{P}| = P = Q$.
    \end{description}
     If $|S| =P$ and $Q=P\paral \inert$, then by \rtit{scNil} we have $S\steq S\paral \con{\emptyset}{\inert}$ and $|S\paral \con{\emptyset}{\inert}| =Q$.
   \item[\bf\rtit{sNew}:] The most difficult case is when $S=\conC{\rest{c}{\inert}}$ and $Q=\inert$.  By \rtit{cLcl}, \rtit{cDsc}, \rtit{cStr} with \rtit{scNew} we obtain $S\reduc^{+}\con{\emptyset}{\inert}$ and $|\con{\emptyset}{\inert}|    =Q$.  The other cases are similar.
   \item[\bf\rtit{sSwp}:] There are three cases; $S=\conC{\rest{c}{\rest{d}{P}}}$, $S=\rest{c}{\conC{\rest{d}{P}}}$ and $S=\rest{c}{\rest{d}{\conC{P}}}$ and proved similar to the cases above using \rtit{cLcl}, \rtit{cStr} and \rtit{scSwp}.
   \item[\bf\rtit{sExp}:]  When $|S| =P\paral\rest{c}{Q}$ we have three cases: $S=\conC{P\paral\rest{c}{Q}}$, $S=\con{\eV_1}{P}\paral\con{\eV_2}{\rest{c}{Q}}$ and $S=\con{\eV_1}{P}\paral\rest{c}{\con{\eV_2}{Q}}$ and the proof follows using the rules \rtit{cSpl}, \rtit{cLcl}, \rtit{cStr} and \rtit{scExt}. The symmetric case when $|S| =\rest{c}{P\paral Q}$ is similar. \qed
  \end{desCription}

\begin{lemmabis}{\ref{lem:eval-correspondence}}{Reduction Correspondence}
  \begin{displaymath}
    S \evaluate  \text{ and } |S|\reduc Q \quad \text{implies}\quad \exists R\text{ such that } S\reduc^{+} R \text{ and } |R| \steq Q
  \end{displaymath}
\end{lemmabis}

\proof
  By rule induction on $|S|\reduc Q$.  We here consider the main cases:
  \begin{desCription}
  \item[\bf\rtit{rCom}] We have $|S| =\pioutA{c}{\lst{e}}\paral\piin{c}{\lst{x}}{P}$ and $Q=P\subC{\lst{v}}{\lst{x}}$ where $\lst{e}\evaluate\lst{v}$.  We have two sub-cases for $S$:
    \begin{description}
      \item[$S=\conC{\pioutA{c}{\lst{e}}\paral\piin{c}{\lst{x}}{P}}$]  By $S\evaluate$ , $\exists \eVV_1,\eVV_2$ such that $\eVV_1\uplus\eVV_2=\eV$ and
        \begin{math}
           S \reduc \con{\eVV_1}{\pioutA{c}{\lst{e}}}\paral\con{\eVV_2}{\piin{c}{\lst{x}}{P}} \reduc \conC{P\subC{\lst{v}}{\lst{x}}}
        \end{math}
        and $|\conC{P\subC{\lst{v}}{\lst{x}}}| = Q$.
      \item[$S = \con{\eVV_1}{\pioutA{c}{\lst{e}}}\paral\con{\eVV_2}{\piin{c}{\lst{x}}{P}}$]  Similar
    \end{description}
  \item[\bf\rtit{rPar}] We have $|S| = P_1\paral P_2$ and $Q = P'_1 \paral P_2$ because $P_1\reduc P'_1$. We have two sub-cases for $S$:
    \begin{description}
    \item[$S=\conC{P_1\paral P_2}$] By $S\evaluate $ $\exists \eVV_1,\eVV_2$ such that $\eVV_1\uplus\eVV_2=\eV$ and
      \begin{math}
        \conC{P_1\paral P_2} \reduc \con{\eVV_1}{P_1}\paral\con{\eVV_2}{P_2} \evaluate
      \end{math}.  
      Now $\con{\eVV_1}{P_1}\paral\con{\eVV_2}{P_2} \evaluate$ implies $\con{\eVV_1}{P_1} \evaluate$ and by $P_1\reduc P'_1$ and I.H. we know $\exists R\text{ such that } \con{\eVV_1}{P_1}\reduc^{+} R \text{ and } |R| \steq P'_1$. Thus, by \rtit{cPar},
      \begin{math}
        \con{\eVV_1}{P_1}\paral\con{\eVV_2}{P_2} \reduc R \paral\con{\eVV_2}{P_2}
      \end{math}
      and $| R\paral\con{\eVV_2}{P_2}| = Q$.
    \item[$S=S_1\paral S_2$ where $|S_1| = P_1$ and $|S_2| = P_2$] Similar       
    \end{description}

  \item[\bf\rtit{rStr}]  We have $|S| = P_1$ and $Q = P_2$  because $\ P_1\steq P'_1,\ P'_1\reduc P'_2,\ P'_2\steq P_2$.  By $|S| = P_1$ and Lemma~\ref{lem:eval-corresp-struct} we know
    \begin{math}
      \exists R_1\text{ such that } S\reduc^\ast\steq R_1\text{ and } R_1\evaluate\text{ and } |R_1| = P'_1
    \end{math}.  By $\ P'_1\reduc P'_2$ and I.H. we know
    \begin{math}
      \exists R_2\text{ such that } R_1\reduc^{+}\steq R_2\text{ and } |R_2| = P'_2
    \end{math},
    and by $P'_2\steq P_2$ and Lemma~\ref{lem:struct-corresp-2} we know
    \begin{math}
      \exists R_3\text{ such that } R_2\reduc^\ast\steq R_3\text{ and }  |R_3| = P_2
    \end{math}.
   This implies $S\reduc^\ast\steq R_1 \reduc^{+}\steq R_2 2\reduc^\ast\steq R_3$, \ie $S\reduc^{+} R_3$  where $|R_3| = Q$. \qed
  \end{desCription}

\begin{lemmabis}{\ref{lem:eval-determ-no-reduc}}{Correspondence and Termination}
  \begin{math}
      |S|\reducNot \text{ and } S\evaluate T\quad\text{implies}\quad |T|\steq |S| 
  \end{math} 
\end{lemmabis}
\proof
  By induction on the number of reductions that lead to a safely-stable system $S\reduc^{n}T$
  \begin{desCription}
  \item[$n=0$:] We have $S = T$ which implies $|S| = |T|$ 
  \item[$n=k+1$:] We have $S\reduc R$ and $R\evaluate T$.  By $S\reduc R$  and Cor.~\ref{cor:correspondence-2} we get $|S| \steq |R|$ and thus $|R|\reducNot$.  Hence by I.H. and $R\evaluate T$ we get $|T|\steq |R|$ and by transitivity we obtain $|T|\steq |S|.$ \qed
  \end{desCription}

\subsection{The Logic}
\label{sec:logic-app}


\begin{lem}\label{lem:in-trg-out-edg}
  When $S\reducNot$ and $\env,S\sat\fV$
  \begin{iteMize}{$\bullet$}
   \item $S\steq \conC{\pioutA{c}{\lst{e}}}\paral R$  implies $\permout{c}\in\edgE{\fV}$ or $\edgE{\fV}$ is undefined; 
   \item  $S\steq \rest{\lst{d}}{\conC{\piin{c}{\lst{x}}{P}}\paral R}$ and $c\in\lst{d}$ implies   $\permout{c}\in\trgE{\fV}$ or $\trgE{\fV}$ is undefined.
  \end{iteMize}
\end{lem}

\begin{proof}
  By induction on the structure of \fV.  
\end{proof}

\begin{lem}  \label{lem:merge-assert-stable}
  \begin{math}
    \env,\,S \sat \fV, S\reducNot \text{ and }\env,\,T \sat \fVV,  T \reducNot \text{ and }  \sepprocE{\fV}{\fVV} \quad\text{implies}\quad \env,\, S\paral T \reducNot
  \end{math}
\end{lem}

\begin{proof}
   Since $S\reducNot \text{ and } T \reducNot$, by Lemma~\ref{lem:reduc-proc-struct} we know that $S\paral T \reduc R$ for some $R$  can only happen if:
  \begin{align}
    \label{eq:102}
    & S \steq \restB{\lst{d}}{\conCC{\piin{c}{\lst{x}}{P}} \paral S'} \text{ where } c\not\in\lst{d} \text{ and } \permin{c}\in\eVV\\
    \label{eq:103}
    & T \steq \conC{\pioutA{c}{\lst{e}}}\paral T' \text{ where } \permout{c}\in\eV 
  \end{align}
or vice-versa.  We here focus on the case where \eqref{eq:102} and \eqref{eq:103} have to hold; the dual case is identical.  By \sepprocE{\fV}{\fVV} we know that \trgE{\fV}, \edgE{\fV}, \trgE{\fVV} and  \edgE{\fVV} must all be defined.  Thus by \eqref{eq:102}, $\env,\,S \sat \fV$ and Lemma~\ref{lem:in-trg-out-edg} we must have $\permout{c}\in\trgE{\fV}$. Similarly by \eqref{eq:103}, $\env,\,T \sat \fVV$ and  Lemma~\ref{lem:in-trg-out-edg} we must have $\permout{c}\in\edgE{\fVV}$.  This however would contradict \sepprocE{\fV}{\fVV} which requires that $\trgE{\fV}\cap \edgE{\fVV} = \emptyset$.
Thus $S\paral T \reducNot$.
\end{proof}

\begin{lemmabis}{\ref{lem:merge-assert}}{Merging Assertions}
  \begin{math}
    \env,\,S \sat \fV\text{ and }\env,\,T \sat \fVV  \text{ and } S\perp T \text{ and } \sepprocE{\fV}{\fVV} \quad\text{implies}\quad \env,\, S\paral T \satS \fcons{\fV}{\fVV}
  \end{math}
\end{lemmabis}

\begin{proof}
  $S\perp T$ implies $S\paral T $ is well-resourced.  From  $\env,\,S \sat \fV$, $\env,\,T \sat \fVV$ and Proposition~\ref{cor:satisfaction-evaluation} we know that $S\evaluate S'$ and $T\evaluate T'$ where $\env,\,S' \sat \fV$ and $\env,\,T' \sat \fVV$.   Lemma~\ref{lem:merge-assert-stable} we know also that   $S\paral T \evaluate S' \paral T'$ and the result follows by satisfaction on \figref{fig:assertion-satisfaction}.
\end{proof}

\subsection{The Proof System}
\label{sec:proof-system-app}

Proofs for the derived rules from \secref{sec:derived-rules}.

\begin{description}
 \item[The proof for \rtit{lCut}] \begin{align*}
  &
  \infer[\rtit{lImp}]{\seqEB{\fV_1}{S\paral T}{\fV_2}}{
    \infer[\rtit{lPar}]{\seqEB{\fcons{\fV_1}{\femp}}{S\paral T}{\fcons{\femp}{\fV_2}}}{
      \infer[\rtit{lImp}]{\seqEB{\fV_1}{S}{\fcons{\femp}{\fVV}}}{\seqEB{\fV_1}{S}{\fVV}}
      & \infer[\rtit{lImp}]{\seqEB{\fcons{\femp}{\fVV}}{T}{\fV_2}}{\seqEB{\fVV}{T}{\fV_1}}
      & \sepprocE{\femp}{\fVV} & \sepprocE{\femp}{\fV_2}
     }
   }
\end{align*}

\item[The proof for  \rtit{lSep}] \begin{align*}
  & \infer[\rtit{lPar}]{\seqEB{\fcons{\fV_1}{\fV_2}}{S\paral T}{\fcons{\fVV_1}{\fVV_2}}}{
      \infer[\rtit{lImp}]{\seqEB{\fV_1}{S}{\fcons{\fVV_1}{\femp}}}{\seqEB{\fV_1}{S}{\fVV_1}}
      & \infer[\rtit{lImp}]{\seqEB{\fcons{\fV_2}{\femp}}{T}{\fVV_2}}{\seqEB{\fV_2}{T}{\fVV_2}}
      & \sepprocE{\fVV_1}{\fVV_2} & \sepprocE{\fV_2}{\femp}
   } 
\end{align*}
\item[The proof for \rtit{lOutD}] \begin{align*}
  & \infer[\rtit{lImp}]
      {\seqEB{\femp}{\conC{\pioutA{c}{\lst{e_1}}}}{\fstate{c}{\lst{e_2}}}}
      {\bV \models \bV\wedge \lst{e_1}\! =\! \lst{e_1} \wedge \lst{e_1}\! =\!\lst{e_2}   &
        \infer[\rtit{lInst}]
        {\seqE{\bV\wedge \lst{e_1} = \lst{e_1} \wedge \lst{e_1} = \lst{e_2}}{\femp}{\conC{\pioutA{c}{\lst{e_1}}}}{\fstate{c}{\lst{e_2}}}}
        {\lst{x}\not\in\fn{\bV} \cup \fn{\lst{e_2}, \lst{e_1}} \quad  \infer[\rtit{lSub}]
          {\seqE{\bV\wedge \lst{x} = \lst{e_1} \wedge  \lst{x} = \lst{e_2}}{\femp}{\conC{\pioutA{c}{\lst{x}}}}{\fstate{c}{\lst{e_2}}}}
          {\infer[\rtit{lOut}]
            {\seqE{\bV\wedge \lst{x} = \lst{e_1} \wedge  \lst{x} = \lst{e_2}}{\femp}{\conC{\pioutA{c}{\lst{e_2}}}}{\fstate{c}{\lst{e_2}}}}
            {\env(c)\subseteq\eV}
           } 
        }
      }
\end{align*}
\item[The proof for \rtit{lInD}] \begin{align*}
  &  \infer[\rtit{lImp}]{\seqEB{\fcons{\fV}{\fstate{c}{\lst{e}}}} { T  } {\fVV}}{
        T \steq \con{\eV\setminus\env(c)} {\piin{c}{\lst{x}}{P}} \paral S  &
        \infer[\rtit{lIn}] {\seqEB{\fcons{\fV}{\fstate{c}{\lst{e}}}} {\con{\eV\setminus\env(c)} {\piin{c}{\lst{x}}{P}} \paral S } {\fVV}}
         {\permin{c}\in\eV & \seqEB{\fV} {\conC{P\subC{\lst{e}}{\lst{x}}} \paral S } {\fVV}
           }
       }
\end{align*}

 \end{description}


\begin{lem}\label{lem:sub-reduc}  Assume that $\subC{e|v}{x}$ is a substitution that non-deterministically substitutes either $e$ or $v$ for $x$.  Then we have 
  \begin{displaymath}
    S\subC{v}{x} \reduc T\subC{v}{x} \text{ and } e\evaluate v \text{ implies } S\subC{e}{x} \reduc R \text{ where } R=T\subC{e|v}{x}    
  \end{displaymath}
 for some non-deterministic substitution  $T\subC{e|v}{x}$
\end{lem}
\proof By rule induction on $S\subC{v}{x} \reduc T\subC{v}{x}$ \qed 

\begin{lem}
  \label{lem:subs-stable}
  \begin{math}
    \env,\,T\subC{\lst{v}}{\lst{x}}\satS \fV\text{ and }\ \lst{e}\evaluate\lst{v}\text{ and } T\reducNot\quad \text{implies}\quad\env,\,T\subC{\lst{e}}{\lst{x}}\satS \fV
  \end{math}
\end{lem}
\begin{proof}
  By induction on the structure of \fV
\end{proof}

\begin{lem}
  \label{lem:subs}
  \begin{math}
    \env,\,T\subC{\lst{v}}{\lst{x}}\satS \fV\text{ and }\ \lst{e}\evaluate\lst{v}\quad \text{implies}\quad\env,\,T\subC{\lst{e}}{\lst{x}}\satS \fV
  \end{math}
\end{lem}
\begin{proof}
  Follows from Lemma~\ref{lem:sub-reduc} and Lemma~\ref{lem:subs-stable}.
\end{proof}


\begin{lem}\label{lem:sat-rest}
  \begin{math}
    \env, S \models \fV \text{ and }\ d \not\in \dom{\env}  \text{ implies } \env, \rest{d}{S} \models \fV
  \end{math}
\end{lem}

\proof
  By induction on the structure of \fV.  For instance:
  \begin{desCription}
  \item[\bf\fstate{c}{\lst{e}}:] We know $S\evaluate \conC{\pioutA{c}{\lst{e_1}}}$ where $\lst{e}\evaluate\lst{v}$, $\lst{e_1}\evaluate\lst{v}$  and $\env(c)\subseteq\eV$.  By \rtit{cRes}  and then by \rtit{cTgh} and $d \not\in \nm{\fstate{c}{\lst{e}}}\cup \nm{\env}$ we deduce
    \begin{align*}
       & \rest{d}{S} \reduc^\ast\steq \restB{d}{\conC{\pioutA{c}{\lst{e_1}}}}\steq \restB{d}{\conC{\pioutA{c}{\lst{e_1}}} \paral \con{\emptyset}{\inert}} \\
&
\restB{d}{\conC{\pioutA{c}{\lst{e_1}}} \paral \con{\emptyset}{\inert}} \reduc \con{\eV\setminus\sset{\permin{d},\permout{d}}}{\pioutA{c}{\lst{e_1}}} \paral \restB{d}{\con{\emptyset}{\inert}} \steq \con{\eV\setminus\sset{\permin{d},\permout{d}}}{\pioutA{c}{\lst{e_1}}} 
    \end{align*}
 Since $d\not\in\dom{\env}$ then by \defref{def:permission-environment}.2, \ie the environment is suitably closed, it follows that $\env(c)\subseteq(\eV\setminus\sset{\permin{d},\permout{d}})$  and hence $\env, \con{\eV\setminus\sset{\permin{d},\permout{d}}}{\pioutA{c}{\lst{e_1}}}\sat  \fstate{c}{\lst{e}}$  and by Proposition~\ref{cor:satisfaction-convergence} that $ \env, \rest{d}{S} \models \fstate{c}{\lst{e}}$.\qed
  \end{desCription}


\begin{defi}[Permission Restriction]\label{def:perm-rest}
  \begin{align*}
    S \setminus \eVV & \deftxt
    \begin{cases}
      \con{\eV\setminus\eVV}{P} & \text{if } S = \conCP\\
      S_1\setminus \eVV \paralS S_2\setminus \eVV & \text{if } S = S_1 \paral S_2\\
      \restB{c}{T\setminus (\eVV\setminus\sset{\permin{c},\permout{c}})} & \text{if } S = \rest{c}{T} 
    \end{cases}
  \end{align*}
\end{defi}

\begin{prop}\label{prop:rest-err}
  \begin{math}
    S\setminus\eVV \errNot \text{ implies } S \errNot
  \end{math}
\end{prop}

\begin{lem}\label{lem:restr-stable}
  \begin{math}
    S \setminus \mu \reducNot \text{ implies } \exists T. \ S \reduc^\ast T\reducNot \text{ where } (T\setminus\mu) \steq (S\setminus \mu)
  \end{math}
\end{lem}

\begin{proof}
  By Proposition~\ref{prop:safe-stab-vs-system-structure} we know,
  \begin{align}
    \label{eq:42}
S\setminus \mu &\steq\restB{\lst{d}}{ \;\paral_{i=0}^n \con{\eV_i}{\pioutA{c_i}{\lst{e_i}}} \; \paral_{j=0}^m \con{\eVV_j}{\piin{c'_j}{\lst{x_j}}{P_j}}}  &  \text{where }&\sset{c_1,\ldots, c_n} \cap \sset{c'_1,\ldots,c'_m} = \emptyset
  \end{align}
 By system structural equivalence, \steq,  the only sub-systems in $S$ that are abstracted away from $S\setminus \mu$  in $\restB{\lst{d}}{ \;\paral_{i=0}^n \con{\eV_i}{\pioutA{c_i}{\lst{e_i}}} \; \paral_{j=0}^m \con{\eVV_j}{\piin{c'_j}{\lst{x_j}}{P_j}}}$   are those of the form $\conC{\inert}$ where $\eV\subseteq \eVV$; the operation made these systems equivalent to  $\con{\emptyset}{\inert}$ which could then be eliminated through \rtit{scNil}.  In $S$, sub-systems of the form  $\conC{\inert}$ can still be eliminated through \rtit{cDsc} and then \rtit{scNil} (as before),  leaving us with the same array of mismatching confined output and input processes found in $S\setminus\eVV$, less the restricted permissions.
\end{proof}

\begin{lem} \label{lem:rest-reduc}
  \begin{math}
    S\setminus \eVV \reduc T \setminus \eVV \text{ implies } S \reduc T
  \end{math}
\end{lem}
\begin{proof}
  By rule induction on $S\setminus \eVV \reduc T \setminus \eVV$.
\end{proof}

\begin{lem}\label{lem:rest-evaluate}
  \begin{math}
    (S \setminus \eVV) \evaluate T \text{ implies } S \evaluate R \text{ where } R\setminus \eVV \steq T 
  \end{math}
\end{lem}
\begin{proof}
  Follows from Lemma~\ref{lem:rest-reduc}, Lemma~\ref{lem:restr-stable} and Proposition \ref{prop:rest-err}.
\end{proof}

\begin{lem}\label{lem:rest-satisfaction}
  \begin{math}
    \env, (S\setminus\eVV) \sat \fV \text{ implies } \env, S \sat \fV
  \end{math}
\end{lem}

\begin{proof}
  By induction on the structure of \fV using Lemma~\ref{lem:rest-evaluate}.
\end{proof}



\end{document}